\documentclass[onecolumn,prx,longbibliography,nofootinbib, superscriptaddress]{revtex4-2}

\usepackage[utf8]{inputenc}
\usepackage[dvips]{graphicx} 
\usepackage{subcaption} 
\usepackage{placeins}
\usepackage{amsfonts}
\usepackage{amssymb}
\usepackage{amscd}
\usepackage{amsmath, bm}    
\usepackage{amsthm}
\usepackage{mathtools}
\usepackage{appendix}
\usepackage{bbm}
\usepackage{dsfont}
\usepackage{enumerate}
\usepackage{enumitem}
\usepackage{epsfig}
\usepackage{float}
\usepackage{makecell}
\usepackage{subfloat}
\usepackage{xcolor}		
\usepackage{physics}
\usepackage{tabularx}
\usepackage{extarrows}
\usepackage{csquotes} 
\usepackage{multirow}
\usepackage{booktabs} 
\usepackage{accents}
\usepackage[colorlinks, linkcolor=red, anchorcolor=blue, citecolor=green]{hyperref}

\newlength{\dhatheight}

\usepackage{tikz}
\usepackage{tikzpeople}
\usepackage{pgfplots}
\usetikzlibrary{arrows}
\usetikzlibrary{shapes,fadings,snakes}
\usetikzlibrary{decorations.pathmorphing,patterns}
\usetikzlibrary{calc}
\usetikzlibrary{positioning}
\definecolor{amber}{rgb}{1.0, 0.75, 0.0}

\usepackage[most]{tcolorbox}

\newtcolorbox[auto counter]{mybox}[2][]{
	enhanced,
	breakable,
	colback=blue!5!white,
	colframe=blue!75!black,
	fonttitle=\bfseries,
	title=Box \thetcbcounter: #2,#1
}

\tcolorboxenvironment{deviceinner}{
	enhanced,
	breakable,
	colback=gray!10!white,
	colframe=gray!50!black,
	boxrule=0.5pt,
	left=5pt,
	right=5pt,
	top=5pt,
	bottom=5pt
}
\newif\ifuseboxprotocol
\useboxprotocoltrue  

\ifuseboxprotocol
\tcbset{
	protocolstyle/.style={
		enhanced,
		breakable,
		colback=blue!5!white,
		colframe=blue!75!black,
		fonttitle=\bfseries,
	}
}
\fi

\newcounter{protocol}

\ifuseboxprotocol
\newenvironment{protocol}[2]
{%
	\par\addvspace{\topsep}
	\refstepcounter{protocol} 
	\begin{tcolorbox}[
		protocolstyle, 
		title=Protocol \theprotocol: #1, 
		label=#2 
		]
	}
	{%
	\end{tcolorbox}
	\par\addvspace{\topsep}
}
\else
	\newenvironment{protocol}[2]
	{\par\addvspace{\topsep}
		\noindent
		\tabularx{\linewidth}{@{} X @{}}
		\vspace{-2mm}
		\refstepcounter{protocol}\textbf{Protocol \theprotocol} #1 \label{#2} \\
	}
	{ \\
		\endtabularx
		\par\addvspace{\topsep}}
	\fi
	
	\makeatletter
	\newcommand{\printappendixtoc}{%
		\section*{Contents}    
		\@starttoc{atoc}       
	}
	
	\newcommand{\l@atocsection}{\@dottedtocline{1}{0em}{2.3em}}
	\newcommand{\l@atocsubsection}{\@dottedtocline{2}{1.5em}{3em}}
	\makeatother

	\newtheorem{theorem}{Theorem}
	
	\newtheorem{fact}{Fact}
	\newtheorem{lemma}{Lemma}
	\newtheorem{corollary}{Corollary}

	\newtheorem{definition}{Definition}

	\newtheorem{remark}{Remark}

	\graphicspath{{./figure/}}
	
	\renewcommand{\vec}[1]{\mathbf{#1}}

	\newcommand{\Pmsm}{\Pr_{\operatorname{MS}}}
	\newcommand{\Dl}{\vartriangleleft}
	\newcommand{\Dr}{\vartriangleright}

	\newcommand{\bE}{\mathbb{E}}

	\newcommand{\poly}{\text{poly}}

	\renewcommand{\norm}[1]{\left\lVert#1\right\rVert}
	\newcommand{\ab}{^\flat}
	\newcommand{\abb}{^{{\flat}{\flat}}}
	\newcommand{\vI}{v_\mathrm{I}}
	\newcommand{\vII}{v_\mathrm{II}}
	\newcommand{\vIII}{v_\mathrm{III}}
	\newcommand{\hvI}{\hat{v}_\mathrm{I}}
	\newcommand{\hvII}{\hat{v}_\mathrm{II}}
	\newcommand{\hvIII}{\hat{v}_\mathrm{III}}

	\newcommand{\opnorm}[1]{\left\lVert #1 \right\rVert_{\infty}} 
	\newcommand{\cO}{\mathcal{O}}

\begin{document}
		
		\title{Scalable self-testing of generic multipartite quantum states}
		
		\author{Jinchang Liu}
		\thanks{These authors contributed equally to this work.}
		\affiliation{Institute for Interdisciplinary Information Sciences, Tsinghua University, Beijing 100084, China}
		
		\author{Elias X.\ Huber}
		\thanks{These authors contributed equally to this work.}
		\affiliation{Fraunhofer Singapore, Nanyang Technological University, Singapore}
		\affiliation{Centre for Quantum Technologies, National University of Singapore, Singapore}
		
		\author{Zhenyu Du}
		\thanks{These authors contributed equally to this work.}
		\affiliation{Center for Quantum Information, Institute for Interdisciplinary Information Sciences, Tsinghua University, Beijing 100084, China}
		
		\author{Xingjian Zhang}
		\thanks{Present address: Centre for Quantum Software and Information, University of Technology Sydney, Sydney, Australia}
		\affiliation{QICI Quantum Information and Computation Initiative, School of Computing and Data Science, The University of Hong Kong, Hong Kong SAR, China}
		
		\author{Xiongfeng Ma}
		\email{xma@tsinghua.edu.cn}
		\affiliation{Center for Quantum Information, Institute for Interdisciplinary Information Sciences, Tsinghua University, Beijing 100084, China}

		\begin{abstract}
			Characterizing large quantum systems with minimal assumptions is a central challenge in quantum information science. 
			Self-testing provides the strongest form of certification by identifying the underlying quantum state solely from observed measurement statistics.
			However, existing self-testing methods for generic $n$-partite states face a scalability barrier, requiring exponentially many samples in the system size.
			In this work, we overcome this barrier by introducing a protocol that robustly self-tests almost all $n$-qubit states with only polynomial sample complexity. 
			The key ingredient is an efficient scheme for device-independently evaluating multipartite Pauli measurements, which can be implemented using only a linear number of ancillary Bell pairs together with standard projective and Bell measurements, well within the reach of current quantum technology.
			Beyond self-testing states, our scheme provides a general framework for implementing a wide range of learning and certification protocols in the device-independent setting, thereby opening a scalable route to device-independent quantum information processing in large-scale quantum networks. 
		\end{abstract}
		
		\maketitle
		
		\section{Introduction}
		
		Certifying the correct functioning of quantum devices lies at the heart of quantum technologies~\cite{Eisert2020CertificationReview, Kliesch2021CertificationSurvey, Carraso2021TheoreticalExperimental}. 
		As quantum platforms and quantum networks continue to scale up~\cite{Liu2024MetropolitanNetwork, Main2025DistributedQuantumComputing, Liu2026LongLivedEntanglement, Stas2026EntanglementAssisted}, the common assumption that certification protocols can rely on trusted and well-calibrated devices becomes increasingly untenable. 
		Quantum nonlocality~\cite{Bell1964EPRParadox} makes it possible to remove this assumption, giving rise to the device-independent (DI) paradigm, the only rigorous framework for certifying quantum systems \emph{without making any assumptions} about the inner workings of the experimental devices~\cite{Mayers1998SelfTesting}.
		This framework has enabled transformative applications of quantum technology, such as secure key distribution~\cite{Ekert1991QKDBell, Mayers1998SelfTesting, Acin2007DeviceIndependentSecurity}, certified random number generation~\cite{Pironio2010RandomNumberBell, Liu2018DIQRNG, Farkas2026MaximalRandomness}, and delegated quantum computing~\cite{Reichardt2013ClassicalCommand}. 
		Developing powerful DI tools for characterizing quantum systems has therefore become a central goal of quantum information science.
		
		Within the DI paradigm, the strongest certification of an underlying quantum system is known as \emph{self-testing}~\cite{Mayers1998SelfTesting}.
		In a self-testing protocol, certain observed statistics can uniquely determine the underlying state and measurements, up to the intrinsic degrees of freedom that are indistinguishable from observations.
		A canonical example is that maximal violation of the Clauser–Horne–Shimony–Holt (CHSH) inequality self-tests the two-qubit Bell state and the associated measurements~\cite{Clauser1969CHSH, Mayers1998SelfTesting}.
		
		A central goal in the certification of large-scale quantum devices and networks is therefore to develop \emph{scalable} self-testing protocols for the underlying \emph{multipartite} quantum systems.
		While seminal works have established that all multipartite entangled pure states can, in principle, be self-tested~\cite{Coladangelo2017PureBipartite,Supic2023NetworkSelftest, Balanzo2024PureMultipartiteEntangled}, the challenge of scalability remains overlooked and largely unaddressed. 
		Specifically, rigorous analyses of robustness against experimental noise are currently restricted to a few highly structured examples, such as GHZ and W states~\cite{Wu2014RobustWState, MaKague2014GraphState, Fadel2017DickeStates, Baccari2020ScalableQubitGraph, Wu2021RobustMultiparticle}. 
		Furthermore, most existing results rely on perfect probability distributions, which are inaccessible via finite sampling.
		The few studies that do account for the resulting statistical deviations~\cite{2022sampleefficientdiverification,bancal2021self} typically demand a sample complexity that scales exponentially with the number of parties, a feature that is highly unfavorable in practice. 
		Together, these limitations present a significant barrier to the deployment of self-testing protocols in realistic applications.
		
		Overcoming this exponential sample-complexity barrier is challenging for several reasons: 
		(i) The black-box treatment of measurement devices makes it difficult to use them to probe the underlying system. 
		(ii) Without perfect distribution reconstruction, reliably extracting global information from finite samples is highly challenging.
		(iii) Even in the device-dependent setting, multipartite entangled states exhibit an exceptionally rich structure that is notoriously difficult to characterize~\cite{Eisert2006MultiparticleEntanglement, Haah2016OptimalTomo, Liu2022Fundamental}. These challenges motivate the following central question:
		\begin{center}
			\emph{Can one robustly self-test generic multipartite quantum states with sample complexity polynomial in the number of parties?}
		\end{center}
		
		In this work, we answer this question affirmatively by introducing a robust self-testing protocol for almost all multipartite qubit states with sample complexity polynomial in the number of parties.
		We overcome challenge (i) by solving an open question: how to efficiently certify, in a \emph{global} sense, multipartite Pauli measurements performed by spatially separated parties device-independently.
		In particular, although single-party Pauli measurements can be self-tested~\cite{Bowles2018DIAllEntangled}, extending such certification to the multipartite setting encounters a local-transpose obstacle~\cite{Supic2023NetworkSelftest}.
		Roughly speaking, standard self-testing techniques fix each party’s Pauli measurements only up to an independent local-transpose freedom, which introduces an ambiguous sign in the $Y$ operator (i.e., $\pm Y$). 
		To certify a target multipartite Pauli measurement, these local freedoms must be aligned consistently across all parties.
		To overcome this obstacle, we introduce a new \emph{transpose-braiding} test that enforces this global consistency. 
		The test is based on inequalities between neighboring parties, constructed from a two-qubit observable $K$ (detailed in the Results section). The key insight is that the maximal eigenvalue of $K$ is strictly larger than that of its partial transpose $K^{T_1}$, which forces the transpose freedoms of neighboring parties to be consistent. 
		By sequentially applying this alignment across the network, one certifies multipartite Pauli measurements up to a global transpose freedom. 
		
		Building on this test, our protocol robustly lifts device-dependent Pauli measurements to their device-independent counterparts using only polynomial overhead. 
		Crucially, this protocol enables sample-efficient estimation of the observables relevant for certifying the target states, without reconstructing the full correlation distribution, thereby overcoming challenge (ii). Furthermore, we remove the assumption that the underlying quantum states are independent and identically distributed (i.i.d.) across rounds by providing a fully non-i.i.d.\ analysis.
		
		Finally, we address challenge (iii) by building on recent device-dependent certification protocols for generic states~\cite{Huang2024Certifying, du2025certifyinglocalizablequantumproperties}. 
		The protocol of Ref.~\cite{Huang2024Certifying} relies only on local Pauli measurements and can therefore be lifted to the DI setting using our multipartite Pauli-measurement scheme, already yielding a self-testing protocol with polynomial sample complexity for generic states. 
		To remove its remaining shared-randomness requirement, we further build on the random-basis-enhanced variant of Ref.~\cite{du2025certifyinglocalizablequantumproperties} and show how the required shared randomness can be replaced by standard local randomness. 
		Taken together, our protocol only requires the minimal assumptions of the DI framework.
		
		For experimental implementation, our protocol requires only standard Bell states and Pauli measurements, both of which are well within the reach of current quantum platforms~\cite{Main2025DistributedQuantumComputing, Liu2026LongLivedEntanglement, Stas2026EntanglementAssisted}. 
		By providing a systematic bridge from powerful device-dependent methods to their DI counterparts, our protocol unlocks a wide range of scalable DI learning and certification protocols for large-scale quantum devices and interconnected quantum networks~\cite{Flammia2011DirectFidelityEstimation, Eisert2020CertificationReview, huang_predicting_2020, elben_randomized_2022, Huang2024Certifying, du2025certifyinglocalizablequantumproperties}.
		
		\section{Setup}
		We consider self-testing properties of quantum states shared by $n$ main parties $A=A_0A_1\cdots A_{n-1}$. 
		The standard notion of self-testing a single state is incorporated as a special case.
		The property we aim to self-test is an expectation value of the form $\tr(L \rho_A)$, where $L$ is a target observable. By choosing $L$ appropriately, this framework captures a wide range of certification tasks: it can witness multipartite entanglement, probe circuit complexity, or even single out a unique target state~\cite{Guhne2009EntanglementDetection, Flammia2011DirectFidelityEstimation, huang_predicting_2020, elben_randomized_2022, Huang2024Certifying, du2025certifyinglocalizablequantumproperties}.
		In this work, we focus on observables $L$ induced by randomized multipartite Pauli schemes, which act on an $n$-qubit state as follows.
		First, a random Pauli basis string $\vec{P}=P_0P_1\cdots P_{n-1}\in\{X,Y,Z\}^{\otimes n}$ is drawn from a distribution $\mathcal{D}$.
		Then, for each $0\le i<n$, the $i$th qubit is measured in the basis $P_i$, yielding an outcome string $\vec{b}=b_0b_1\cdots b_{n-1}\in\{0,1\}^n$.
		Finally, one assigns a weight $\omega(\vec{b}|\vec{P})$ depending on both the outcomes $\vec{b}$ and the chosen bases $\vec{P}$.
		Averaging this weighting coefficient over the randomness in $\vec{P}$ and the measurement outcomes is equivalent to estimating $\tr(L\rho_A)$, where
		\begin{equation}\label{eq:L-def}
			L := \mathbb{E}_{\vec{P}\leftarrow \mathcal{D}}\left[\sum_{\vec{b}\in \{0,1\}^n}\omega(\vec{b}|\vec{P}) \bigotimes_{l=0}^{n-1} \frac{\mathbb{I}+ (-1)^{b_l} P_l}{2}\right].
		\end{equation}
		A single round of Pauli measurements is illustrated in Fig.~\ref{fig:lifting}.
		
		In the self-testing protocol, as shown in Fig.~\ref{fig:lifting}, we introduce $n$ auxiliary parties $B = B_0B_1\cdots B_{n-1}$ that assist the main parties in characterizing their shared state, following the network-assisted architecture introduced in previous works~\cite{Renou2018SelfTestingQuantumNetwork, Bowles2018DIAllEntangled, Supic2020selftestingof, Supic2023NetworkSelftest}. 
		Our setting requires minimal, device-independent assumptions.
		Concretely, each party’s device is treated as a black box, and the only trusted ingredient is the randomness of the classical inputs, provided by a referee or generated locally (i.e., a party’s “free will”).
		Spatial separation between parties imposes a tensor-product structure on their measurements.  
		Formally, each main party $A_l$ receives an input $x_l$ and performs an uncharacterized positive operator-valued measure (POVM) $\{M_{a_l|x_l}^{(l)}\}_{a_l}$, producing an outcome $a_l$. 
		Similarly, each auxiliary party $B_l$ receives an input $y_l$, performs an uncharacterized POVM $\{N_{b_l|y_l}^{(l)}\}_{b_l}$, and outputs $b_l$. 
		The observed statistics are then given by Born's rule:
		\begin{equation}\label{eq:self-testing-statistics-main-text}
			\Pr[a_0\ldots a_{n-1}b_0\ldots b_{n-1}|x_0 \ldots x_{n-1}y_0 \ldots y_{n-1}] = \tr[\left(\bigotimes_{l=0}^{n-1}M_{a_l|x_l}^{(l)}\otimes\bigotimes_{l=0}^{n-1}N_{b_l|y_l}^{(l)}\right)\rho],
		\end{equation}
		where $\rho$ is the unknown state shared across all $2n$ parties. 
		In the main text, we analyze the i.i.d.\ setting, where each experimental round is modeled as an independent use of the same state and the same measurement operators. The non-i.i.d.\ scenario is addressed in Appendix~\ref{sec:non-iid}.
		We emphasize that no additional assumptions are made on the inner working of quantum devices.
		
		\begin{figure*}
			\centering
			\includegraphics[width=0.9\linewidth]{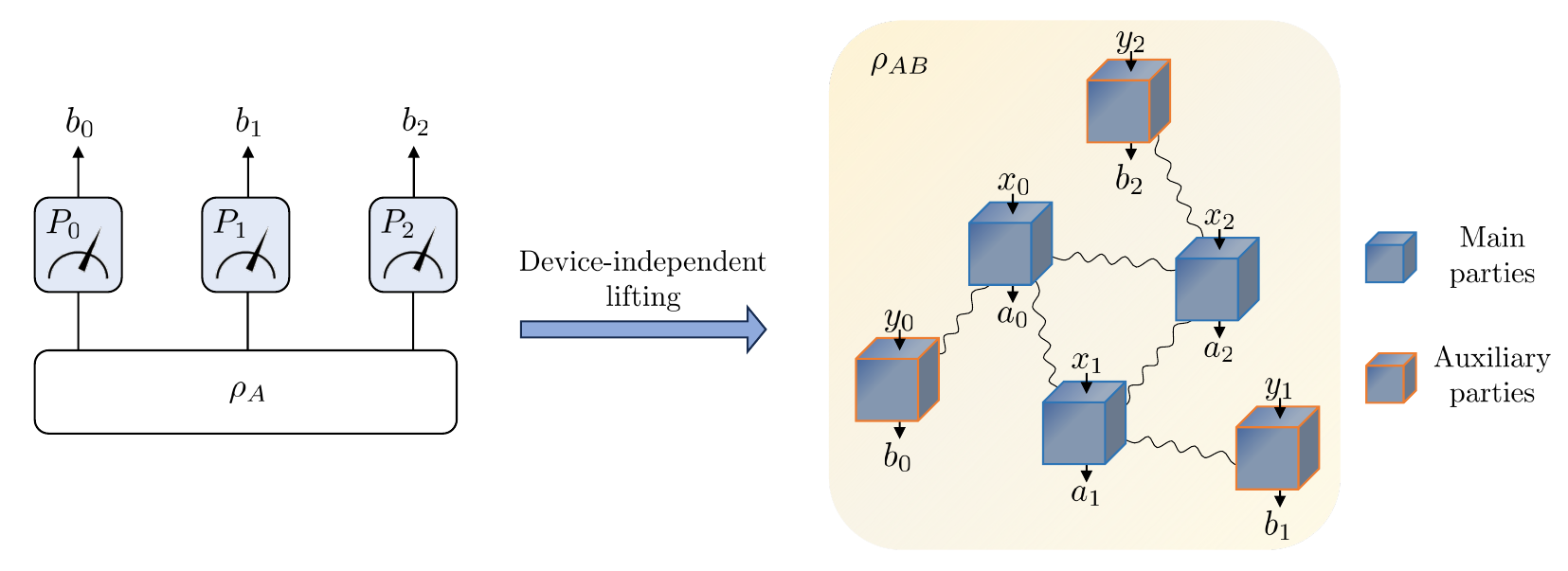}
			\caption{
				Device-independent lifting of a randomized multipartite Pauli measurement protocol, illustrated for three qubits.
				\textbf{Left:} In each round, the parties measure $\rho_A$ in a chosen Pauli basis $P_0P_1\cdots P_{n-1}$ and obtain outcomes $b_0b_1\cdots b_{n-1}$.
				\textbf{Right:} Our protocol lifts this device-dependent Pauli-measurement protocol into its DI version by introducing $n$ auxiliary parties.
				In the DI implementation, both the underlying states and the measurement devices are treated as uncharacterized black boxes with only classical inputs and outputs.
			}
			\label{fig:lifting}
		\end{figure*}
		
		Our self-testing protocol proceeds by repeating the experiment for multiple rounds. 
		In each round, the parties’ outcomes are sampled according to Eq.\,\eqref{eq:self-testing-statistics-main-text}. 
		From the collected data, we construct an estimator $\hat{\omega}$ and aim to certify that $\tr(L\,\rho_A)=\hat{\omega}$ up to a precision $\varepsilon$, achieving self-testing of the desired property and the multipartite Pauli measurements.
		Importantly, self-testing necessarily leaves some intrinsic degrees of freedom undetermined. 
		First, there is the local extraction freedom: if parties share an initial state $\rho'$ and can transform it into $\rho$ via local channels, their measurement devices can produce statistics identical to those of $\rho$. 
		Consequently, self-testing results are always defined up to a local extraction channel. Common examples of such operations include local basis rotations and tracing out ``junk'' auxiliary systems upon which the measurements act trivially.
		Second, the statistics are invariant under global complex conjugation, since simultaneously replacing the state and all measurement operators with their complex conjugates does not affect the real-valued probabilities. 
		Beyond these unavoidable freedoms, any practical notion of self-testing must be robust to statistical fluctuations arising from finite samples.
		Taking these into account, we say that the parties self-test a state or a property of the state if the desired target can be approximately extracted from $\rho$ by some local channels $\{\Gamma_l\}_l$ acting on the main parties, up to the unavoidable global-conjugation freedom~\cite{Supic2020selftestingof}.
		
		\begin{definition}[$\varepsilon$-approximate self-testing]\label{def:approx-selftest-property}
			We say that the main parties $\varepsilon$-approximately share a state $\tau$  if there exists a local extraction channel $\Gamma=\bigotimes_{i=0}^{n-1}\Gamma_i$ acting on the main parties $A=A_0\cdots A_{n-1}$, such that
			\begin{equation}\label{eq:self-testing-extraction}
				D\left(\Gamma(\rho_A), p\tau_+\otimes\ketbra{0}^{\otimes n}+(1-p)\tau_-^*\otimes\ketbra{1}^{\otimes n}\right) \le \varepsilon, 
			\end{equation} 
			for some (possibly unknown) states $\tau_+, \tau_-$ and $p\in [0,1]$ such that $\tau=p\tau_++(1-p)\tau_-$. Here, $D(\cdot,\cdot)$ denotes trace distance. 
			
			We say that the main parties $\varepsilon$-approximately self-test the expectation value $\hat{\omega}$ for an observable $L$ if there exists a state $\tau$ such that they $\varepsilon$-approximately share $\tau$ and $|\tr(L\,\tau)-\hat{\omega}|\le \varepsilon$.
			We say that the main parties $\varepsilon$-approximately self-test a pure state $\Psi$ if they $\varepsilon$-approximately share $\Psi$.
		\end{definition}
		Here, the ancilla $\ketbra{0}^{\otimes n}$ or $\ketbra{1}^{\otimes n}$ are shared across each party, indicating whether the states and measurements are conjugated globally.
		Crucially, the flag is either all zeros or all ones across the \(n\) parties, and therefore represents global conjugation freedom.
		Note that for pure state $\tau = \Psi$, we have $\tau_+=\tau_- = \Psi$, recovering the normal notion of pure-state self-testing.
		
		In general, self-testing an expectation value for an observable $L$ is a weaker task than self-testing a state, as it only certifies a single property of the extracted state. By contrast, self-testing a pure state $\Psi$ requires certifying the entire underlying multipartite state. These two notions can be connected by suitably choosing a witness observable $L$: if $L$ has $\Psi$ as an eigenstate corresponding to its unique maximal eigenvalue, and has a sufficient spectral gap, then any state achieving a nearly maximal expectation value for $L$ must be close to $\Psi$. 
		Our later results leverage precisely this connection.
		We emphasize that Definition~\ref{def:approx-selftest-property} is operationally optimal, as neither the freedom of the local extraction channel nor the intrinsic global conjugation can be further removed.

		\section{Results}
		We now present our main results on self-testing expectation values from multipartite Pauli measurements, with a particularly important application to the self-testing of generic multipartite states.
		
		\subsection{Efficient self-testing of multipartite Pauli measurements}
		We construct a robust and efficient protocol for self-testing the expectation value of an observable $L$ defined by a randomized Pauli scheme. A high-level overview of the protocol is shown in Fig.~\ref{fig:protocol-high-level} and Table~\ref{tab:main-experiment_setup}. We emphasize that these describe the ideal implementation. In this implementation, the target state is shared among the main parties on the systems $T_l$. Each main party $A_l$ contains the four qubits $T_l$, $S_l$, $R_{l}^{\Dr}$ and $R_{l}^{\Dl}$. The protocol uses $n$ Bell pairs shared between each main--auxiliary pair $(A_l,B_l)$, specifically between $S_l$ and $B_l$, as well as $n-1$ additional Bell pairs shared between neighboring main parties $(A_l,A_{l+1})$, specifically between $R_{l}^{\Dr}$ and $R_{l+1}^{\Dl}$. This ideal implementation is illustrated in Fig.~\ref{fig:protocol-high-level}a.
		The Bell pairs shared between neighboring main parties $(A_l,A_{l+1})$ could be replaced by Bell pairs shared between neighboring auxiliary parties $(B_l,B_{l+1})$. The analysis in that case is similar but somewhat more involved, so we focus on the former setting for simplicity. In the actual experiment, we make no assumptions about the internal workings of the devices beyond the classical input and output alphabets, the randomness of the classical measurement inputs, and the spatial separation of the parties $A_l$ and $B_l$.
		
		The input alphabets for the main parties $A_l$ and the auxiliary parties $B_l$ are
		$x_l \in \{0, \dots, 5, \vartriangleleft, \vartriangleright, \diamond\}$ and $y_l \in \{0,1,2\}$, respectively. The corresponding outputs are $a_l \in \{0,1\}$ for $x_l \in \{0,\dots,5\}$, $\vec{a}_l \in \{00,01,10,11\}$ for $x_l \in \{\Dl,\Dr,\diamond\}$, and $b_l \in \{0,1\}$ for the auxiliary parties. We focus on three types of input configurations: (i) CHSH tests, (ii) transpose-braiding tests, and (iii) target Pauli measurements implemented via teleportation. These ingredients are shown separately in Fig.~\ref{fig:protocol-high-level}b--d, and are described quantitatively by the following three families of estimators:
		\begin{equation}
			\begin{split}
				v_\mathrm{I}(l,k) &= 2\sqrt{2}-\sum_{i,j\in\{0,1\}}(-1)^{ij}\bE[(-1)^{a_l+b_l}\mid x_l=x_l(i,k),y_l=y_l(j,k)];\\
				v_{\mathrm{II}}(l,l+1) &= \mathbb{E}_{P\gets\{X,Y,Z\}}\left\{1-(-1)^{\mathbf{1}[P=Y]}\bE[(-1)^{f(b_l|P,\vec{a}_l)}(-1)^{f(b_{l+1}|P,\vec{a}_{l+1})}\mid x_l={\Dr}, x_{l+1}={\Dl}, y_l=y_{l+1}=P]\right\};\\
				v_{\mathrm{III}} &= \mathbb{E}_{\vec{P}\gets\mathcal{D}}[\omega(f(\vec{b}|\vec{P},\{\vec{a}_l\}_{l=0}^{n-1})|\vec{P})\mid \forall l, x_l={\diamond},y_l=P_l];
			\end{split}
		\end{equation}
		where we identify the inputs $y_l=0,1,2$ for $B_l$ with the Pauli operators $X,Y,Z$, respectively,
		\begin{equation}
			\begin{split}
				\text{CHSH settings:}&\quad (x_l(0,k),x_l(1,k),y_l(0,k),y_l(1,k))=\begin{cases}
					(0, 1, 0, 2) & \text{if } k=0, \\
					(3, 2, 0, 1) & \text{if } k=1, \\
					(5, 4, 2, 1) & \text{if } k=2;
				\end{cases}\\
				\text{Teleportation correction:}&\quad  f(b|P,\vec{a}) = b \oplus \mathbf{1}[\{P, U_{\vec{a}}\} = 0],\quad \text{where } U_{\vec{a}} = X^{a_0}Z^{a_1} \text{ for } \vec{a}=a_0a_1,
			\end{split}
		\end{equation}
		and $f(\vec{b}|\vec{P}, \{\vec{a}_l\})$ applies $f$ to each vector component, that is, $f(\vec{b}|\vec{P}, \{\vec{a}_l\})=(f(b_l|P_l,\vec{a}_l))_{l=0}^{n-1}$.
		During the protocol, we obtain unbiased estimators of $\vI$, $\vII$, and $\vIII$. We go on to show that $\vI,\vII\approx0$ certify that $\vIII$ indeed estimates the observable $L$ on a state shared by the main parties. We describe the role of each ingredient below, with the complete protocols and full analysis deferred to Appendix~\ref{sec:sample_comp} (see Protocols~\ref{prot:di-witness-local-randomness} and \ref{prot:di-witness-shared-randomness} for the versions with local and shared randomness, respectively).
		\begin{figure*}
			\centering
			\includegraphics[width=0.9\linewidth]{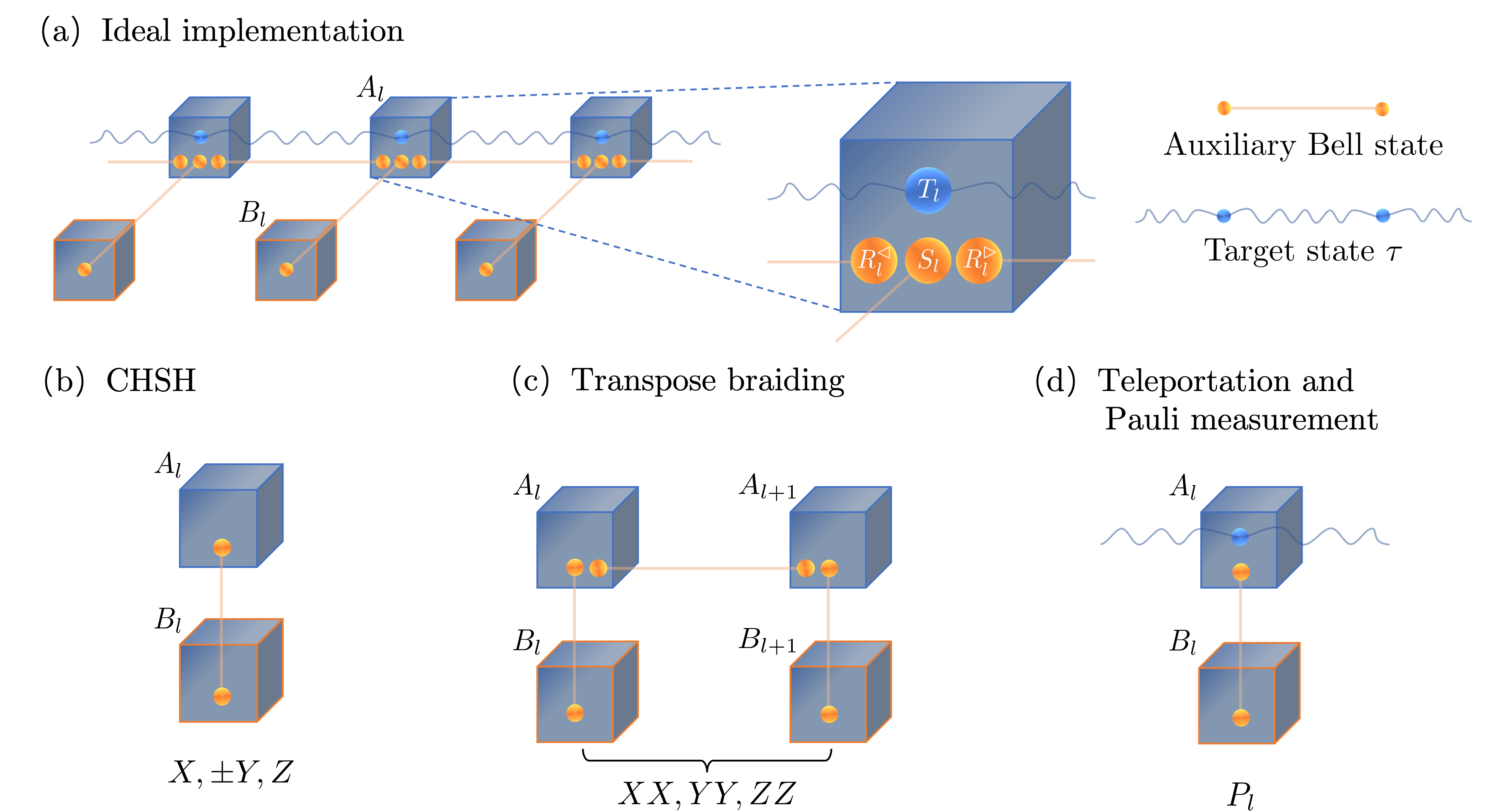}
			\caption{
				Self-testing protocol for implementing randomized multipartite Pauli measurement schemes.
				(a) The protocol introduces $n$ auxiliary parties and $2n-1$ additional Bell pairs, with $n$ pairs between $(A_l,B_l)$ and $n-1$ pairs between neighboring main parties. Together, these resources enable DI lifting of randomized multipartite Pauli measurements. 
				(b) CHSH tests certify the auxiliary Bell pairs and the relevant Pauli observables up to a local complex-conjugation freedom, which manifests as an undetermined sign flip for $Y$ measurements.
				(c) The transpose-braiding test aligns these local freedoms by enforcing consistent conjugation across adjacent auxiliary parties. Neighboring parties are either both conjugated or not conjugated at all, so that adjacent parties jointly implement $YY$ measurements.
				(d) To obtain a self-tested expectation value, the main parties teleport their shared state to the auxiliary parties, who then perform the self-tested Pauli measurements in basis $\vec{P} = P_0P_1\cdots P_{n-1}$.
			}
			\label{fig:protocol-high-level}
		\end{figure*}
		
		\begin{table}
			\centering
			\begin{tabular}{c|c|c}
				\hline
				Party & Input & Measurement in the ideal implementation\\
				\hline
				\multirow{6}{*}{$A_l$} & $x_l\in \{0,1\}$  & $\frac1{\sqrt2}(X+(-1)^{x_l}Z)$ on ${S_l}$\\ 
				& $x_l\in \{2,3\}$  & $\frac1{\sqrt2}(X+(-1)^{x_l} Y)$ on ${S_l}$ \\
				& $x_l\in \{4,5\}$  & $\frac1{\sqrt2}(Z+(-1)^{x_l} Y)$ on ${S_l}$ \\
				& $x_l =\ \Dr$ & BSM on $R_l^{\Dr}S_{l}$ \\
				& $x_l =\ \Dl$ & BSM on $R_l^{\Dl}S_{l}$ \\
				& $x_l = \diamond$ & BSM on $T_lS_l$ \\
				\hline
				\multirow{3}{*}{$B_l$} & $y_l = 0$ & $X$ \\
				& $y_l = 1$ & $Y$ \\
				& $y_l = 2$ & $Z$ \\
				\hline
			\end{tabular}
			\caption{Experimental setup. For each party index $l=0,\ldots,n-1$, the table specifies the measurements in the ideal implementations.}
			\label{tab:main-experiment_setup}
		\end{table}
		
		\emph{(i) CHSH test.} 
		The CHSH tests in Fig.~\ref{fig:protocol-high-level}b, corresponding to the estimator $\vI$ with input settings $x_l\in\{0,1\}, y_l\in\{0,2\}$, $x_l\in\{2,3\}, y_l\in\{0,1\}$, and $x_l\in\{4,5\}, y_l\in\{1,2\}$, are used to robustly self-test both the presence of a Bell pair shared between $S_l$ and $B_l$ and the anti-commutation structure of the measurements. 
		Here, $2\sqrt2-\vI(l,k)$ is the value of a CHSH test where the input values for $A_l$, $B_l$ are $(x_l(0,k),x_l(1,k))$ and $(y_l(0,k),y_l(1,k))$, respectively.
		When $\vI(l,k)\approx 0$ for all $l,k$, that is, when all observed CHSH values are close to their extremum, the auxiliary measurements are certified to implement observables that are approximately $X$ and $Z$, together with a third observable that is $Y$ up to a local complex-conjugation freedom~\cite{Bowles2018DIAllEntangled}.
		
		However, self-testing Pauli observables only up to local-transpose freedom is not sufficient for estimating a multipartite Pauli observable. In particular, while one can self-test the observables $X$, $Z$, and $\pm Y$, the sign of $Y$ cannot be fixed locally. As a simple example, consider estimating $L=Y_1\otimes Y_2$ on some state. If party~1 effectively measures $-Y$ while party~2 measures $Y$, then the implemented observable becomes $L^{T_1}=(-Y_1)\otimes Y_2=-L$. If this mismatch were ignored, the protocol would, in general, estimate a partially transposed version of $L$ across some subsystems.
		
		\emph{(ii) Transpose braiding.} The transpose-braiding test in Fig.~\ref{fig:protocol-high-level}c (corresponding to estimator $\vII$ and input settings $x_l={\vartriangleright}$, $x_{l+1}={\vartriangleleft}$, $y_l=y_{l+1}$) resolves this mismatch and enforces that the auxiliary parties perform consistent $YY$ measurements.
		The test checks that the expectation value for the two-qubit observable $K = X_1\otimes X_2 - Y_1 \otimes Y_2 + Z_1 \otimes Z_2$ is close to $3$.
		In the ideal implementation, this observable is measured on the Bell pair shared between $R_l^{\Dr}$ and $R_{l+1}^{\Dl}$, which is teleported to $B_l$ and $B_{l+1}$ via Bell-state measurements on $R_l^{\Dr}S_l$ and $R_{l+1}^{\Dl}S_{l+1}$.
		This Bell state is the unique eigenstate of $K$ with eigenvalue $3$.
		By contrast, as discussed above, applying a partial transpose on either subsystem flips the sign of the $Y\otimes Y$ term. This case corresponds to
		$K^{T_1}=K^{T_2}=X_1\otimes X_2 + Y_1\otimes Y_2 + Z_1\otimes Z_2$,
		whose maximal eigenvalue is only $1$. This gap separation robustly rules out misaligned partial transposes.
		Consequently, only the consistent measurements $\{X_1,Y_1,Z_1\}$ and $\{X_2,Y_2,Z_2\}$, or $\{X_1,-Y_1,Z_1\}$ and $\{X_2,-Y_2,Z_2\}$ can produce the expected correlations. 
		Chaining these constraints along the line of neighboring parties eliminates the local-transpose freedom across all parties.
		
		\emph{(iii) Teleportation and measurement.} To obtain an estimate of the target observable $L$, we introduce a measurement setting (corresponding to estimator $\vIII$ and input setting $x_l=\diamond$ for all $l$) in which each main party teleports its share of the unknown state to the corresponding auxiliary party, using the Bell pair shared between $S_l$ and $B_l$ that is also used in the CHSH tests,  as depicted in Fig.~\ref{fig:protocol-high-level}d.
		The auxiliary parties then perform the Pauli measurements certified by the CHSH tests. Although standard teleportation requires a correction unitary $U_{\vec a}$ conditioned on the BSM outcome $\vec a$, for Pauli measurements this correction can be incorporated entirely into classical post-processing: the measurement outcome is flipped precisely when $U_{\vec a}$ anti-commutes with the measured Pauli observable. That is, $f(b_l|P_l,\vec{a}_l)$ gives the corrected Pauli measurement outcome. Thus, $\vIII$ is an estimate of the observable $L$.
		
		We now outline the construction of the extraction channel $\Gamma$ in Definition~\ref{def:approx-selftest-property}. First, satisfying all the CHSH tests (see (i)) allows us to extract a qubit system $B'_l$ from the system $B_l$ for each $l$, such that measurements on $B_l$ are equivalent to Pauli measurements on $ B'_l$~\cite{Bowles2018DIAllEntangled}. One could then choose the input $x_l=\diamond$ for all main parties, obtain the BSM outcomes $\vec{a}_l$, transform from $B_l$ to $B'_l$, and apply the corresponding correction unitary $U_{\vec a_l}$ to each system $B'_l$. This approximately extracts a state $\Gamma(\rho)$ which satisfies Eq.~\eqref{eq:self-testing-extraction} on the $ B'_l$ systems.
		However, a technical difficulty arises: extracting the state at the auxiliary parties depends on the BSM outcomes at the main parties, which, in turn, requires classical communication between the parties.
		Self-testing, however, requires that the target state must be produced by a fully local extraction channel, as in Definition~\ref{def:approx-selftest-property}, and thus no communication between parties is allowed. Prior network-assisted self-testing protocols~\cite{Supic2023NetworkSelftest} do rely on such additional communication. 
		In contrast, we overcome this difficulty by explicitly constructing a fully local extraction channel. The key insight is that the 3-CHSH inequalities in Fig.~\ref{fig:protocol-overview}b not only certify local Pauli measurements, but also imply that the shared main--auxiliary state is close to a Bell pair. This allows the extraction channel to replace the physical main--auxiliary Bell pair with a locally prepared Bell pair on the main side, thereby reproducing the effect of the teleportation entirely locally, extracting the desired target state without any communication. The detailed construction of this local extraction channel is given in Appendix~\ref{subsec:refinement-extraction-channel}.

		In summary, those components enable an efficient DI lifting of multipartite Pauli measurements:
		\begin{theorem}[Efficient self-testing of multipartite Pauli measurements]\label{thm:efficient-selftest-pauli-msm-informal-main-text}
			Consider an $n$-qubit observable $L$ specified by a bounded weighting function $\omega(\cdot\mid\cdot)$ and a Pauli-basis distribution $\mathcal{D}$. Then the self-testing protocol depicted in Fig.~\ref{fig:protocol-high-level} that uses $n$ auxiliary parties and $2n-1$ auxiliary Bell states satisfies
			with probability at least $1-\delta$:
			\begin{enumerate}
				\item Soundness: If the protocol outputs \texttt{CERTIFIED}, then the main parties $\varepsilon$-approximately self-test the reported estimate $\hat{\omega}$ of the observable $L$.
				\item Completeness: If the underlying experiment is sufficiently close to the ideal implementation (with inverse-polynomial noise tolerance in state preparation and measurements), then the protocol outputs \texttt{CERTIFIED}.
			\end{enumerate}
			The protocol requires
			$N=\cO\!\left(\mathrm{poly}(n,\varepsilon^{-1})\log(\delta^{-1})\right)$
			rounds.
			These guarantees hold (i) for product basis distributions $\mathcal{D}=\bigotimes_{l=0}^{n-1}\mathcal{D}_l$, where each marginal distribution $\mathcal{D}_l$ satisfies $p_X^{(l)}, p_Y^{(l)}, p_Z^{(l)} = \Omega(1)$, using only local randomness, and (ii) for arbitrary $\mathcal{D}$ when the parties are allowed shared randomness.
		\end{theorem}
		
		The formal version of Theorem~\ref{thm:efficient-selftest-pauli-msm-informal-main-text} is given by Theorems~\ref{thm:fin-sample-DI-estimation-L-local-randomness},~\ref{thm:fin-sample-DI-estimation-L-shared-randomness}, and~\ref{thm:completeness} in the Appendix. 
		Here, the soundness statement guarantees the correctness of the reported value, while completeness ensures robustness to experimental imperfections. 
		For the case when the protocol uses only local randomness, the requirement $p_X^{(l)}, p_Y^{(l)}, p_Z^{(l)} = \Omega(1)$ allows the auxiliary parties to sample measurement settings according to the distribution $\mathcal{D}$ while simultaneously ensuring that each setting is sampled sufficiently often and therefore well self-tested.  
		While local randomness is a standard assumption in device-independent protocols, the shared-randomness requirement can, in principle, be achieved by first establishing a shared random string via DI quantum key distribution or conference key agreement. Incorporating this additional step and its associated security analysis is beyond our current scope and is left for future work.
		
		We emphasize that when the classical input to each party is trusted, for example by being supplied by an additional referee, and spatial separation is guaranteed, our self-testing protocol does not assume honest reporting of measurement results by the main or auxiliary parties. Indeed, any adversarial strategy can be absorbed into the local measurement, and such a strategy can pass our test only if the parties effectively implement the desired measurements.

		\subsection{Efficient self-testing of generic multipartite states}
		We now apply our lifting theorem to the strongest notion of self-testing: Certifying a target state. 
		The crux of our approach is to construct a randomized Pauli measurement scheme whose associated observable $L$ efficiently and robustly certifies a target state $\Psi$. 
		Ideally, $L$ should have $\Psi$ as its eigenvector with maximal eigenvalue $\lambda_1 = 1$, and it should exhibit a sufficient spectral gap $\Delta>0$ between $\lambda_1$ and the second-largest eigenvalue. 
		Such a gap ensures robustness: if one can estimate $\tr(L\rho) \approx  1 - c_1\Delta$ with additive error $\varepsilon=c_2\Delta$, then $\rho$ must be close to $\Psi$: $\tr(\rho \Psi) \ge 1 - c_1-c_2$.
		
		A recent shadow-overlap protocol achieves this for generic states~\cite{Huang2024Certifying} by constructing an observable with maximal eigenvalue $1$ and spectral gap $\Delta=\Omega(n^{-2})$. 
		Combined with our lifting theorem, this immediately yields an efficient and robust self-testing protocol for generic states. 
		Yet, the resulting construction inherits the shared-randomness requirement from~\cite{Huang2024Certifying}. 
		To remove this requirement, we introduce a random-basis-enhanced variant that can be implemented using only local randomness. Specifically, we consider the observable
		\begin{equation}
			L = M_{\Psi} \coloneqq \frac{1}{n3^n}\sum_{\vec{P} \in \{X,Y,Z\}^{\otimes n}}\sum_{k=0}^{n-1} \sum_{\vec{b}^{(k)} \in \{0,1\}^{n-1}}
			\ketbra{\vec{b}^{(k)}, \vec{P}^{(k)}}_{[n]\setminus\{k\}} \otimes \ketbra{\Psi_{\vec{b}, \vec{P}}^{(k)}},
		\end{equation}
		where $\vec{b}^{(k)} \coloneqq (b_0,\cdots,b_{k-1},b_{k+1},\cdots,b_{n-1}), \vec{P}^{(k)} = (P_0,\cdots,P_{k-1}, P_{k+1},\cdots,P_{n-1})$, $\ket{\vec{b}^{(k)}, \vec{P}^{(k)}}$ denotes the basis state corresponding to measurement outcomes $\vec{b}^{(k)}$ in basis $\vec{P}^{(k)}$, and $\ket{\Psi_{\vec{b}, \vec{P}}^{(k)}}$ denotes the normalized state on system $k$ after projecting $\ket{\Psi}$ onto $\ket{\vec{b}^{(k)}, \vec{P}^{(k)}}_{[n]\setminus\{k\}}$.
		Intuitively, this observable first chooses a qubit $k$ uniformly at random, measures the remaining $n-1$ qubits in a random Pauli basis, and then measures the fidelity of the resulting single-qubit post-measurement state against the ideal reduced state $\ket{\Psi_{\vec{b}, \vec{P}}^{(k)}}$.
		We prove in Appendix~\ref{app:self-testing-fin-samples} that $M_{\Psi}$ preserves the spectral gap $\Delta=\Omega(n^{-2})$.
		Moreover, it can be written in the form of Eq.~\eqref{eq:L-def} using local shadow tomography~\cite{huang_predicting_2020}. 
		Applying our lifting theorem to this Pauli scheme then yields an efficient and robust self-testing protocol that self-tests almost all $n$-qubit pure states using only local randomness and polynomial sample complexity. This presents an exponential improvement over previously known multipartite self-testing schemes~\cite{Supic2023NetworkSelftest, Balanzo2024PureMultipartiteEntangled}. 
		Moreover, the protocol provides inverse-polynomial noise tolerance against device imperfections.
		
		\begin{theorem}[Efficient self-testing of generic multipartite states]\label{thm:self-testing-state-main-text}
			Let $\Psi$ be an $n$-qubit pure state drawn Haar-randomly. Then, with probability at least $1-\exp(-\Omega(n))$ over the choice of $\Psi$, the protocol obtained by applying Theorem~\ref{thm:efficient-selftest-pauli-msm-informal-main-text} to $M_{\Psi}$ achieves, with probability at least $1-\delta$, soundness and completeness for $\varepsilon$-approximate self-testing of $\Psi$ (in the sense of Theorem~\ref{thm:efficient-selftest-pauli-msm-informal-main-text} and Definition~\ref{def:approx-selftest-property}). The protocol uses $N=\cO\!\left(\mathrm{poly}(n,\varepsilon^{-1})\log\delta^{-1}\right)$ rounds and requires only local randomness.
		\end{theorem}
		
		The formal version of Theorem~\ref{thm:self-testing-state-main-text} is given by Theorem~\ref{thm:self-testing-state-local-randomness} in the Appendix.
		We remark that Theorem~\ref{thm:self-testing-state-main-text} does not directly apply to certain highly structured states, such as GHZ states. Nonetheless, such states often admit highly efficient self-testing protocols precisely because of their strong structure~\cite{Baccari2020ScalableQubitGraph}. 
		Our result therefore complements the existing self-testing literature by providing a general, sample-efficient method that applies to all but an exponentially small fraction of the state space, thereby dramatically enlarging the class of states known to admit sample-efficient self-testing. 
		
		Our protocol also enjoys additional advantages. First, the local measurements performed by each party are independent of the target state being self-tested, so a single dataset can be reused to self-test many target states simultaneously. Second, the auxiliary parties obtain no information about the state held by the main parties as long as the Bell-state-measurement outcomes arising from teleportation are not disclosed to them. In this sense, the protocol provides perfect privacy of both the target state and the underlying state against the auxiliary parties.

		\section{Discussion}\label{sec:discussion}
		
		An important next step is to reduce the resource requirements of our protocols further. For instance, the sample complexity of our self-testing protocol for generic states can be further improved, as recent results suggest that certifying generic quantum states may be possible with constant sample complexity~\cite{du2025certifyinglocalizablequantumproperties}. 
		Moreover, our protocol can be substantially more efficient for self-testing structured target states or properties. 
		As one example, the overhead associated with teleporting all parties could be largely avoided for states and properties that admit local certification, such as qLDPC codes~\cite{Breuckmann2021QuantumLDPC} or code subspaces~\cite{Baccari2020DISubspace}. These tasks are directly relevant to protecting quantum information and implementing quantum secret sharing in networks~\cite{Hillery1999QuantumSecretSharing}. 
		As another example, adapting our techniques to self-test ground states of local Hamiltonians~\cite{Wang2024CertifyingGroundState} could provide new routes to verifying quantum computations~\cite{fitzsimons2018PostHoc}. 
		Our framework can also be applied to certify complex structures of the states distributed across the network, such as circuit complexity and multipartite entanglement~\cite{du2025certifyinglocalizablequantumproperties}.  
		Since our self-testing protocol is robust against adversarial parties and, in the state-self-testing setting, also provides privacy against auxiliary parties, it is intriguing to explore its cryptographic applications, for example, in verifiable blind quantum computing, multi-prover interactive proofs, and quantum zero-knowledge proofs.

		Finally, it remains important to fully understand the ultimate capabilities and limitations of self-testing. 
		While our method applies to almost all multipartite qubit states, an important open question is whether \emph{arbitrary} quantum states can be self-tested robustly and sample-efficiently.
		A closely related challenge is to achieve sample-efficient self-testing in the standard Bell scenario, eliminating the need for auxiliary Bell pairs or additional parties~\cite{Balanzo2024PureMultipartiteEntangled}. 
		Furthermore, extending our framework to multi-qudit systems of prime dimension provides a natural direction for future work. 
		We expect this generalization to be feasible, drawing confidence from recent successes in self-testing odd-prime maximally entangled states~\cite{meyer2025robustly} and $\mathbb{Z}_d$ toric code ground states~\cite{vallee2025tailoring}.

		\begin{acknowledgements}
			This work is supported by the National Natural Science Foundation of China Grants No.~12174216,  the Innovation Program for Quantum Science and Technology Grant No.~2021ZD0300804,  No.~2021ZD0300702, and the IIIS Young Scholar Innovation Fund of Turing AI Institute (Nanjing).
		\end{acknowledgements}
		
		\bibliography{ref}
		
		\newpage
		\onecolumngrid
		\appendix

		\begin{center}
			{\large \textbf{Supplementary Material}}
		\end{center}
		
		\section*{Contents} 
		
		\makeatletter
		\let\oldaddcontentsline\addcontentsline
		\renewcommand{\addcontentsline}[3]{%
			\def\target{#1}%
			\def\toclabel{toc}%
			\ifx\target\toclabel
			\oldaddcontentsline{atoc}{atoc#2}{#3}%
			\else
			\oldaddcontentsline{#1}{#2}{#3}%
			\fi
		}
		
		\@starttoc{atoc}
		\makeatother
		
		\vspace{1cm}
		
		\renewcommand{\thetheorem}{S\arabic{theorem}}
		\renewcommand{\thefact}{S\arabic{fact}}
		\renewcommand{\thelemma}{S\arabic{lemma}}
		\renewcommand{\thedefinition}{S\arabic{definition}}
		\renewcommand{\theproposition}{S\arabic{proposition}}
		\renewcommand{\thecorollary}{S\arabic{corollary}}
		\renewcommand{\theclaim}{S\arabic{claim}}
		\renewcommand{\thepage}{S\arabic{page}}
		\renewcommand{\thefigure}{S\arabic{figure}}
		
		\renewcommand{\theHtheorem}{S\arabic{theorem}}
		\renewcommand{\theHfact}{S\arabic{fact}}
		\renewcommand{\theHlemma}{S\arabic{lemma}}
		\renewcommand{\theHdefinition}{S\arabic{definition}}
		\renewcommand{\theHproposition}{S\arabic{proposition}}
		\renewcommand{\theHcorollary}{S\arabic{corollary}}
		\renewcommand{\theHclaim}{S\arabic{claim}}
		\renewcommand{\theHfigure}{S\arabic{figure}}
		
		\setcounter{theorem}{0}
		\setcounter{fact}{0}
		\setcounter{lemma}{0}
		\setcounter{equation}{0}
		\setcounter{definition}{0}
		\setcounter{proposition}{0}
		\setcounter{claim}{0}
		\setcounter{corollary}{0}
		\setcounter{figure}{0}
		\setcounter{page}{1}
		\setcounter{section}{0}
		\setcounter{equation}{0}
		
		\section{Setup}\label{sec:setup}
		
		We begin by establishing the notation and framework for our self-testing protocols. 
		A self-testing protocol compares a \textit{reference experiment} with a \textit{physical experiment}.
		The reference experiment represents an ideal setup that generates a target probability distribution. 
		The physical experiment corresponds to the actual laboratory implementation, and is treated with minimal assumptions regarding its internal functioning. 
		By comparing the observed statistics of the physical experiment with those of the reference experiment, we can certify properties of the physical implementation.
		
		\subsection{The reference experiment} 
		The reference experiment involves an $n$-qubit target state $\tau$ (or $\Psi$, if pure) distributed across $n$ separated systems $T=T_0T_1\ldots T_{n-1}$. 
		The setup utilizes two types of auxiliary entanglement: $n$ Bell states $\ket{\phi^+} = \frac{1}{\sqrt{2}} (\ket{00} + \ket{11})$ shared between systems $S=S_0S_1\ldots S_{n-1}$ and auxiliary parties $B=B_0B_1\ldots B_{n-1}$, and $n-1$ Bell states shared between adjacent link qubits $R_{l}^{\vartriangleright}$ and $R_{l+1}^{\vartriangleleft}$ for $l \in [n-1]$ (where $[n] \coloneqq \{0, \dots, n-1\}$).
		We define the main party $A_l$ for each $l$ to include the system $T_l$ and $S_l$, as well as the relevant link qubits. Specifically, internal sites ($0 < l < n-1$) contain $\{T_l, S_l, R_{l}^{\vartriangleright}, R_{l}^{\vartriangleleft}\}$, while the boundary sites $l=0$ and $l=n-1$ contain only one link qubit ($R_{0}^{\vartriangleright}$ and $R_{n-1}^{\vartriangleleft}$, respectively).
		The complete setup is illustrated in Fig.~\ref{fig:exp-setup}a.
		
		\begin{figure}[!htbp]
			\includegraphics[width=\linewidth]{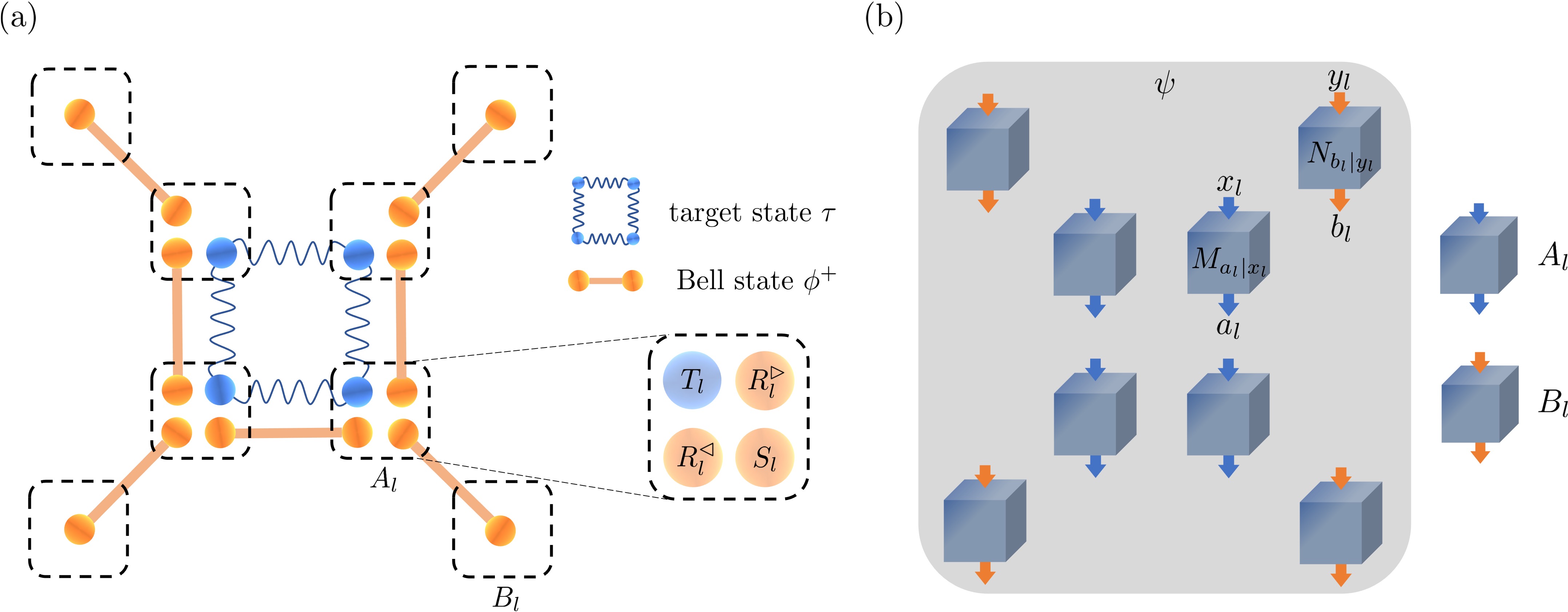}
			\caption{
				Experimental setup.
				(a) Schematic of the reference experiment. The setup involves $n$ main parties $A_0,\dots, A_{n-1}$ and $n$ auxiliary parties $B_0,\dots,B_{n-1}$.
				Each main party $A_l$ contains: (i) one qubit $T_l$ from the $n$-qubit target state $\tau$; (ii) one qubit $S_l$ of a Bell state $\phi^+$ maximally entangled with the auxiliary qubit $B_l$; and (iii) the necessary link qubits $R_{l}^{\vartriangleright}$ and $R_{l}^{\vartriangleleft}$ that form Bell states with adjacent sites. Specifically, one qubit $R_{l}^{\vartriangleright}$ of a Bell state in $A_l$ is maximally entangled with $R_{l+1}^{\vartriangleleft}$ in $A_{l+1}$ for $l \in [n-1]$. The boundary parties $A_0$ and $A_{n-1}$ each contain only one link qubit.
				(b) Schematic of the physical experiment. The setup involves $2n$ spatially separated parties, treated as black boxes with no knowledge of their internal workings. Each party receives an input and returns an output. The whole system is described by an unknown pure state $\psi$.}
			\label{fig:exp-setup}
		\end{figure}

		\begin{remark}
			It is possible to modify the network structure in Fig.~\ref{fig:exp-setup}a according to experimental needs with only minor adjustments to our subsequent analysis. For example, entangled link qubits $R_{l}^{\vartriangleright}$ and $R_{l+1}^{\vartriangleleft}$ could be moved from $A_l$ and $A_{l+1}$ to $B_l$ and $B_{l+1}$ instead. In this work, we always focus on the architecture outlined in Fig.~\ref{fig:exp-setup}a for the sake of concreteness and clarity.
		\end{remark}
		
		\begin{table}[htbp]
			\centering
			\begin{tabular}{c|c|c|c|c}
				\hline
				Input & Reference measurement & Party & POVM  & Observable\\
				\hline
				$x_l\in \{0,1\}$  & $\frac{\sigma_x+(-1)^{x_l}\sigma_z}{\sqrt{2}}_{S_l}$ & $A_l$ & $M^{(l)}_{a_l|x_l}, a_l=0,1$ & $A_{x_l}^{(l)} = \sum_{a_l = 0,1}(-1)^{a_l} M_{a_l|x_l}^{(l)}$ \\ 
				$x_l\in \{2,3\}$  & $\frac{\sigma_x+(-1)^{x_l} \sigma_y}{\sqrt{2}}_{S_l}$ & $A_l$ & $M^{(l)}_{a_l|x_l}, a_l=0, 1$ & $A_{x_l}^{(l)} = \sum_{a_l = 0,1}(-1)^{a_l} M_{a_l|x_l}^{(l)}$ \\
				$x_l\in \{4,5\}$  & $\frac{\sigma_z+(-1)^{x_l} \sigma_y}{\sqrt{2}}_{S_l}$ & $A_l$ & $M^{(l)}_{a_l|x_l}, a_l=0, 1$ & $A_{x_l}^{(l)} = \sum_{a_l = 0,1}(-1)^{a_l} M_{a_l|x_l}^{(l)}$ \\
				$y_l = 0$ & $\sigma_x$ & $B_l$ & $N^{(l)}_{b_l|0}, b_l=0,1$ & $B_{0}^{(l)} = \sum_{b_l = 0,1}(-1)^{b_l} N_{b_l|0}^{(l)}$ \\
				$y_l = 1$ & $\sigma_y$ & $B_l$ & $N^{(l)}_{b_l|1}, b_l=0,1$ & $B_{1}^{(l)} = \sum_{b_l = 0,1}(-1)^{b_l} N_{b_l|1}^{(l)}$\\
				$y_l = 2$ & $\sigma_z$ & $B_l$ & $N^{(l)}_{b_l|2}, b_l=0,1$ & $B_{2}^{(l)} = \sum_{b_l = 0,1}(-1)^{b_l} N_{b_l|2}^{(l)}$ \\
				$x_l =\ \Dr$ & BSM on $R_l^{\vartriangleright}S_{l}$ & $A_l$ & $M^{(l)}_{\vec{a}_l|\Dr}, \vec{a}_l\in \{0,1\}^2$ & 
				\\
				$x_l =\ \Dl$ & BSM on $R_l^{\vartriangleleft}S_{l}$ & 
				$A_l$ & $M^{(l)}_{\vec{a}_l|\Dl}, \vec{a}_l\in \{0,1\}^2$ & 
				\\
				$x_l = \diamond$ & BSM on $T_lS_l$ & $A_l$ & $M^{(l)}_{\vec{a}_l|\diamond}, \vec{a}_l\in \{0,1\}^2$ &  \\
				\hline
			\end{tabular}
			\caption{Experimental setup. For each index $l\in [n]$, the table specifies the measurement settings (inputs) and the ideal implementations for the reference experiment, as well as the corresponding measurement operators and observables in the physical system.}
			\label{tab:experiment_setup}
		\end{table}
		
		The input sets for the main parties $A_l$ and the auxiliary parties $B_l$ are 
		$x_l \in \{0, \dots, 5, \vartriangleleft, \vartriangleright, \diamond\}$ and $y_l \in \{0, 1, 2\}$, respectively. 
		The correspondence between these inputs and the reference measurements is summarized in Table~\ref{tab:experiment_setup}. Below, we briefly describe the role of those measurements.
		
		For inputs $x_l \in \{0,\dots,5\}$ and $y_l \in \{0,1,2\}$, the main party $A_l$ and the auxiliary party $B_l$ perform the CHSH game~\cite{Clauser1969CHSH} on the Bell state $\phi^+$ shared between the systems $S_l$ and $B_l$. 
		Each CHSH game self-tests two of the three measurement settings on the auxiliary system, locking their algebraic properties to those of the Pauli group and ensuring that the shared state is locally isometric to a Bell state. 
		The combined 3-CHSH game robustly self-tests that the auxiliary parties are implementing measurements that, up to complex conjugation, are locally isometric to ideal Pauli measurements~\cite{McKague2012RobustSinglet, Bamps2015SoSDecomposition, Bowles2018DIAllEntangled} on a Bell state. See Sec.~\ref{app:3CHSH} for details.
		
		For inputs $x_l = \diamond$ and $y \in \{0,1,2\}$, the main parties $A_l$ teleport the target state $\tau$ on $T_l$ to $B_l$ via a local BSM. The auxiliary parties $B_l$ then perform Pauli measurements, certified by the 3-CHSH game, to perform self-tested multipartite Pauli measurements on $\tau$. 
		See Sec.~\ref{app:DI_pauli_partial_transpose} for details.
		
		However, the 3-CHSH game only self-tests multipartite Pauli measurements up to any combination of partial-transposes. 
		To overcome this partial-transpose problem, for inputs $x_l \in \{ \vartriangleleft, \vartriangleright\}$ and $y_l \in \{0,1,2\}$, the main parties perform BSMs on $R_{l}^{\vartriangleright}S_l$ and $R_{l+1}^{\vartriangleleft}S_{l+1}$ to teleport the Bell state on the link qubits $R_l^{\vartriangleright}R_{l+1}^{\vartriangleleft}$ to the auxiliary parties $B_lB_{l+1}$. 
		The auxiliary parties then perform their certified Pauli measurements on the teleported Bell state. 
		This procedure resolves the partial transpose problem and is detailed in Sec.~\ref{app:global_tranpose}.
		
		\subsection{The physical experiment}
		The physical experiment involves $2n$ spatially separated parties, denoted as $A_l$ and $B_l$ for $l \in [n]$, as illustrated in Fig.~\ref{fig:exp-setup}b. 
		We assume the devices operate within the framework of quantum mechanics and satisfy no-signaling constraints; yet, the specific state and measurements are uncharacterized. 
		Consequently, each party $A_l$ ($B_l$) is treated as a black box that receives an input $x_l$ ($y_l$) and produces an output $a_l$ ($b_l$). 
		The global system is described by an arbitrary $2n$-partite quantum state $\psi$. 
		Upon receiving input $x_l$ ($y_l$), party $A_l$ ($B_l$) performs the POVM $\{M_{a_l|x_l}^{(l)}\}_{a_l}$ ($\{N_{b_l|y_l}^{(l)}\}_{b_l}$). 
		
		The self-testing framework relies on several standard physical assumptions:
		\begin{enumerate}
			\item Quantum Mechanics: Quantum mechanics provides a complete description of the physical experiment.
			\item Locality: Measurements performed on a local system act trivially on any spatially separated system.
			\item Free Will: Each party possesses the ``free will'' to choose their measurement settings. Consequently, the settings are statistically independent of the source and are distributed according to the probabilities specified by the protocol.
		\end{enumerate}
		Under these assumptions, the joint probability distribution of the measurement outcomes is determined by Born's rule:
		\begin{equation}\label{eq:measurement_prob}
			\Pr(\vec{a}, \vec{b} | \vec{x}, \vec{y}) = \tr \left[ \left( \bigotimes_{l=0}^{n-1} M_{a_l|x_l}^{(l)} \otimes N_{b_l|y_l}^{(l)} \right) \rho \right].
		\end{equation}

		We further make two standard assumptions without loss of generality. 
		First, we assume the underlying quantum state $\psi = \ketbra{\psi}$ is a pure state. 
		This is justified because any mixed state can be purified by introducing an (inaccessible) environment $R$.
		Second, we assume all measurements are projective, meaning the POVM elements $M_{a_l|x_l}^{(l)}$ and $N_{b_l|y_l}^{(l)}$ are projectors. 
		This is valid because any general POVM can be implemented equivalently by a projective measurement on an enlarged Hilbert space by Naimark's dilation theorem.
		For projective measurements with binary outcomes ($a_l,b_l \in \{0,1\}$), we define the relevant dichotomic observable $A_{x_l}^{(l)} \coloneqq \sum_{a_l = 0,1}(-1)^{a_l} M_{a_l|x_l}^{(l)}$ and $B_{y_l}^{(l)} = \sum_{b_l = 0,1}(-1)^{b_l} N_{b_l|y_l}^{(l)}$. 
		The notations are summarized in Table~\ref{tab:experiment_setup}.

		For later convenience, we define 
		\begin{equation}\label{eq:Pauli-notation-definition}
			\begin{split}
				X_A^{ (l)} &= \frac{A_0^{(l)} + A_1^{(l)}}{\sqrt{2}} \quad\quad  X_B^{ (l)} = B_0^{(l)} \\
				Y_A^{ (l)} &= \frac{A_2^{(l)} - A_3^{(l)}}{\sqrt{2}} \quad\quad  Y_B^{ (l)} = B_1^{(l)}\\
				Z_A^{ (l)} &= \frac{A_0^{(l)} - A_1^{(l)}}{\sqrt{2}} \quad\quad  Z_B^{ (l)} = B_2^{(l)}
			\end{split}
		\end{equation}
		which become three Pauli matrices $\sigma_X, \sigma_Y$ and $\sigma_Z$ in the reference experiment.
		
		For simplicity, we initially assume that the physical experiment is repeated in an i.i.d.\ manner—that is, the state and measurement operators remain the same across rounds. 
		This i.i.d.\ assumption will be removed in Section~\ref{sec:non-iid}.
		
		\section{Protocol overview}\label{sec:protocol-overview}
		In this section, we provide an overview of the DI protocols employed in this work. We begin by introducing the estimators used to characterize the physical experiment based on the observed statistics $\Pr(\vec{a}, \vec{b} | \vec{x}, \vec{y})$ (Section~\ref{subsec:overview-estimators}). Next, we describe the measurement protocols designed to approximate these estimators from finite experimental samples (Section~\ref{subsec:overview-sampling}). Finally, we discuss how these estimations facilitate the self-testing of quantum states and the DI certification of generic physical properties (Section~\ref{subsec:overview-certification}).
		
		\subsection{Self-testing multipartite Pauli measurement schemes with infinite samples}\label{subsec:overview-estimators}
		An overview of the estimators employed is provided in Box~\ref{box:ProtocolEstimatorsInformal}. Our primary goal is to implement a multipartite Pauli measurement scheme, defined by the observable $L$ in Eq.~\eqref{eq:L-def}, in a DI manner.
		
		The first two estimators, $\vI$ and $\vII$, are independent of the specific observable $L$ and serve to self-test the underlying measurement and teleportation apparatus of the quantum network:
		\begin{enumerate}[label=(\Roman*)]
			\item A near-zero value for $\vI$ guarantees that parties $A_l$ and $B_l$ approximately share Bell states and that their \enquote{black boxes} implement the Pauli measurement scheme up to a local transpose. This estimator is analyzed in detail in Appendix~\ref{app:3CHSH}.
			
			\item While a vanishing $\vI$ ensures certification of Pauli measurements up to local transpose, a technical challenge remains: correlating these local transposes to ensure the Pauli measurements are consistent up to a single, global conjugation. This issue is resolved if $\vII$ is close to zero in the physical experiment, as discussed in Appendix~\ref{app:global_tranpose}.
		\end{enumerate}
		
		Conditional on $\vI, \vII \approx 0$, we can execute any certification protocol based on local Pauli measurements. 
		The expectation value $\tr(L\rho)$ is then estimated by $\vIII$, valid up to a global conjugation and a local extraction map $\Lambda$. 
		The challenge of estimating $\tr(L\rho)$ relative to a global conjugation, rather than a mixture of local transposes, is resolved in Appendix~\ref{app:global_tranpose}.
		
		An additional technical challenge is refining the extraction map $\Lambda$. We show that this map can be simplified from a map that acts on both main and auxiliary parties and involves classical communication to a product of local maps that act only on the main parties $A_l$. This refinement is shown in Appendix~\ref{subsec:refinement-extraction-channel}.
		
		\begin{mybox}[label={box:ProtocolEstimatorsInformal}]{Estimation protocols (informal version of Box~\ref{box:ProtocolEstimators})}
			\begin{enumerate}[label=(\Roman*)]
				\item \textbf{Goal:} Approximately verify that each pair $(A_l, B_l)$ shares a Bell state and that, for $P \in \{X, Y, Z\}$, the observables $P_A^{(l)}$ and $P_B^{(l)}$ (as defined in \eqref{eq:Pauli-notation-definition}) implement Pauli matrices up to a local transpose.
				
				\textbf{Method:} Verify that the observed statistics $p_l(a_l, b_l | x_l, y_l)$ approximately minimize three independent CHSH inequalities (indexed by $k$) of the form:
				\begin{equation}
					\vI(l,k) = 2\sqrt{2} - \sum_{i,j\in \{0,1\}}\sum_{a_l,b_l\in \{0,1\}}p_l(a_l,b_l|x(i,k),y(j,k))(-1)^{ij}(-1)^{a_l+b_l}\quad
					\begin{cases}
						\geq 0, \quad \mathrm{Physical~Experiment}\\
						= 0, \quad \mathrm{Reference~Experiment}
					\end{cases}\label{eq:CHSH-informal}
				\end{equation}
				for a suitable choice of functions $x: i,k \mapsto x(i,k)\in \{0,1,2,3,4,5\},\ y: i,k \mapsto y(i,k)\in \{0,1,2\}$.
				
				\item \textbf{Goal:}  Correlate the local transposes across different systems $l \in [n]$ to ensure a consistent global transpose.
				
				\textbf{Method:} Use quantum teleportation on the main systems to entangle the auxiliary systems $B_l$ and $B_{l+1}$, then verify that the measurement statistics approximately minimize the expectation
				\begin{equation}
					\vII(l,l+1) = 1 - \frac13\mathbb{E}[K]\quad\begin{cases}
						\geq 0, \quad \mathrm{Physical~Experiment}\\
						= 0, \quad \mathrm{Reference~Experiment}
					\end{cases}
				\end{equation}
				for the stabilizer-like operator $K = X_B^{(l)}X_B^{(l+1)} - Y_B^{(l)}Y_B^{(l+1)}+Z_B^{(l)}Z_B^{(l+1)}$.
				
				\item\textbf{Goal:} Estimate the expectation value $\tr(L\psi)$ up to a global conjugation and local extraction channels.
				
				\textbf{Method:} Teleport the state $\psi$ from the main systems to the auxiliary systems. The value $\tr(L\psi)$ is then estimated via local Pauli measurements according to the distribution $\mathcal{D}$ and the post-processing function $\omega(\vec{b}|\vec{P})$:
				\begin{equation}
					\vIII = \mathbb{E}[\omega(\vec{b}|\vec{P})]\quad\begin{cases}
						\approx \tr[\left(L\otimes\ketbra{0}^{\otimes n}+L^*\otimes\ketbra{1}^{\otimes n} \right)\Lambda(\rho)], \quad \mathrm{Physical~Experiment~(given~} \vI,\vII\approx 0 \mathrm{)}\\
						= \tr(L\psi), \quad \mathrm{Reference~Experiment}
					\end{cases}.
				\end{equation}
			\end{enumerate}
		\end{mybox}
		
		\subsection{Self-testing multipartite Pauli measurement schemes with finite samples}\label{subsec:overview-sampling}
		
		In Fig.~\ref{fig:protocol-overview}, we sketch the procedure for estimating $\vI, \vII$, and $\vIII$ from finite measurement samples. 
		A more comprehensive technical discussion of these protocols is provided in Appendix~\ref{sec:sample_comp}, and Protocols~\ref{prot:di-witness-local-randomness} and \ref{prot:di-witness-shared-randomness}.
		
		\begin{figure}[!htbp]
			\centering
			
			\includegraphics[width = 0.7\textwidth]{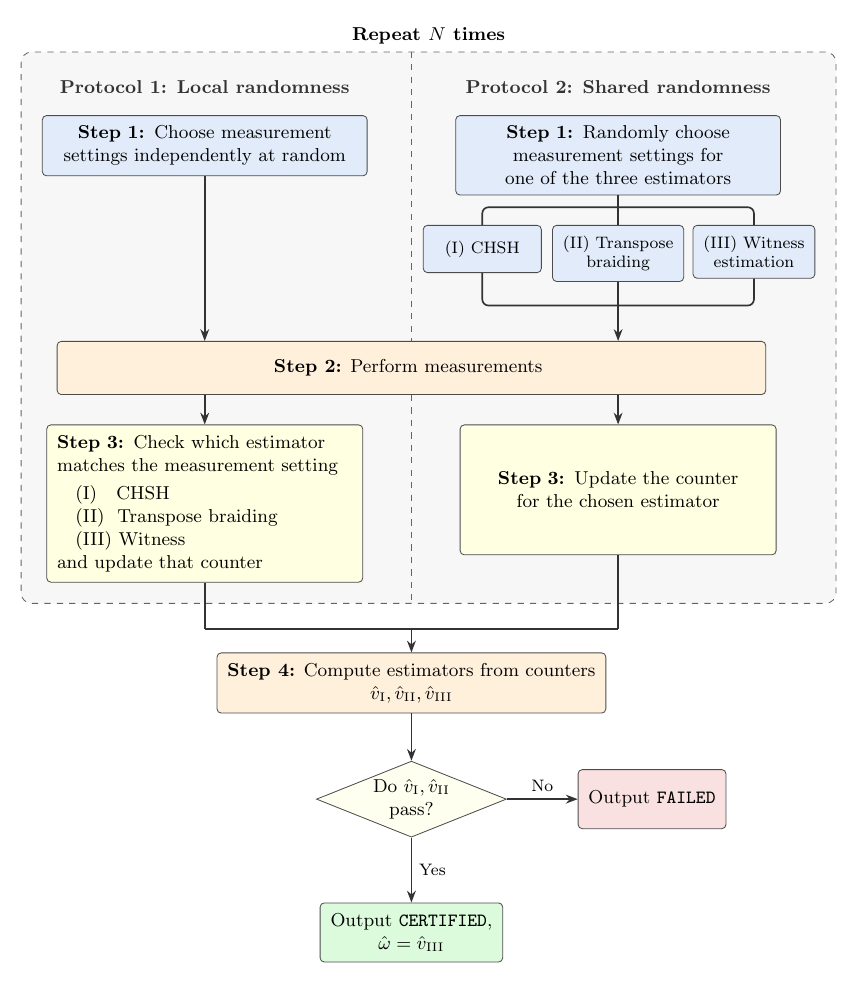}
			\caption{Overview of our DI self-testing protocols (local randomness on the left, shared randomness on the right). 
				With shared randomness, the parties can select an estimator in advance (see Box~\ref{box:ProtocolEstimatorsInformal}) and sample only those measurement settings that contribute to it. 
				With only local randomness, the estimator cannot be coordinated a priori, so Step~3 performs a sifting procedure to choose an estimator relevant to the measurement settings.}
			\label{fig:protocol-overview}
		\end{figure}
		
		The primary objective of these protocols is to estimate the true values of the estimators defined in Box~\ref{box:ProtocolEstimatorsInformal} to a specified precision. This required precision is derived in Theorem~\ref{thm:DI-estimation-L}, which constitutes one of the main results of this work. 
		Once the target precision is established, we perform a statistical analysis to derive upper bounds on the corresponding sample complexity and the total number of measurement rounds required. Here, we consider two self-testing scenarios:
		\begin{enumerate}
			\item Local Randomness: Measurement settings are chosen independently by each spatially separated party (Protocol~\ref{prot:di-witness-local-randomness}).
			\item Shared Randomness: Measurement settings are sampled from a joint distribution, enabled by shared randomness between the parties (Protocol~\ref{prot:di-witness-shared-randomness}).
		\end{enumerate}
		
		We distinguish between these two cases because the lack of shared randomness increases the sample complexity by a factor of $n$. 
		This overhead stems from the requirement to choose uncorrelated local strategies while ensuring that the highly correlated setting necessary for teleportation ($x_l = \diamond$ for all $l$, as shown in Table~\ref{tab:experiment_setup}) is sampled with sufficient frequency.

		\subsection{Robust self-testing of quantum states}\label{subsec:overview-certification}
		
		Finally, we employ the estimator $\hvIII$, obtained via the protocols illustrated in Fig.~\ref{fig:protocol-overview}, to approximate the expectation value $\tr(L\rho)$. Provided that $L$ serves as an effective witness for a specific property, this quantity provides a DI certification of that same property.
		
		This certification is robustly sound, meaning that any state with a trace distance to the target state exceeding a small threshold $\varepsilon$ will be rejected with high probability. 
		Furthermore, it is robustly complete: if the physical experiment sufficiently approximates the reference experiment, the protocol will successfully certify the state with high probability. 
		Formal definitions of soundness and completeness, along with the corresponding proofs, are provided in Appendices~\ref{app:DI-multipartite-Pauli} and~\ref{app:self-testing-fin-samples}.
		
		As a concrete application and a major contribution of this work, we apply our protocols to a multipartite Pauli measurement scheme for fidelity certification, as introduced in Refs.~\cite{Huang2024Certifying, du2025certifyinglocalizablequantumproperties}. Our results, presented in Corollary~\ref{col:self-testing-Haar-random-states}, establish a robust and sample-efficient self-testing scheme for almost all pure quantum states $\Psi$. This represents an exponential improvement over previous protocols~\cite{Supic2023NetworkSelftest, Balanzo2024PureMultipartiteEntangled}.
		
		\section{Robust self-testing of Bell states and Pauli measurements}\label{app:3CHSH}
		In this section, we present a robust analysis for self-testing Bell states and Pauli measurements using three independent CHSH inequalities. Specifically, the self-testing is performed by constructing three CHSH inequalities, each utilizing a different pair of measurement bases drawn from the Pauli set $\{\sigma_x, \sigma_y, \sigma_z\}$~\cite{Bowles2018DIAllEntangled}.
		
		Let $l \in [n]$ denote the location and $k \in \{0, 1, 2\}$ index the specific CHSH test performed. The measurement settings (inputs) $(x_l(0,k), x_l(1,k))$ for party $A_l$ and $(y_l(0,k), y_l(1,k))$ for party $B_l$ are specified by:
		\begin{equation}\label{eq:3-CSHS-inputs}
			(x_l(0,k), x_l(1,k), y_l(0,k), y_l(1,k)) = 
			\begin{cases}
				(0, 1, 0, 2) & \text{if } k=0, \\
				(3, 2, 0, 1) & \text{if } k=1, \\
				(5, 4, 2, 1) & \text{if } k=2.
			\end{cases}
		\end{equation}
		Let $a_l(i,k) \in \{0, 1\}$ and $b_l(j,k) \in \{0, 1\}$ denote the measurement outcomes for party $A_l$ when input $x_l(i,k)$ is chosen, and for party $B_l$ when input $y_l(j,k)$ is chosen, respectively, where $i, j \in \{0, 1\}$.
		The CHSH violation $v_\mathrm{I}(l,k)$ is defined as the distance from the quantum maximum:
		\begin{equation}\label{eq:CHSHestimator}
			v_\mathrm{I}(l,k) = 2\sqrt{2}-\sum_{ij = 0,1}(-1)^{ij}\mathbb{E}[(-1)^{a_l(i,k)+b_l(j,k)}\mid x_l=x_l(i,k),y_l=y_l(j,k)].
		\end{equation}
		
		The original CHSH inequality establishes that $v_\mathrm{I}(l,k) \ge 0$. Crucially, the maximum quantum violation, $v_\mathrm{I}(l,k)=0$, is achieved if and only if the state and measurements are equivalent (up to a local isometry) to a Bell state and the optimal CHSH measurements on that state.
		This equivalence is also robust. Specifically, a small violation $v_\mathrm{I}(l,k) \approx 0$ implies that the underlying state and measurements are approximately equivalent to the ideal Bell state and optimal measurements~\cite {Bamps2015SoSDecomposition}.
		We provide a self-contained analysis of this robust self-testing procedure here.
		
		Below, for two scalars or vectors $v$ and $w$, we use the notation $v \approx_{\epsilon} w$ to denote $\norm{v-w} = \cO(\epsilon)$.
		
		\begin{lemma}[Self-testing operator relations from CHSH game]\label{lemma:chsh-implications}
			Suppose for $l\in [n],\ \forall k\in \{0,1,2\}$, we have $v_\mathrm{I}(l,k)\le\epsilon$. Then:
			\begin{enumerate}
				\item Teleportation: $\forall P\in\{X,Y,Z\}$,
				\begin{equation}\label{eq:Bell-stat-X-YZ}
					\norm{\left(P^{(l)}_{A}-(-1)^{\delta_{PY}} P_{B}^{(l)}\right)\ket{\psi}} = \cO(\sqrt{\epsilon})  
				\end{equation}
				\item Anti-commutation: $\forall C\in \{A,B\}, \forall P,Q \in \{X,Y,Z\}$,
				\begin{equation}\label{eq:anti-comm}
					\norm{\left(\{P^{(l)}_{C}, Q^{(l)}_{C}\}-2\delta_{P,Q}\mathbb{I}\right)\ket{\psi}} = \cO(\sqrt{\epsilon})
				\end{equation}
				\item\label{bullet:unitary-extension} The above two claims remain valid if we replace the operators $P_C^{(l)}$ and $Q_C^{(l)}$ by their unitary versions $\hat{P}_C^{(l)}$ and $\hat{Q}_C^{(l)}$. For a given Hermitian operator $O$, its unitary version $\hat{O}$ is defined by replacing non-zero eigenvalues of $O$ with their sign and zero eigenvalues by one. 
			\end{enumerate}
		\end{lemma}
		
		\begin{proof}
			For the first claim, we utilize the standard identity, decomposing the CHSH operator to a sum of squared terms~\cite{Bamps2015SoSDecomposition}.
			For each setting $k \in \{0, 1, 2\}$, define the dichotomic observables $\Tilde{A}_i$ on system $A$ and $\Tilde{B}_j$ on system $B$ as:
			\begin{equation}
				(\Tilde{A}_0, \Tilde{A}_1, \Tilde{B}_0, \Tilde{B}_1) =
				\begin{cases}
					(A^{(l)}_0, A^{(l)}_1, B^{(l)}_0, B^{(l)}_2) \quad &\mathrm{if}\ k=0,\\
					(A^{(l)}_3, A^{(l)}_2, B^{(l)}_0, B^{(l)}_1) \quad &\mathrm{if}\ k=1,\\
					(A^{(l)}_5, A^{(l)}_4, B^{(l)}_2, B^{(l)}_1) \quad &\mathrm{if}\ k=2.
				\end{cases}\label{eq:tilde-ops-def}
			\end{equation}
			The CHSH operator $\mathcal{B}_k$ for these settings is $\mathcal{B}_k = \sum_{i,j = 0,1}(-1)^{ij} \Tilde{A}_i\otimes \Tilde{B}_j$. It is known that the quantity $2\sqrt{2}\mathbb{I} - \mathcal{B}_k$ can be quadratically decomposed:
			\begin{equation}\label{eq:CHSH-quadratic}
				2\sqrt{2}\mathbb{I}-\mathcal{B}_k = \frac{1}{\sqrt{2}}\left[\left(\Tilde{B}_0-\frac{\Tilde{A}_0+\Tilde{A}_1}{\sqrt{2}}\right)^2+\left(\Tilde{B}_1-\frac{\Tilde{A}_0-\Tilde{A}_1}{\sqrt{2}}\right)^2\right].
			\end{equation}
			The condition $v_{\mathrm{I}}(l,k)\le \epsilon$, where $v_{\mathrm{I}}(l,k) = \tr[(2\sqrt{2}\mathbb{I}-\mathcal{B}_k) \psi]$, implies that the expectation value of the non-negative operator on the right-hand side is small:
			\begin{equation}
				\tr\left[\left(\Tilde{B}_i-\frac{\Tilde{A}_0+(-1)^i\Tilde{A}_1}{\sqrt{2}}\right)^2 \psi\right] \le \sqrt{2}\epsilon \quad \text{for } i\in \{0,1\}.
			\end{equation}
			Therefore, 
			\begin{equation}\label{eq:P1-approx-on-rho}
				\norm{\left(\Tilde{B}_i-\frac{\Tilde{A}_0+(-1)^i\Tilde{A}_1}{\sqrt{2}}\right)\ket{\psi}} = \cO(\sqrt{\epsilon}).
			\end{equation}
			The first claim is obtained by substituting the specific assignments: (1) for $P=X$, set $k=0$ and $i=0$; (2) for $P=Y$, set $k=1$ and $i=1$; (3) for $P=Z$, set $k=0$ and $i=1$.
			
			For the second claim, when $C=B$ and $P=Q$, since the operators $\Tilde{B}_i$ are dichotomic observables resulting from projectors, they satisfy $\Tilde{B}_i^2 = \mathbb{I}$, so the claim holds exactly (i.e., with $\epsilon=0$).
			
			When $C=A$ and $P=Q$, 
			\begin{equation}\label{eq:squared-to-one-AfromB}
				\begin{split}
					\norm{\ket{\psi}-\left(\frac{\Tilde{A}_0+(-1)^i\Tilde{A}_1}{\sqrt{2}}\right)^2\ket{\psi}}
					&= \norm{\Tilde{B}_i^2\ket{\psi}-\Tilde{B}_i\frac{\Tilde{A}_0+(-1)^i\Tilde{A}_1}{\sqrt{2}}\ket{\psi}+\Tilde{B}_i\frac{\Tilde{A}_0+(-1)^i\Tilde{A}_1}{\sqrt{2}}\ket{\psi}-\left(\frac{\Tilde{A}_0+(-1)^i\Tilde{A}_1}{\sqrt{2}}\right)^2\ket{\psi}} \\
					&\le \norm{\Tilde{B}_i\left(\Tilde{B}_i-\frac{\Tilde{A}_0+(-1)^i\Tilde{A}_1}{\sqrt{2}}\right)\ket{\psi}} + \norm{\frac{\Tilde{A}_0+(-1)^i\Tilde{A}_1}{\sqrt{2}}\left(\Tilde{B}_i-\frac{\Tilde{A}_0+(-1)^i\Tilde{A}_1}{\sqrt{2}}\right)\ket{\psi}} \\
					& = \cO(\sqrt{\epsilon}).
				\end{split}
			\end{equation}
			The final step follows from Eq.~\eqref{eq:P1-approx-on-rho}.
			
			When $C=A$ and $P\neq Q$, the operators $\frac{\Tilde{A}_0\pm\Tilde{A}_1}{\sqrt{2}}$ anti-commute by definition.
			
			When $C=B$ and $P\neq Q$, 
			\begin{equation}
				\begin{split}\label{eq:B-anticomm-from-A}
					\norm{\{\Tilde{B}_0,\Tilde{B}_1\}\ket{\psi}} &= \norm{\left(\{\Tilde{B}_0,\Tilde{B}_1\} - \left\{\frac{\Tilde{A}_0+\Tilde{A}_1}{\sqrt{2}},\frac{\Tilde{A}_0-\Tilde{A}_1}{\sqrt{2}}\right\}\right)\ket{\psi}} \\
					&= \Bigg\Vert \Bigg[
					\Tilde{B}_1\left(\Tilde{B}_0 - \frac{\Tilde{A}_0+\Tilde{A}_1}{\sqrt{2}}\right) + \Tilde{B}_0\left(\Tilde{B}_1 - \frac{\Tilde{A}_0-\Tilde{A}_1}{\sqrt{2}}\right) \\
					&\quad\quad+ \frac{\Tilde{A}_0+\Tilde{A}_1}{\sqrt{2}}\left(\Tilde{B}_1 - \frac{\Tilde{A}_0-\Tilde{A}_1}{\sqrt{2}}\right) 
					+ \frac{\Tilde{A}_0-\Tilde{A}_1}{\sqrt{2}}\left(\Tilde{B}_0 - \frac{\Tilde{A}_0+\Tilde{A}_1}{\sqrt{2}}\right)\Bigg]\ket{\psi} \Bigg\Vert  \\
					& = \cO(\sqrt{\epsilon}).
				\end{split}
			\end{equation}
			The final step follows from combining the triangular inequality and Eq.~\eqref{eq:P1-approx-on-rho}. Substituting the definitions for $k\in\{0,1,2\}$ proves this final case of the second claim.
			
			For the third claim, we first establish the closeness of $P_C^{(l)}$ and $\hat{P}_C^{(l)}$:
			\begin{equation}\label{eq:raw-and-unitarypart-close-on-psi-claim}
				\begin{split}
					\norm{\left(\hat{P}_C^{(l)}-P_C^{(l)}\right)\ket{\psi}}&=\norm{\left(\mathbb{I}-(\hat{P}_C^{(l)})^{\dag}P_C^{(l)}\right)\ket{\psi}}\\
					&=\norm{\left(\mathbb{I}-\abs{P_C^{(l)}}\right)\ket{\psi}}\\
					&\leq \norm{\left(\mathbb{I}+\abs{P_C^{(l)}|}\right)\left(\mathbb{I}-\abs{P_C^{(l)}}\right)\ket{\psi}}\\
					&=\norm{\left(\mathbb{I}-P_C^{(l) 2}\right)\ket{\psi}} \\
					&= \cO(\sqrt{\epsilon}).
				\end{split}
			\end{equation}
			The last line follows from Eq.~\eqref{eq:anti-comm}. This closeness relation allows us to replace $P_C^{(l)}$ with $\hat{P}_C^{(l)}$ in the above two claims while incurring an error of $\cO(\sqrt{\epsilon})$.
			For example, 
			\begin{equation}
				\begin{split}
					\norm{\{\hat{X}_A^{(l)}, \hat{Z}_A^{(l)}\} \ket{\psi}} &= \norm{(\hat{X}_A^{(l)}\hat{Z}_A^{(l)} + \hat{Z}_A^{(l)}\hat{X}_A^{(l)}) \ket{\psi}} \\
					&\approx_{\sqrt{\epsilon}} \norm{(\hat{X}_A^{(l)}Z_A^{(l)} + \hat{Z}_A^{(l)}X_A^{(l)}) \ket{\psi}} \\
					&\approx_{\sqrt{\epsilon}} \norm{(\hat{X}_A^{(l)}Z_B^{(l)} + \hat{Z}_A^{(l)}X_B^{(l)}) \ket{\psi}} \\
					&\approx_{\sqrt{\epsilon}} \norm{(X_A^{(l)}Z_B^{(l)} + Z_A^{(l)}X_B^{(l)}) \ket{\psi}} \\
					&\approx_{\sqrt{\epsilon}} \norm{(X_A^{(l)}Z_A^{(l)} + Z_A^{(l)}X_A^{(l)}) \ket{\psi}} \\
					&=\cO(\sqrt{\epsilon}).
				\end{split}
			\end{equation}
			The second and fourth lines use Eq.~\eqref{eq:raw-and-unitarypart-close-on-psi-claim}, the third and fifth lines use Eq.~\eqref{eq:Bell-stat-X-YZ}, and the last line uses Eq.~\eqref{eq:anti-comm}. 
			This completes the proof.
		\end{proof}
		
		Moreover, we show that not only do the measurement operators $\{\hat{P}_C\}$ satisfy the same operator relations as the Pauli operators when acting on $\ket{\psi}$, but the existence of these relations guarantees the existence of a local isometry $V_{A_l} \otimes V_{B_l}$. This isometry acts on the physical systems $A_l$ and $B_l$ to extract a Bell state $\ket{\phi^+}$. Under the action of this isometry, the physical operators $\hat{P}_C^{(l)}$ are mapped to the standard Pauli operators on the extracted Bell state.
		This equivalence is subject to a local-transpose freedom for the $\hat{Y}_C^{(l)}$ operator, which manifests as a tensor factor of $\sigma_Z$ (Eq.~\eqref{eq:intertwiner-claim}). 
		Importantly, this local freedom is correlated across both $A$ and $B$ systems, as the resulting ancillary state $\ket{\xi}$ in the following lemma only has $\ket{00}$ and $\ket{11}$ components on the auxiliary registers $A''B''$.
		
		\begin{lemma}[Local isometry]\label{lemma:pauli-upto-conj}
			Let $\ket{\psi}$ be a pure state on $\mathcal{H}_{ABR}$, and let $\hat{X}_C, \hat{Y}_C, \hat{Z}_C$ for $C \in \{A,B\}$ be unitary operators acting on $\mathcal{H}_C$. Suppose for $\epsilon > 0$ and all $C \in \{A,B\}, P,Q\in \{X,Y,Z\}$, the following conditions hold:
			\begin{align}
				\norm{\left(\hat{P}_A-(-1)^{\delta_{PY}} \hat{P}_B\right) \ket{\psi}}&\le \epsilon,  \label{eq:P1-system-indep-eps}\\
				\norm{\left(\{\hat{P}_C,\hat{Q}_C\}-2\delta_{P,Q}\mathbb{I}\right)\ket{\psi}} &\le\epsilon \label{eq:P2-Clifford-property-eps}.
			\end{align}
			Then, for two-dimensional qubit spaces $\mathcal{H}_{C'}$ and $\mathcal{H}_{C''}$, there exist local isometries $V_C: \mathcal{H}_C \to \mathcal{H}_{C}\otimes\mathcal{H}_{C'}\otimes\mathcal{H}_{C''}$ that satisfy the following properties:
			\begin{enumerate}
				\item Bell state extraction:
				\begin{equation}\label{eq:image-under-isometry}
					V_A\otimes V_B \ket{\psi} \approx_{\epsilon} \ket{\phi^+}_{A'B'} \ket{\xi}_{ABRA''B''}, 
				\end{equation}
				where $\ket{\xi} = \frac{1}{2\sqrt{2}}\left[\ket{00}_{A''B''}(\mathbb{I}+i\hat{Y}_A\hat{X}_A)+\ket{11}_{A''B''}(\mathbb{I}-i\hat{Y}_A\hat{X}_A)\right](\mathbb{I}+\hat{Z}_A)\ket{\psi}_{ABR}$.
				\item Pauli matrices up to local transpose:
				\begin{equation}\label{eq:intertwiner-claim}
					V_C \hat{P}_C \ket{\psi} \approx_{\epsilon} \left[(\sigma_{P})_{C'}\otimes (\sigma_Z)_{C''}^{\delta_{PY}}\right]V_C \ket{\psi}
				\end{equation}
			\end{enumerate}
		\end{lemma}
		\begin{proof}
			\begin{figure}[!htbp]
				\centering
				\includegraphics[width=0.5\linewidth]{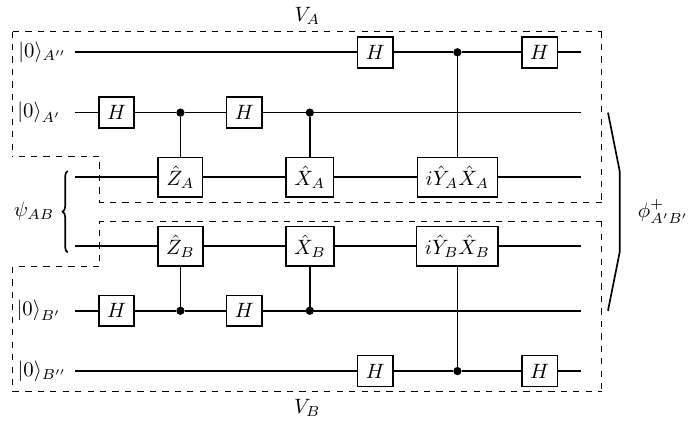}
				\caption{Local isometry from the extended SWAP test~\cite{Bowles2018DIAllEntangled} to extract a Bell state from the underlying state in the physical experiment. }
				\label{fig:extended-swap-iso}
			\end{figure}
			
			The isometries  $V_A$ and $V_B$ are explicitly constructed using the extended SWAP isometry given in~\cite{Bowles2018DIAllEntangled}. Its action is illustrated in Fig.~\ref{fig:extended-swap-iso}. 
			We now directly calculate the effect of the isometry $V_A\otimes V_B$ on $\ket{\psi}$ on the physical state $\ket{\psi}$: 
			\begin{align}
				&\quad \, \ket{0}_{A''} \ket{0}_{A'} \ket{\psi} \ket{0}_{B'} \ket{0}_{B''}  \nonumber \\
				\xrightarrow{H^{\otimes 4}}  &=  \frac{1}{4} \sum_{a'',a',b',b''} \ket{a''a'b'b''}_{A''A'B'B''}\ket{\psi}  \nonumber\\
				\xrightarrow{C'\hat{Z}}  &=  \frac{1}{4} \sum_{a'',a',b',b''} \ket{a''a'b'b''}_{A''A'B'B''} \hat{Z}_A^{a'}\hat{Z}_B^{b'} \ket{\psi} \nonumber \\
				&\approx_{\epsilon}  \frac{1}{4} \sum_{a'',a',b',b''} \ket{a''a'b'b''}_{A''A'B'B''} \hat{Z}_A^{a'+b'} \ket{\psi}  \tag{By Eq.~\eqref{eq:P1-system-indep-eps}}\\
				\xrightarrow{H_{A'}H_{B'}}  &= \frac{1}{8} \sum_{a'',b'',k_A,k_B} \ket{a''b''k_Ak_B}_{A''B''A'B'} \left(\sum_{a',b'} (-1)^{a'k_A+b'k_B}\hat{Z}_A^{a'+b'}\right) \ket{\psi} \label{eq:isometry-step-after-H-dash0}\\
				&= \frac{1}{4} \sum_{a'',b''}\ket{a''b''}_{A''B''} \left[\ket{00}_{A'B'}(\mathbb{I}+\hat{Z}_A)+\ket{11}_{A'B'}(\mathbb{I}-\hat{Z}_A)\right]\ket{\psi}  \label{eq:isometry-step-after-H-dash}\\
				\xrightarrow{C'\hat{X}}  &= \frac{1}{4} \sum_{a'',b''}\ket{a''b''}_{A''B''} \left[\ket{00}_{A'B'}(\mathbb{I}+\hat{Z}_A)+\ket{11}_{A'B'}\hat{X}_A\hat{X}_B(\mathbb{I}-\hat{Z}_A)\right]\ket{\psi} \nonumber\\
				&\approx_{\epsilon}  \frac{1}{4} \sum_{a'',b''}\ket{a''b''}_{A''B''} \left(\ket{00}_{A'B'}+\ket{11}_{A'B'}\right)(\mathbb{I}+\hat{Z}_A)\ket{\psi} \tag{By Eqs.~\eqref{eq:P1-system-indep-eps},\eqref{eq:P2-Clifford-property-eps}}  \\
				\xrightarrow{C''i\hat{Y}\hat{X}}  &= \frac{1}{4}\left(\ket{00}_{A'B'}+\ket{11}_{A'B'}\right)\sum_{a'',b''}\ket{a''b''}_{A''B''} (i\hat{Y}_A \hat{X}_A)^{a''}(i\hat{Y}_B \hat{X}_B)^{b''}(\mathbb{I}+\hat{Z}_A)\ket{\psi} \nonumber\\ 
				&\approx_{\epsilon}  \frac{1}{4}\left(\ket{00}_{A'B'}+\ket{11}_{A'B'}\right)\sum_{a'',b''}\ket{a''b''}_{A''B''} (i\hat{Y}_A \hat{X}_A)^{a''+b''}(\mathbb{I}+\hat{Z}_A)\ket{\psi} \tag{By Eqs.~\eqref{eq:P1-system-indep-eps},\eqref{eq:P2-Clifford-property-eps}} \\
				\xrightarrow{H_{A''}H_{B''}}   &=  \frac{1}{4}\left(\ket{00}_{A'B'}+\ket{11}_{A'B'}\right) \left[\ket{00}_{A''B''}(\mathbb{I}+i\hat{Y}_A\hat{X}_A)+\ket{11}_{A''B''}(\mathbb{I}-i\hat{Y}_A\hat{X}_A)\right](\mathbb{I}+\hat{Z}_A)\ket{\psi}. \nonumber 
			\end{align}
			This proves Eq.~\eqref{eq:image-under-isometry}.
			
			Analogously, we can trace the action of $V_A\otimes V_B$ on $\hat{P}_C\ket{\psi}$ for $C\in \{A,B\}$ to verify Eq.~\eqref{eq:intertwiner-claim}. We focus on $C=A$, as the analysis for $C=B$ is similar. For $P=Z$, applying $V_A\otimes V_B$ on $\hat{Z}_A\ket{\psi}$ results in the replacement in Eq.~\eqref{eq:isometry-step-after-H-dash}:
			\begin{equation}
				\begin{split}
					(\mathbb{I}+\hat{Z}_A)\ket{\psi} &\mapsto (\hat{Z}_A + \hat{Z}_A^2) \ket{\psi} \approx_{\epsilon} (\mathbb{I}+\hat{Z}_A)\ket{\psi},\\
					(\mathbb{I}-\hat{Z}_A)\ket{\psi} &\mapsto (\hat{Z}_A - \hat{Z}_A^2) \ket{\psi} \approx_{\epsilon} -(\mathbb{I}-\hat{Z}_A)\ket{\psi}.
				\end{split}
			\end{equation}
			This replacement gives the state $(\sigma_Z)_{A'} \ket{\phi^+}_{A'B'} = \frac{1}{\sqrt{2}}(\ket{00} - \ket{11})_{A'B'}$ on system $A'B'$. 
			
			Similarly, for $P=X$, Eq.~\eqref{eq:isometry-step-after-H-dash0} is changed to
			\begin{equation} 
				\begin{split}
					& \quad\, \frac{1}{8} \sum_{a'',b'',k_A,k_B} \ket{a''b''k_Ak_B}_{A''B''A'B'} \hat{X}_A\left(\sum_{a',b'} (-1)^{a'}(-1)^{a'k_A+b'k_B}\hat{Z}_A^{a'+b'}\right) \ket{\psi} \\
					&= \frac{1}{4} \sum_{a'',b''} \ket{a''b''}_{A''B''} \left[\ket{10}_{A'B'} \hat{X}_A (\mathbb{I}+\hat{Z}_A)+\ket{01}_{A'B'}\hat{X}_A (\mathbb{I}-\hat{Z}_A)\right]\ket{\psi} \\
					\xrightarrow{C'\hat{X}}&\approx_{\epsilon} \frac{1}{4} \sum_{a'',b''}\ket{a''b''}_{A''B''} \left(\ket{10}_{A'B'}+\ket{01}_{A'B'}\right)(\mathbb{I}+\hat{Z}_A)\ket{\psi}.
				\end{split}
			\end{equation} 
			This gives the state $(\sigma_X)_{A'} \ket{\phi^+}_{A'B'} = \frac{1}{\sqrt{2}}(\ket{10} + \ket{01})$ on system $A'$. 
			
			For $P=Y$, Eq.~\eqref{eq:isometry-step-after-H-dash} is changed to
			\begin{equation}
				\begin{split}
					&\quad\, \frac{1}{4} \sum_{a'',b''} \ket{a''b''}_{A''B''} \left[\ket{10}_{A'B'}\hat{Y}_A (\mathbb{I}+\hat{Z}_A)+\ket{01}_{A'B'}\hat{Y}_A (\mathbb{I}-\hat{Z}_A)\right]\ket{\psi} \\
					\xrightarrow{C'\hat{X}} & \approx_{\epsilon} -\frac{1}{4} \sum_{a'',b''} \ket{a''b''}_{A''B''}\left(\ket{10}_{A'B'}-\ket{01}_{A'B'}\right) \hat{Y}_A \hat{X}_A  (\mathbb{I}+\hat{Z}_A))\ket{\psi} \\
					\xrightarrow{C''i\hat{Y}\hat{X}} & \approx_{\epsilon}i\frac{1}{4} \sum_{a'',b''} \ket{a''b''}_{A''B''}\left(\ket{10}_{A'B'}-\ket{01}_{A'B'}\right) (i\hat{Y}_A\hat{X}_A)^{a''+b''+1}(\mathbb{I}+\hat{Z}_A)\ket{\psi} \\
					\xrightarrow{H_{A''}H_{B''}} & = i\frac{1}{4} \left(\ket{10}_{A'B'}-\ket{01}_{A'B'}\right) \left[\ket{00}_{A''B''}(\mathbb{I}+i\hat{Y}_A\hat{X}_A)-\ket{11}_{A''B''}(\mathbb{I}-i\hat{Y}_A\hat{X}_A)\right](\mathbb{I}+\hat{Z}_A)\ket{\psi} \\
					&= (\sigma_Y)_{A'}\otimes (\sigma_Z)_{A''} V_A\otimes V_B \ket{\psi}.
				\end{split}
			\end{equation}
			Combining the analysis above, we have shown that
			\begin{equation}
				(V_A\otimes V_B) \hat{P}_C \ket{\psi} \approx_{\epsilon} \left[(\sigma_{P})_{C'}\otimes (\sigma_Z)_{C''}^{\delta_{PY}}\right](V_A \otimes V_B) \ket{\psi}.
			\end{equation}
			This proves Eq.~\eqref{eq:intertwiner-claim}.
		\end{proof}
		
		\section{Self-testing multipartite Pauli measurements}\label{app:DI-multipartite-Pauli}
		
		This section presents the main technical contribution of this work: a self-testing scheme for performing multipartite Pauli measurements.
		The crucial feature of our protocol is its robustness to practical device imperfections and statistical fluctuations, which is essential for guaranteeing the sample efficiency of our subsequent multipartite self-testing procedures. 
		
		We build the full self-testing protocol by a sequence of incremental steps. 
		First, we define the ideal multipartite Pauli measurement protocols that our DI scheme aims to implement.
		Then, we leverage the robust self-testing results for bipartite Pauli measurements established in Sec.~\ref{app:3CHSH} to construct a self-testing protocol that implements the multipartite Pauli measurements, albeit only up to partial transpose (Lemma~\ref{lemma:DI-estimation-L-up-to-partial-transpose}).
		Recognizing that this partial transpose is insufficient for achieving sample-efficient self-testing for the global state, we devise a novel, more efficient self-testing protocol that correlates partial transposes across all subsystems (Lemma~\ref{cor:DI-estimation-L-up-to-global-conj}). 
		This protocol rigorously implements the target multipartite Pauli measurements up to global transpose, which, as we show in Section~\ref{app:self-testing-fin-samples}, is sufficient to guarantee polynomial sample efficiency for self-testing the global state up to global conjugation.
		Finally, we detail the construction of the local quantum channels to extract the desired target quantum state from the state $\psi$ present in the physical experiment.
		
		\subsection{Multipartite Pauli measurement protocols}\label{subsec:multipatite-pauli-msm-protocols}
		An $n$-partite Pauli measurement scheme is usually designed to estimate the expectation value of a target observable $L$ for an input state $\rho$. This procedure is defined by the following four steps:
		\begin{enumerate}
			\item Basis selection: A specific product Pauli basis, $\vec{P} = P_0 \otimes P_1 \otimes \cdots \otimes P_{n-1}$, is sampled according to a predetermined probability distribution $\mathcal{D}$ over the set $\{X, Y, Z\}^{\otimes n}$.
			\item Local measurement: Each party $l \in \{0, \dots, n-1\}$ performs a local Pauli measurement on their subsystem along the basis $P_l$.
			\item Result announcement and weighting: The parties record their collective measurement outcomes $\vec{b} = (b_0, b_1, \dots, b_{n-1}) \in \{0, 1\}^n$. A weight $\omega(\vec{b}|\vec{P})$ is assigned to this outcome, dependent on the selected basis $\vec{P}$.
			\item Expectation value estimation: By repeating steps (1)-(3) multiple times and averaging over both the basis choice $\vec{P}$ and the measurement outcome $\vec{b}$, an empirical estimate of the expectation value $\tr(L\rho)$ is obtained, where the observable $L$ is defined by Eq.~\eqref{eq:L-def}.
		\end{enumerate}
		In this section, we focus on the infinite-sample regime.
		In practice, statistical fluctuations due to finite sampling are crucial for determining sample complexity, which we rigorously address in Sec.~\ref{sec:sample_comp}.
		
		\subsection{Self-testing up to \emph{partial} transpose}\label{app:DI_pauli_partial_transpose}
		We now construct the self-testing protocol for realizing the $n$-partite Pauli measurement schemes defined above. The core strategy employs a local teleportation-and-measure approach~\cite{Supic2023NetworkSelftest}.
		However, a serious problem with such an approach is that the resulting DI realization suffers from a partial transpose problem~\cite{Bowles2018DIAllEntangled, Supic2023NetworkSelftest, sekatski2025entangledstatesnonlocalselftestable}. 
		Prior work~\cite{Supic2023NetworkSelftest, sekatski2025entangledstatesnonlocalselftestable} resolves this issue through full tomography on a genuine multipartite entangled state or through additional source assumptions. The former does not generalize to arbitrary states, and the latter introduces additional device-dependent assumptions. Importantly, the full tomography introduces exponential sample complexity in $n$.
		This exponential scaling presents a major obstacle to practical, scalable self-testing protocols. We address this partial-transpose problem in a sample-efficient, fully DI manner later in Sec.~\ref{app:global_tranpose}.
		Here, we focus on developing a robust DI protocol for performing multipartite Pauli measurements using the teleportation-and-measure approach up to the partial transposes.

		The key idea is to utilize the shared entanglement between each main party $A_l$ and its auxiliary party $B_l$ to effectively teleport the input state from the main parties $A=A_0A_1\dots A_{n-1}$ to the auxiliary parties $B=B_0B_1\dots B_{n-1}$. Once the state is teleported, the desired $n$-partite Pauli measurements can be performed by the parties $B$ using measurement operators rigorously self-tested (up to a local isometry) via the CHSH inequalities.
		
		Concretely, the main parties $A_l$ first perform the local measurements $\{M_{\vec{a}_l|\diamond}^{(l)}\}_{\vec{a}_l \in \{0,1\}^2}$. In the physical experiment, the measurements $\{M_{\vec{a}_l|\diamond}^{(l)}\}_{\vec{a}_l \in \{0,1\}^2}$ are completely uncharacterized, while in the ideal reference experiment, they implement a BSM on $T_lS_l$, teleporting the input state on $T_l$ to the auxiliary system $B_l$ via the Bell state shared on $S_lB_l$.
		In standard teleportation, one needs to perform a Pauli correction $U_{\vec{a_l}}$ after obtaining BSM outcome $\vec{a}_l$, where
		\begin{equation}
			U_{\vec{a}} = X^{a_0}Z^{a_1}, \quad \vec{a}=a_0a_1.
		\end{equation} 
		Crucially, the auxiliary parties $B_l$ do not physically apply this Pauli correction here. 
		Instead, they first perform the certified Pauli measurement $\{N_{b_l|P_l}^{(l)}\}_{b_l}$ along the basis $P_l \in \{X, Y, Z\}$ to obtain a raw outcome $b_l$. 
		Note that for convenience, we use $P_l = X,Y,Z$ to correspond to the measurement setting $y_l = 0, 1, 2$ (see Table~\ref{tab:experiment_setup}). 
		The necessary correction is then incorporated by adjusting the measurement outcome conditional on the BSM result $\vec{a}_l$ received from the measurement on $A_l$. 
		Specifically, the raw outcome $b_l$ is mapped to a corrected outcome $f(b_l|P_l,\vec{a}_l)$ by flipping $b_l$ when the correction operator $U_{\vec{a}_l}$ and the measurement basis $P_l$ anti-commute:
		\begin{equation}\label{eq:Pauli_correction}
			f(b_l|P_l,\vec{a}_l) = b_l \oplus \mathbf{1}[\{P_l, U_{\vec{a}_l}\} = 0].
		\end{equation}
		
		The overall estimator $v_{\mathrm{III}}$ for measuring the observable $L$ is then defined as the expectation over the basis distribution $\mathcal{D}$, where the weight $\omega$ is applied to the corrected outcome vector $f(\vec{b}|\vec{P},\{\vec{a}_l\}_{l=0}^{n-1})$:
		\begin{equation}\label{eq:v_III-protocol-specification}
			v_{\mathrm{III}} = \mathbb{E}_{\vec{P}\leftarrow \mathcal{D}}[\mathbb{E}[\omega(f(\vec{b}|\vec{P},\{\vec{a}_l\}_{l=0}^{n-1})|\vec{P})\mid\forall l, x_l={\diamond},y_l=P_l]].
		\end{equation}
		where $f(\vec{b}|\vec{P}, \{\vec{a_l}\}) = (f(b_0|P_0,\vec{a}_0),f(b_1|P_1,\vec{a}_1),\cdots, f(b_{n-1}|P_{n-1},\vec{a}_{n-1}))$ is the vector of corrected outcomes, and we identify $P_l = X,Y,Z$ to $y_l = 0, 1, 2$.
		
		We establish that, provided the auxiliary Pauli measurements are approximately certified via the 3-CHSH game, the estimator $v_{\mathrm{III}}$ approximately estimates $\tr(L\Gamma(\psi))$ up to partial transpose. 
		Here, $\Gamma(\psi)$ is the state after the teleportation and ideal recovery:
		\begin{equation}\label{eq:Gamma_channel}
			\Gamma:=\bigotimes_{l=0}^{n-1}\Gamma_l,\quad \Gamma_l:\rho\mapsto\sum_{\vec{a}_l\in\{0,1\}^2}U^{B_l'}_{\vec{a}_l}\tr_{A_lB_l}(M_{\vec{a}_l|\diamond}^{(l)} V_{B_l}\rho V_{B_l}^\dagger) U^{B_l'\dagger}_{\vec{a}_l}
		\end{equation}
		and $V_{B_l}$ is the isometry given in Lemma~\ref{lemma:pauli-upto-conj}.
		We summarize the above protocol in the following lemma.
		
		\begin{lemma}[Self-testing multipartite Pauli measurements up to partial transpose]\label{lemma:DI-estimation-L-up-to-partial-transpose}
			Let $L$ be the $n$-partite observable defined in Eq.~\eqref{eq:L-def}, with the weighting coefficients bounded by a constant $W>1$, i.e., $|\omega(\vec{b}|\vec{P})| \leq W$.
			Let $\Gamma$ be the quantum channel representing the teleportation followed by Pauli correction (Eq.~\eqref{eq:Gamma_channel}).
			Suppose that the deviation from the maximal CHSH value (Eq.~\eqref{eq:CHSHestimator}) for each party $l\in [n]$ and CHSH choice $k\in \{0,1,2\}$ is bounded by $v_\mathrm{I}(l,k) \leq \epsilon_l$.
			
			Then, the estimator $v_{\mathrm{III}}$ in the teleportation-and-measurement protocol (Eq.~\eqref{eq:v_III-protocol-specification}) approximates the expectation value of the target observable $L$ up to partial transpose:
			\begin{equation}\label{eq:thm-di-est-main-result-product-LOCC-partial-transpose}
				\left\lvert v_\mathrm{III}-\tr[\left(\sum_{\iota\in\{0,1\}^n}L^{T_\iota}_{B'}\otimes\ketbra{\iota}_{B''} \right)\Gamma(\psi)]\right\rvert \leq \cO\left(W\sum_{l=0}^{n-1}\sqrt{\epsilon_l}\right).
			\end{equation}
			Here, $(\cdot)^{T_\iota}$ denotes partial transpose applied to the subsystems $B'_l$ specified by the bit-string $\iota$. 
		\end{lemma}
		
		\begin{proof}
			We begin by expanding the definition of the estimator $v_{\mathrm{III}}$:
			\begin{equation}\label{eq:Exp-vIII-proof1}
				\begin{split}
					\vIII &= \bE_{\vec{P}\leftarrow \mathcal{D}}\left\{\sum_{\substack{\{\vec{a}_l\}_{l=0}^{n-1}\\\vec{b}\in\{0,1\}^n}}
					\omega(f(\vec{b}|\vec{P},\{\vec{a}_l\})|\vec{P}) \tr[M_{\vec{a}} \left(\bigotimes_{l=0}^{n-1} N^{(l)}_{b_l|P_l}\right) \psi]\right\} \\
					&= \bE_{\vec{P}\leftarrow \mathcal{D}}\left\{\sum_{\substack{\{\vec{a}_l\}_{l=0}^{n-1}\\\vec{b}\in\{0,1\}^n}}
					\omega(f(\vec{b}|\vec{P},\{\vec{a}_l\})|\vec{P}) \tr[M_{\vec{a}}V_B^{\dagger} \left(\bigotimes_{l=0}^{n-1} V_{B_l} N^{(l)}_{b_l|P_l}\right) \psi]\right\}.
				\end{split}
			\end{equation}
			Here, $M_{\vec{a}}:=\bigotimes_{l=0}^{n-1} M_{\vec{a}_l|\diamond}^{(l)}$ is the product of BSM-like operators, $V_B:= \bigotimes_{l=0}^{n-1} V_{B_l}$ is the product of isometries. The second line uses the identity $V_B^\dagger V_B = \mathbb{I}$.
			
			Define $\tau_{P,b} = \frac{1}{2}[\mathbb{I} + (-1)^b \sigma_P \otimes \sigma_Z^{\delta_{PY}}]$. Combining Lemma~\ref{lemma:chsh-implications} and Eq.~\eqref{eq:intertwiner-claim} in Lemma~\ref{lemma:pauli-upto-conj}, we have
			\begin{align}\label{eq:replacement-B-by-Pauli}
				V_{B_l} N^{(l)}_{b_l|P_l} \ket{\psi} \approx_{\sqrt{\epsilon_l}} (\tau_{P_l, b_l})_{B_l'B_l''} V_{B_l} \ket{\psi}.
			\end{align}
			Our objective is to replace the term $V_{B_l} N^{(l)}_{b_l|P_l}$ in Eq.~\eqref{eq:Exp-vIII-proof1} with $(\tau_{P_l,b_l})_{B_l'B_l''} V_{B_l}$ for all $l=0, \dots, n-1$. We perform this replacement incrementally.
			Let $v_j$ be the estimator value after the replacement has been performed for the first $j$ subsystems:
			\begin{equation}
				\begin{split}
					v_j &\coloneqq \bE_{\vec{P}\leftarrow \mathcal{D}}\left\{\sum_{\substack{\{\vec{a}_l\}_{l=0}^{n-1}\\\vec{b}\in\{0,1\}^n}}
					\omega(f(\vec{b}|\vec{P},\{\vec{a}_l\})|\vec{P}) \tr[M_{\vec{a}}V_B^{\dagger} \left(\bigotimes_{l=0}^{j-1} (\tau_{P_l,b_l})_{B_l'B_l''} V_{B_l}\right)\left(\bigotimes_{l=j}^{n-1} V_{B_l} N^{(l)}_{b_l|P_l}\right) \psi]\right\} \\
					&=  \bE_{\vec{P}\leftarrow \mathcal{D}}\left\{\sum_{\substack{\{\vec{a}_l\}_{l=0}^{n-1}\\\vec{b}\in\{0,1\}^n}}
					\omega(f(\vec{b}|\vec{P},\{\vec{a}_l\})|\vec{P}) p^{j}_{\vec{P}}(\{\vec{a}_l\}, \vec{b})\right\},
				\end{split}
			\end{equation}
			where $p^j_{\vec{P}}(\{\vec{a}_l\}, \vec{b}) \coloneqq \tr[M_{\vec{a}}V_B^{\dagger} \left(\bigotimes_{l=0}^{j-1} (\tau_{P_l,b_l})_{B_l'B_l''} V_{B_l}\right)\left(\bigotimes_{l=j}^{n-1} V_{B_l} N^{(l)}_{b_l|P_l}\right) \psi]$.

			By definition, $v_0 = v_{\mathrm{III}}$. 
			Now we aim to show that (i)  $\abs{v_{j+1}-v_j} \le \cO(W\sqrt{\epsilon_j})$; and (ii) $v_n = \tr[\left(\sum_{\iota\in \{0,1\}^n}L^{T_\iota}_{B'}\otimes\ketbra{\iota}_{B''} \right)\Gamma(\psi)]$. 
			Then, Eq.~\eqref{eq:thm-di-est-main-result-product-LOCC-partial-transpose} follows by expanding $v_n-v_0$ into a telescoping sum and applying triangular inequalities.
			
			We first prove (i). By substituting the definitions of $v_j$ and $v_{j+1}$, we have:
			\begin{equation}\label{eq:vj_distance}
				\begin{split}
					\abs{v_{j+1} - v_j} &= \bE_{\vec{P}\leftarrow \mathcal{D}}\left\{\sum_{\substack{\{\vec{a}_l\}_{l=0}^{n-1}\\\vec{b}\in\{0,1\}^n}}
					\omega(f(\vec{b}|\vec{P},\{\vec{a}_l\})|\vec{P}) [p^{j}_{\vec{P}}(\{\vec{a}_l\}, \vec{b}) - p^{j+1}_{\vec{P}}(\{\vec{a}_l\}, \vec{b})]\right\} \\
					&\le \bE_{\vec{P}\leftarrow \mathcal{D}}\left\{\sum_{\substack{\{\vec{a}_l\}_{l=0}^{n-1}\\\vec{b}\in\{0,1\}^n}}
					\abs{\omega(f(\vec{b}|\vec{P},\{\vec{a}_l\})|\vec{P})} \cdot \abs{p^{j}_{\vec{P}}(\{\vec{a}_l\}, \vec{b}) - p^{j+1}_{\vec{P}}(\{\vec{a}_l\}, \vec{b})}\right\} \\
					& \le 2W \bE_{\vec{P}\leftarrow \mathcal{D}} \left\{ \mathrm{TV}(p^{j}_{\vec{P}}, p^{j+1}_{\vec{P}})\right\},
				\end{split}
			\end{equation}
			where $\mathrm{TV}(\cdot,\cdot)$ denotes the total variational distance. 
			
			We define the sets $S_{>,0} \coloneqq \{(\{\vec{a}_l\}, \vec{b}) : p^j_{\vec{P}}(\{\vec{a}_l\}, \vec{b}) > p^{j+1}_{\vec{P}}(\{\vec{a}_l\}, \vec{b}), b_j = 0\}$, and $S_{>,1}, S_{<,0}, S_{<,1}$ analogously. Then, 
			\begin{equation}\label{eq:total_variational_pj}
				\begin{split}
					2\mathrm{TV}(p^{j}_{\vec{P}}, p^{j+1}_{\vec{P}}) = \Bigg(\sum_{(\{\vec{a}_l\}, \vec{b}_l) \in S_{>,0}} &+ \sum_{(\{\vec{a}_l\}, \vec{b}_l) \in S_{>,1}}\Bigg) (p^{j}_{\vec{P}}(\{\vec{a}_l\}, \vec{b}) - p^{j+1}_{\vec{P}}(\{\vec{a}_l\}, \vec{b})) \\
					+  &\left(\sum_{(\{\vec{a}_l\}, \vec{b}_l) \in S_{<,0}} + \sum_{(\{\vec{a}_l\}, \vec{b}_l) \in S_{<,1}}\right) (p^{j+1}_{\vec{P}}(\{\vec{a}_l\}, \vec{b}) - p^{j}_{\vec{P}}(\{\vec{a}_l\}, \vec{b})).
				\end{split}
			\end{equation}
			We now bound the first summation term $S_{>,0}$, and the other three terms follow similarly. 
			\begin{equation}
				\begin{split}
					\sum_{S_{>,0}} p^{j}_{\vec{P}}(\{\vec{a}_l\}, \vec{b}) - p^{j+1}_{\vec{P}}(\{\vec{a}_l\}, \vec{b}) &= \sum_{S_{>,0}} \tr[M_{\vec{a}}V_B^{\dagger} \left(\bigotimes_{l=0}^{j-1} (\tau_{P_l,b_l})_{B_l'B_l''} V_{B_l}\right) \left((\tau_{P_j, b_j})_{B_j'B_j''} V_{B_j} - V_{B_j} N^{(j)}_{b_j|P_j}\right)\left(\bigotimes_{l=j+1}^{n-1} V_{B_l} N^{(l)}_{b_l|P_l}\right) \psi] \\
					&= \sum_{S_{>,0}} \tr[M_{\vec{a}}V_B^{\dagger} \left(\bigotimes_{l=0}^{j-1} (\tau_{P_l, b_l})_{B_l'B_l''} V_{B_l}\right)\left(\bigotimes_{l=j+1}^{n-1} V_{B_l} N^{(l)}_{b_l|P_l}\right) \ket{\tilde{\psi}}\bra{\psi}].
				\end{split}
			\end{equation}
			Here, $\ket{\tilde{\psi}} \coloneqq \left((\tau_{P_j, b_j})_{B_j'B_j''} V_{B_j} - V_{B_j} N^{(j)}_{b_j|P_j}\right) \ket{\psi}$, which satisfies $\norm{\ket{\tilde{\psi}}} \le \sqrt{\epsilon_j}$ by Eq.~\eqref{eq:replacement-B-by-Pauli}. Define 
			\begin{equation}
				\Delta_j \coloneqq \sum_{S_{>,0}}M_{\vec{a}}V_B^{\dagger} \left(\bigotimes_{l=0}^{j-1} (\tau_{P_l, b_l})_{B_l'B_l''} V_{B_l}\right)\left(\bigotimes_{l=j+1}^{n-1} V_{B_l} N^{(l)}_{b_l|P_l}\right),
			\end{equation}
			Since $\Delta_j$ is a partial sum of measurement operators of a POVM, we have $\opnorm{\Delta_j} \leq 1$. Hence, 
			\begin{equation}
				\begin{split}
					\sum_{S_{>,0}} p^{j}_{\vec{P}}(\{\vec{a}_l\}, \vec{b}) - p^{j+1}_{\vec{P}}(\{\vec{a}_l\}, \vec{b}) &= \tr(\Delta_j \ket{\tilde{\psi}}\bra{\psi}) \\
					& \le \opnorm{\Delta_j} \norm{\ket{\tilde{\psi}}} \\
					& \le \sqrt{\epsilon_j}.
				\end{split}
			\end{equation}
			Substituting these bounds into Eqs.~\eqref{eq:total_variational_pj}, \eqref{eq:vj_distance} proves (i).
			
			We now proceed to prove (ii). By definition, $U_{\vec{a}_l}^{B_l'}  (\tau_{P_l,b_l})_{B_l'B_l''} U^{B_l'\dagger}_{\vec{a}_l} =  (\tau_{P_l,f_l})_{B_l'B_l''}$, where $f_l=f(b_l|P_l,\vec{a}_l)$. Therefore, 
			\begin{equation}\label{eq:vn_expression}
				\begin{split}
					v_n &= \bE_{\vec{P}\leftarrow \mathcal{D}}\left\{\sum_{\substack{\{\vec{a}_l\}_{l=0}^{n-1}\\\vec{b}\in\{0,1\}^n}}
					\omega(f(\vec{b}|\vec{P},\{\vec{a}_l\})|\vec{P}) \tr[M_{\vec{a}}V_B^{\dagger} \left(\bigotimes_{l=0}^{n-1} (\tau_{P_l,b_l})_{B_l'B_l''} \right)V_B \psi]\right\} \\
					&=  \bE_{\vec{P}\leftarrow \mathcal{D}}\left\{\sum_{\substack{\{\vec{a}_l\}_{l=0}^{n-1}\\\vec{b}\in\{0,1\}^n}}
					\omega(f(\vec{b}|\vec{P},\{\vec{a}_l\})|\vec{P}) \tr[M_{\vec{a}}\left(\bigotimes_{l=0}^{n-1} U^{B'_l}_{\vec{a}_l}(\tau_{P_l,b_l})_{B_l'B_l''} U_{\vec{a}_l}^{B'_l\dagger}\right) \left(\bigotimes_{l=0}^{n-1} U^{B'_l}_{\vec{a}_l}\right)V_B \psi V_B^{\dagger}\left(\bigotimes_{l=0}^{n-1} U_{\vec{a}_l}^{B'_l\dagger}\right)]\right\} \\
					&= \bE_{\vec{P}\leftarrow \mathcal{D}}\left\{\sum_{\substack{\{\vec{a}_l\}_{l=0}^{n-1}\\f_0f_1\cdots f_{n-1}\in\{0,1\}^n}}
					\omega(f_0f_1\cdots f_{n-1}|\vec{P}) \tr[M_{\vec{a}}\left(\bigotimes_{l=0}^{n-1} (\tau_{P_l,f_l})_{B_l'B_l''}\right) \left(\bigotimes_{l=0}^{n-1} U^{B'_l}_{\vec{a}_l}\right)V_B \psi V_B^{\dagger}\left(\bigotimes_{l=0}^{n-1} U_{\vec{a}_l}^{B'_l\dagger}\right)]\right\} \\
					&= \bE_{\vec{P}\leftarrow \mathcal{D}}\left\{\sum_{\substack{\vec{b}\in\{0,1\}^n}}
					\omega(\vec{b}|\vec{P}) \tr[\left(\bigotimes_{l=0}^{n-1} (\tau_{P_l,b_l})_{B_l'B_l''}\right) \Gamma(\psi)]\right\}
				\end{split}
			\end{equation}
			where we used that $M_{\vec{a}}$ acts only on the main systems $A_l$ and that $f(\circ|P_l,\vec{a}_l)$ is bijective.
			Moreover, 
			\begin{equation}
				\begin{split}
					&\quad \bE_{\vec{P}\leftarrow \mathcal{D}}\left\{\sum_{\substack{\vec{b}\in\{0,1\}^n}}
					\omega(\vec{b}|\vec{P}) \bigotimes_{l=0}^{n-1} (\tau_{P_l,b_l})_{B_l'B_l''} \right\} \\
					&= \bE_{\vec{P}\leftarrow \mathcal{D}}\left\{\sum_{\substack{\vec{b}\in\{0,1\}^n}}
					\omega(\vec{b}|\vec{P}) \bigotimes_{l=0}^{n-1}\left[\left(\frac{\mathbb{I} + (-1)^{b_l}\sigma_{P_l}}{2}\right)_{B_l'}\otimes \ketbra{0}_{B_l''}  + \left(\frac{\mathbb{I} + (-1)^{b_l+\delta_{PY}}\sigma_{P_l}}{2}\right)_{B_l'}\otimes \ketbra{1}_{B_l''}\right] \right\} \\
					&= \bE_{\vec{P}\leftarrow \mathcal{D}}\left\{\sum_{\substack{\vec{b}\in\{0,1\}^n}}
					\omega(\vec{b}|\vec{P}) \sum_{\iota\in\{0,1\}^n} \left[\bigotimes_{l=0}^{n-1} \left(\frac{\mathbb{I} + (-1)^{b_l}\sigma_{P_l}}{2}\right) \right]_{B'}^{T_{\iota}}\otimes \ketbra{\iota}_{B''} \right\} \\
					&= \sum_{\iota\in\{0,1\}^n} L^{T_\iota}_{B'}\otimes \ketbra{\iota}_{B''}. \\
				\end{split}
			\end{equation}
			Substituting this into Eq.~\eqref{eq:vn_expression} proves (ii). 
		\end{proof}

		\subsection{Self-testing up to \emph{global} transpose}\label{app:global_tranpose}
		We now remove the undesirable partial transpose terms in Lemma~\ref{lemma:DI-estimation-L-up-to-partial-transpose}, thereby yielding a robust DI protocol that implements multipartite Pauli measurements up to only an intrinsic global transpose.
		
		Recall that the partial transpose terms originate from the sum over $\iota \in S$, where $S\coloneqq \{0,1\}^n \backslash \{0^n, 1^n\}$ (i.e., the set of transpose configurations that are neither the identity nor the global transpose). If we can demonstrate that the probability of extracting these unwanted terms is negligible, specifically if
		$\tr\left[\left(\sum_{\iota \in S} \ketbra{\iota}_{B''} \right) \Gamma(\rho)\right] \approx 0$, then these components can be disregarded.
		
		Crucially, any configuration $\iota \in S$ must contain at least one adjacent pair of distinct partial transpose indicators (i.e., $\iota_l \neq \iota_{l+1}$ for some $l \in [n-1]$). Therefore, the overall removal condition reduces to proving that for any adjacent pair of parties $l$ and $l+1$, the probability of the outcomes $\ket{01}$ or $\ket{10}$ on the auxiliary systems $B_l''B_{l+1}''$ is approximately zero:
		\begin{equation}
			\tr\left[(\ketbra{01} + \ketbra{10})_{B''_lB''_{l+1}} \Gamma(\rho)\right] \approx 0.
		\end{equation}
		
		We achieve this by choosing another operator $K_{B_l'B_{l+1}'}$ with specific properties related to its partial transpose: its maximal eigenvalue $\lambda_{\max}(K)$ is significantly larger than the maximal eigenvalues of its partial transposes, $\lambda_{\max}(K^{T_{B_l'}})$ and $\lambda_{\max}(K^{T_{B_{l+1}'}})$.
		Specifically, we choose the two-qubit observable $K$ as:
		\begin{equation}
			K = \sigma_X \otimes \sigma_X - \sigma_Y \otimes \sigma_Y + \sigma_Z \otimes \sigma_Z.
		\end{equation}
		We observe the following crucial properties: (i) The maximal eigenvalue of $K$ is $\lambda_{\max}(K)=3$, which is attained by the Bell state $\phi^+$; (ii) The partial transpose of $K$ with respect to either subsystem is: 
		\begin{equation}
			K^{T_{B_l'}} = K^{T_{B_{l+1}'}} = \sigma_X \otimes \sigma_X + \sigma_Y \otimes \sigma_Y + \sigma_Z \otimes \sigma_Z, 
		\end{equation} 
		whose maximal eigenvalue is strictly smaller: $\lambda_{\max}(K^{T_{B_l'}}) = \lambda_{\max}(K^{T_{B_{l+1}'}}) = 1$.
		
		This strategy translates to the following procedure in the reference experiment:
		First, the main parties $A_l$ and $A_{l+1}$ share a Bell state on the designated link qubits $R_{l}^{\vartriangleright}$ and $R_{l+1}^{\vartriangleleft}$. 
		They then perform BSM on $R_{l}^{\vartriangleright}S_l$ and $R_{l+1}^{\vartriangleleft}S_{l+1}$ to teleport the shared Bell state to the auxiliary systems $B_l B_{l+1}$. 
		Finally, the auxiliary parties perform specific Pauli measurements designed to estimate the expectation value of $K$.
		
		In the physical experiment, this procedure corresponds to the main parties $A_l$ and $A_{l+1}$ performing their measurements $\{M_{\vec{a}_l | \vartriangleright}^{(l)}\}_{\vec{a}_l}$ and $\{M_{\vec{a}_{l+1} | \vartriangleleft}^{(l+1)}\}_{\vec{a}_{l+1}}$ respectively. The auxiliary parties $B_l$ and $B_{l+1}$ perform their certified Pauli measurements $\{N_{b_j|P_j}^{(j)}\}_{b_j}$ for $j=l,l+1$ and $P_j \in \{X,Y,Z\}$. Utilizing the Pauli corrections (Eq.~\eqref{eq:Pauli_correction}), they calculate the estimator $v_{\mathrm{II}}$ for the average two-party correlation:
		\begin{equation}\label{eq:v_II-protocol-specification}
			v_{\mathrm{II}}(l,l+1) = \sum_{P\in\{X,Y,Z\}}\frac{1}{3}\left[1-(-1)^{\mathbf{1}[P=Y]}\mathbb{E}\left[(-1)^{f(b_{l}|P,\vec{a}_l) + f(b_{l+1}|P,\vec{a}_{l+1})}\;\middle|\; x_l={\Dr}, x_{l+1}={\Dl}, y_l=y_{l+1}=P\right]\right],
		\end{equation}
		where the product bases $P_l P_{l+1}$ are uniformly drawn from the set $\{XX, YY, ZZ\}$. 
		
		The gap between $\lambda_{\max}(K)=3$ and $\lambda_{\max}(K^{T_{B_l'}})=1$ is directly reflected in the bounds of the estimator. Specifically, an ideal measurement estimating $K$ yields the minimal value, $v_{\mathrm{II}}(l,l+1) = 0$. Conversely, an expectation value corresponding to the presence of the unwanted local partial transpose $K^{T_{B_l'}}$ must yield $v_{\mathrm{II}}(l,l+1) \ge \frac{2}{3}$. This constant gap provides a mechanism to exclude partial transpose. The robust analysis is formalized in the following lemma.
		
		\begin{lemma}[Removing partial transpose]\label{lemma:II-implications-approximate}
			Let $\{V_{B_l}\}_{l=0}^{n-1}$ be the isometry defined in Lemma~\ref{lemma:pauli-upto-conj}.
			For any adjacent parties $l$ and $l+1$ and basis $k\in \{0,1,2\}$, suppose the deviation from the maximal CHSH value (Eq.~\eqref{eq:CHSHestimator}) is bounded by $v_\mathrm{I}(l,k)\leq \epsilon_{\mathrm{I},l}$ and $v_\mathrm{I}(l+1,k)\leq \epsilon_{\mathrm{I},l+1}$ for all $k\in \{0,1,2\}$, and the two-party estimator (Eq.~\eqref{eq:v_II-protocol-specification}) is bounded by $v_{\mathrm{II}}(l,l+1)\leq \epsilon_{\mathrm{II},(l,l+1)}$. 
			
			Then, the probability of the unequal partial transpose components between adjacent parties is bounded by:
			\begin{equation}\label{eq:correlating_nearby_parties}
				\tr[(\ketbra{01}+\ketbra{10})_{B_l''B_{l+1}''}V_{B_l}\otimes V_{B_{l+1}}\psi V_{B_l}^\dagger\otimes V_{B_{l+1}}^\dagger] = \cO\left(\epsilon_{\mathrm{II},(l,l+1)}+\sqrt{\epsilon_{\mathrm{I},l}} + \sqrt{\epsilon_{\mathrm{I},l+1}}\right).
			\end{equation}
		\end{lemma}

		\begin{proof}
			W.l.o.g., we prove the claim for $l = 0$.
			The two-party witness estimator $v_{\mathrm{II}}(0, 1)$ is a specific instance of the general Pauli measurement estimator $v_{\mathrm{III}}$ (Eq.~\eqref{eq:v_III-protocol-specification}) under the following substitutions: (i) The number of parties is $n=2$; (ii) The basis distribution $\mathcal{D}$ for $\vec{P} = P_0 P_1$ is uniform over $\{XX, YY, ZZ\}$; (iii) The weighting coefficient is $\omega(\vec{b}|\vec{P}) = 1-(-1)^{\delta_{P_0 Y}}(-1)^{b_0+b_1}$, bounded by $W = \cO(1)$.
			Moreover, the channel $\Gamma$ is defined by replacing the original BSM operators in Eq.~\eqref{eq:Gamma_channel} with the new ones:  $M^{(0)}_{\vec{a}_0|\diamond} \rightarrow M^{(0)}_{\vec{a}_0|\Dr}, M^{(1)}_{\vec{a}_{1}|\diamond} \rightarrow M^{(1)}_{\vec{a}_{1}|\Dl}$. 
			The associated observable, $L$, derived from this specific choice of weights and bases, is $L = \mathbb{I} - \frac{1}{3}K$.
			
			By Lemma~\ref{lemma:DI-estimation-L-up-to-partial-transpose},
			\begin{equation}\label{eq:applied-lemma-vIII-to-vII}
				\abs{\tr[\left(\sum_{\iota\in\{0,1\}^2}L_{B_0'B_1'}^{T_\iota}\otimes\ketbra{\iota}_{B_0''B_{1}''}\right)\Gamma(\psi)] - v_{\mathrm{II}}(0,1)} = \cO(\sqrt{\epsilon_{\mathrm{I},0}} + \sqrt{\epsilon_{\mathrm{I},1}}).
			\end{equation}
			Since $v_{\mathrm{II}}(0, 1) \le \epsilon_{\mathrm{II},(0, 1)}$, we have:
			\begin{equation}
				\tr[\left(\sum_{\iota\in\{0,1\}^2}L_{B_0'B_1'}^{T_\iota}\otimes\ketbra{\iota}_{B_0''B_{1}''}\right)\Gamma(\psi)]  = \cO\left(\epsilon_{\mathrm{II},(0,1)}+\sqrt{\epsilon_{\mathrm{I},0}} + \sqrt{\epsilon_{\mathrm{I},1}}\right).
			\end{equation}
			Moreover, we utilize the spectral properties of $L$: $\lambda_{\min}(L) = \lambda_{\min}(L^T) = 1-\tfrac{1}{3}\lambda_{\max}(K) = 0$,  and $\lambda_{\min}(L^{T_0}) = \lambda_{\min}(L^{T_1}) = 1 - \tfrac{1}{3}\lambda_{\max}(K^{T_0})= \tfrac{2}{3}$. Therefore, 
			\begin{equation}
				\tr[\left(\sum_{\iota\in\{0,1\}^2}L_{B_0'B_1'}^{T_\iota}\otimes\ketbra{\iota}_{B_0''B_{1}''}\right)\Gamma(\psi)] \ge \frac{2}{3} \tr[(\ketbra{01} + \ketbra{10})_{B_0''B_{1}''}\Gamma(\psi)]
			\end{equation}
			Combining the above two equations gives 
			\begin{equation}\label{eq:correlating_nearby_parties_Gamma}
				\tr[(\ketbra{01} + \ketbra{10})_{B_0''B_{1}''}\Gamma(\psi)] = \cO\left(\epsilon_{\mathrm{II},(0,1)}+\sqrt{\epsilon_{\mathrm{I},0}} + \sqrt{\epsilon_{\mathrm{I},1}}\right).
			\end{equation}
			Finally, the channel $\Gamma(\rho)$ (see~\eqref{eq:Gamma_channel}) can be decomposed as 
			\begin{equation}\label{eq:gamma-decomp-A-and-V-isometry}
				\Gamma(\rho) = \Gamma_A(V_{B_0}\otimes V_{B_{1}}\rho V_{B_0}^\dagger\otimes V_{B_{1}}^\dagger)
			\end{equation}
			for a channel $\Gamma_A$ that acts trivially on the $B''_0B''_1$ systems and hence commutes with $(\ketbra{01} + \ketbra{10})_{B_0''B_{1}''}$. 
			Then Eq.~\eqref{eq:correlating_nearby_parties} is proved by combining Eq.~\eqref{eq:correlating_nearby_parties_Gamma} with the trace-preserving property of $\Gamma_A$.
		\end{proof}
		
		As mentioned, removing partial transposes between all adjacent pairs of parties is sufficient to implement multipartite Pauli measurements up to global transpose. This result is summarized in the following lemma.

		\begin{lemma}[DI multipartite Pauli measurements up to global transpose]\label{cor:DI-estimation-L-up-to-global-conj}
			Let $L$ be the $n$-partite observable defined in Eq.~\eqref{eq:L-def}, with the weighting coefficients bounded by a constant $W>1$, i.e., $|\omega(\vec{b}|\vec{P})| \leq W$.
			Let $\Gamma$ be the quantum channel representing the teleportation followed by Pauli correction (Eq.~\eqref{eq:Gamma_channel}).
			Suppose that the deviation from the maximal CHSH value (Eq.~\eqref{eq:CHSHestimator}) for each party $l\in [n]$ in basis $k\in \{0,1,2\}$ is bounded by $v_\mathrm{I}(l,k) \leq \epsilon_{\mathrm{I},l}$, and the two-party estimator (Eq.~\eqref{eq:v_II-protocol-specification}) for $l\in[n-1]$  is bounded by $v_{\mathrm{II}}(l,l+1)\leq \epsilon_{\mathrm{II},(l,l+1)}$. 
			Let $\varepsilon = \sum_{j=0}^{n-2}\epsilon_{\mathrm{II},(j,j+1)}+\sum_{l=0}^{n-1}\sqrt{\epsilon_{\mathrm{I},l}}$.
			
			Then, the total probability of the non-identity and non-full transpose terms is negligible:
			\begin{equation}\label{eq:thm-di-est-main-result-trace-lb}
				\tr[(\ketbra{0}^{\otimes n} +\ketbra{1}^{\otimes n})_{B''}\Gamma(\rho)] = 1-\cO(\varepsilon)
			\end{equation}
			and the estimator $v_{\mathrm{III}}$ in the teleportation-and-measurement protocol (Eq.~\eqref{eq:v_III-protocol-specification}) approximates the expectation value of the target observable $L$ up to global transpose, specifically:
			\begin{equation}\label{eq:thm-di-est-main-result-up-to-global-conj}
				\left\lvert v_\mathrm{III}-\tr[\left(L_{B'}\otimes\ketbra{0}_{B''}^{\otimes n}+L_{B'}^*\otimes\ketbra{1}_{B''}^{\otimes n} \right)\Gamma(\rho)]\right\rvert = \cO(W\varepsilon).
			\end{equation} 
		\end{lemma}

		\begin{proof}
			Define $S_j \coloneqq \{\iota : \iota \in \{0,1\}^n, \iota_j \neq \iota_{j+1}\}$ for $j \in [n-1]$. Then, 
			\begin{equation}\label{eq:intermediate-result-cor-DI-est-L-global-conj}
				\begin{split}
					1 - \tr[(\ketbra{0}^{\otimes n} +\ketbra{1}^{\otimes n})_{B''}\Gamma(\rho)] & = \tr[\left(\sum_{\iota \neq 0^n, 1^n}\ketbra{\iota}_{B''}\right)\Gamma(\rho)] \\
					&\le \tr[\left(\sum_{j\in[n-1]}\sum_{\iota \in S_j}\ketbra{\iota}_{B''}\right)\Gamma(\rho)] \\
					& = \sum_{j \in [n-1]} \cO\left(\epsilon_{\mathrm{II}, (j, j+1)} + \sqrt{\epsilon_{\mathrm{I},j}} + \sqrt{\epsilon_{\mathrm{I},j+1}}\right).
				\end{split}
			\end{equation}
			Here, the last equation comes from Lemma \ref{lemma:II-implications-approximate} (also see~\eqref{eq:gamma-decomp-A-and-V-isometry}) and $\sum_{\iota \in S_j}\ketbra{\iota}_{B''} = \ketbra{01}_{B''_j,B''_{j+1}}+\ketbra{10}_{B''_j,B''_{j+1}}$. This proves Eq.~\eqref{eq:thm-di-est-main-result-trace-lb}.
			
			Moreover, 
			\begin{equation}\label{eq:intermediate-result-cor-DI-est-L-global-conj2}
				\begin{split}
					\abs{\tr[\left(L_{B'}\otimes\ketbra{0}_{B''}^{\otimes n}+L_{B'}^*\otimes\ketbra{1}_{B''}^{\otimes n} \right)\Gamma(\rho)] - \tr[\left(\sum_{\iota\in\{0,1\}^n}L^{T_\iota}_{B'}\otimes\ketbra{\iota}_{B''} \right)\Gamma(\rho)]} &= \abs{\tr[\left(\sum_{\iota \neq 0^n, 1^n}L^{T_\iota}_{B'}\otimes\ketbra{\iota}_{B''} \right)\Gamma(\rho)]} \\
					&\le W \tr[\left(\sum_{\iota \neq 0^n, 1^n}\ketbra{\iota}_{B''} \right)\Gamma(\rho)] \\
					&= W \sum_{j \in [n-1]} \cO\left(\epsilon_{\mathrm{II}, (j, j+1)} + \sqrt{\epsilon_{\mathrm{I},j}} + \sqrt{\epsilon_{\mathrm{I},j+1}}\right).
				\end{split}
			\end{equation}
			Here, the second equation comes from $\opnorm{L^{T_\iota}} \le \max\{\abs{\omega(\vec{b}|\vec{P})}\} \le W$ for any $\iota \in \{0,1\}^n$, and the third equation comes from Eq.~\eqref{eq:intermediate-result-cor-DI-est-L-global-conj}.
			Combining Eq.~\eqref{eq:intermediate-result-cor-DI-est-L-global-conj2} and Eq.~\eqref{eq:thm-di-est-main-result-product-LOCC-partial-transpose} in Lemma \ref{lemma:DI-estimation-L-up-to-partial-transpose} by triangular inequality proves Eq.~\eqref{eq:thm-di-est-main-result-up-to-global-conj}.
		\end{proof}

		\subsection{Summary of estimators}
		In Box~\ref{box:ProtocolEstimators} we summarize the quantities introduced in the preceding sections. The definition of the function $f$ is given in \eqref{eq:Pauli_correction}. Our subsequent results will make extensive use of the expectation values $\vI,\vII,\vIII$ as well as their finite-sample estimators, denoted by $\hvI,\hvII,\hvIII$.
		
		\begin{mybox}[label={box:ProtocolEstimators}]{Estimation protocols}
			\begin{enumerate}[label=(\Roman*)]
				\item\label{bullet:I} Self-testing Bell states and Pauli measurements via 3-CHSH: 
				For all configurations $l\in [n], k\in \{0,1,2\}$: Let $(x_l(0,k),x_l(1,k),y_l(0,k),y_l(1,k))$ be the one defined in Eq.~\eqref{eq:3-CSHS-inputs}.
				For $i,j\in \{0,1\}$ measure $A^{(l)}_{x_l(i,k)} \to a_l(i,k)$ and $B_{y_l(j,k)}^{(l)}\to b_l(j,k)$.
				Then set
				\begin{equation}
					v_\mathrm{I}(l,k) = 2\sqrt{2}-\sum_{i,j\in\{0,1\}}(-1)^{ij}\bE[(-1)^{a_l+b_l}\mid x_l=x_l(i,k),y_l=y_l(j,k)].
				\end{equation}
				
				\item\label{bullet:II}  Removing partial transpose:  For all configurations $l\in [n-1]$: Uniformly randomly sample $P\in \{X,Y,Z\}$.
				Measure $M^{(l)}_{\vec{a}_l|\Dr}\to \vec{a}_l$ and $M^{(l+1)}_{\vec{a}_{l+1}|\Dl}\to \vec{a}_{l+1}$ on system $A_l$ and $A_{l+1}$.
				Measure $P^{(l)}_{B}\to b_l$ and $P^{(l+1)}_{B} \to b_{l+1}$. 
				Then set
				\begin{equation}
					v_{\mathrm{II}}(l,l+1) = \mathbb{E}_{P\gets\{X,Y,Z\}}[1-(-1)^{\mathbf{1}[P=Y]}\bE[(-1)^{f(b_l|P,\vec{a}_l)}(-1)^{f(b_{l+1}|P,\vec{a}_{l+1})}\mid x_l={\Dr}, x_{l+1}={\Dl}, y_l=y_{l+1}=P]].
				\end{equation}
				
				\item\label{bullet:III}  Estimating $L$: Measure $M^{(l)}_{\vec{a}_l|\diamond}$ on all systems $A_l$, obtaining $\{\vec{a}_l\}_{l=0}^{n-1}$.
				Given a local measurement protocol defined through $\mathcal{D}$ and $\omega(\vec{b}|\vec{P})$  (see~\eqref{eq:L-def}), randomly sample $\vec{P}$ according to $\mathcal{D}$.
				Identify $X \mapsto X_B^{(l)},\ Y \mapsto Y_B^{(l)},\ Z \mapsto Z_B^{(l)}$ and measure the systems $B$ according to $\vec{P}$, obtaining a bit-string of outcomes $\vec{b}$.
				Set
				\begin{equation}
					v_{\mathrm{III}} = \mathbb{E}_{\vec{P}\gets\mathcal{D}}[\omega(f(\vec{b}|\vec{P},\{\vec{a}_l\}_{l=0}^{n-1})|\vec{P})\mid \forall l, x_l={\diamond},y_l=P_l].
				\end{equation}
			\end{enumerate}
		\end{mybox}
		
		\subsection{Refining the extraction channel $\Gamma$}\label{subsec:refinement-extraction-channel}
		The result established in Lemma~\ref{cor:DI-estimation-L-up-to-global-conj}, while guaranteeing DI multipartite Pauli measurements up to a global transpose, presents a drawback: the extraction channel $\Gamma = \bigotimes_l \Gamma_l$ is not a product of standard quantum channels acting only on the main parties $A_l$. This is because $\Gamma_l$ requires classical communication between $A_l$ and $B_l$ for the Pauli correction. 
		To develop a more conventional and practical self-testing protocol, it is highly desirable, yet challenging to construct an extraction channel $\Lambda=\bigotimes_l \Lambda_l$ that is a product of local quantum channels $\Lambda_l$, acting only on the main systems $A_l$ where the target state resides. This subtlety has not always been addressed in prior work \cite{Supic2023NetworkSelftest}.
		
		Here, we settle this problem by explicitly constructing such a channel. We summarize the result in the following theorem. Detailed proofs are deferred to the end of this section.
		\begin{theorem}[Self-testing multipartite Pauli measurements with local extraction channel]\label{thm:DI-estimation-L}
			Let $L$ be the $n$-partite observable defined in Eq.~\eqref{eq:L-def}, with the weighting coefficients bounded by a constant $W>1$, i.e., $|\omega(\vec{b}|\vec{P})| \leq W$.
			Suppose that the deviation from the maximal CHSH value (Eq.~\eqref{eq:CHSHestimator}) for each party $l\in [n]$ in basis $k\in \{0,1,2\}$ is bounded by $v_\mathrm{I}(l,k) \leq \epsilon_{\mathrm{I},l}$, and the two-party estimator (Eq.~\eqref{eq:v_II-protocol-specification}) for $l\in[n-1]$  is bounded by $v_{\mathrm{II}}(l,l+1)\leq \epsilon_{\mathrm{II},(l,l+1)}$.  
			Let $\varepsilon = \sum_{j=0}^{n-2}\epsilon_{\mathrm{II},(j,j+1)}+\sum_{l=0}^{n-1}\sqrt{\epsilon_{\mathrm{I},l}}$.

			Then, there exists a product of local quantum channels (with explicit construction given in Eq.~\eqref{eq:explict_construction_Lambda}):
			\begin{equation}
				\Lambda = \bigotimes_{l=0}^{n-1} \Lambda_l, \quad \Lambda_l: A_l \rightarrow A\ab_l A\abb_l,
			\end{equation}
			where $A\ab_l, A\abb_l$ are two-dimensional qubit spaces, such that the total probability of the non-identity and non-full transpose terms is negligible:
			\begin{equation}\label{eq:thm-Pauli-prob}
				\tr[(\ketbra{0}^{\otimes n} +\ketbra{1}^{\otimes n})_{A\abb}\Lambda(\rho)] = 1-\cO(\varepsilon),
			\end{equation}
			and the estimator $v_{\mathrm{III}}$ in the teleportation-and-measurement protocol (Eq.~\eqref{eq:v_III-protocol-specification}) approximates the expectation value of the target observable $L$ up to global transpose, specifically:
			\begin{equation}\label{eq:thm-Pauli-value}
				\left\lvert v_\mathrm{III}-\tr[\left(L_{A\ab}\otimes\ketbra{0}_{A\abb}^{\otimes n}+L_{A\ab}^*\otimes\ketbra{1}_{A\abb}^{\otimes n} \right)\Lambda(\rho)]\right\rvert = \cO(W\varepsilon).
			\end{equation} 
		\end{theorem}
		
		Now we proceed to construct the local channel $\Lambda_l$. Comparing the target Eq.~\eqref{eq:thm-Pauli-value} with the previous result in Eq.~\eqref{eq:thm-di-est-main-result-up-to-global-conj}, the task is to replace the auxiliary systems $B'_l$ and $B''_l$ with the new local systems $A\ab$ and $A\abb$.
		We leverage the properties of the local isometry $V_{A_l}V_{B_l}$ (given in Eq.~\eqref{eq:image-under-isometry}), which maps the input state $\ket{\psi}$ to an approximately factorized form:
		\begin{equation}\label{eq:distance_Vl_xi}
			V_{A_l}V_{B_l}\ket{\psi}_{ABR} \approx_{\sqrt{\epsilon_{\mathrm{I},l}}} \ket{\phi^+}_{A_l'B_l'}\ket{\xi_l}_{ABRA''_lB''_l},
		\end{equation}
		where 
		\begin{equation}
			\ket{\xi_l} = \frac{1}{2\sqrt{2}}\left[\ket{00}_{A''_lB''_l}(\mathbb{I}+i\hat{Y}_{A}^{(l)}\hat{X}_{A}^{(l)})+\ket{11}_{A''_lB''_l}(\mathbb{I}-i\hat{Y}_{A}^{(l)}\hat{X}_{A}^{(l)})\right](\mathbb{I}+\hat{Z}_{A}^{(l)})\ket{\psi}_{ABR}.
		\end{equation}
		Our strategy for constructing $\Lambda_l$ is based on the following ideas: 
		\begin{enumerate}
			\item To replace $B'_l$ with $A\ab_l$, we throw away the Bell state $\ket{\phi^+}_{A'_lB'_l}$, and locally prepare a Bell state $\ket{\phi^+}_{A'_lA\ab_l}$ to ``counterfeit'' the original one that was essential to the teleportation mechanism.  
			\item To replace $B''_l$ with $A\abb_l$, we introduce a fresh ancilla $\ket{0}_{A\abb_l}$ and apply $CX_{B''_lA\abb_l}$ to copy the information from $B''_l$ to $A\abb_l$.  
			\item By observing from $\ket{\xi_l}$ that the $A''_l$ register is perfectly correlated with $B''_l$, we can replace the non-local unitary $CX_{B''_lA\abb_l}$ with a local unitary $CX_{A''_lA\abb_l}$. 
		\end{enumerate}
		
		To formalize the conceptual steps outlined above, we introduce an intermediate channel $\Lambda_0$:
		\begin{equation}\label{eq:intermediate_channel}
			\Lambda_0 \coloneqq \bigotimes_{l=0}^{n-1} \Lambda_{l}^0, \quad \Lambda^0_l: \rho \rightarrow \mathcal{CX}_{A''_lA\abb_l}(\tr_{A_l'B_l'}[\mathcal{V}_{A_l} (\mathcal{V}_{B_l} (\rho))] \otimes \phi^+_{A'A\ab} \otimes \ketbra{0}_{A\abb_l}).
		\end{equation}
		Here, $\mathcal{V}: \rho \rightarrow V\rho V^{\dagger}$ denotes the channel representation of the corresponding isometry $V$, and $\mathcal{CX}$ is defined analogously. For compactness, we use $\mathcal{V}_A \coloneqq \bigotimes_{l=0}^{n-1} \mathcal{V}_{A_l}$ and similarly for $\mathcal{V}_B, \mathcal{CX}_{A''A\abb}, \ketbra{0}_{A\abb}, \phi^+_{A'A\ab}$.
		We show that $\Lambda_0$ successfully replaces $B'_l$ with $A\ab_l$, and copies information from $B''_l$ to $A\abb$, in the following sense. 
		\begin{lemma}[Replacing $B_l'B_l''$ with $A^{\ab}A^{\abb}$]\label{lem:replace_system_B}
			Suppose that the deviation from the maximal CHSH value (Eq.~\eqref{eq:CHSHestimator}) for each party $l\in [n]$ in basis $k\in \{0,1,2\}$ is bounded by $v_\mathrm{I}(l,k) \leq \epsilon_l$. 
			Then, the channel $\Lambda_0$ applied to the pure state $\psi$ satisfies:
			\begin{equation}\label{eq:replacing_B}
				D\left(\Lambda_0(\psi), \mathcal{CX}_{B''A\abb}(\phi^+_{A'A\ab} \otimes \xi_{ABRA''B''} \otimes \ketbra{0}_{A\abb})\right) = \cO\left( \sum_{l=0}^{n-1}\sqrt{\epsilon_l}\right),
			\end{equation}
			and 
			\begin{equation}
				\ket{\xi}_{ABRA''B''} \coloneqq \bigotimes_{l=0}^{n-1} \left\{\frac{1}{2\sqrt{2}}\left[\ket{00}_{A''_lB''_l}(\mathbb{I}+i\hat{Y}_{A}^{(l)}\hat{X}_{A}^{(l)})+\ket{11}_{A''_lB''_l}(\mathbb{I}-i\hat{Y}_{A}^{(l)}\hat{X}_{A}^{(l)})\right](\mathbb{I}+\hat{Z}_{A}^{(l)}) \right\}\ket{\psi}_{ABR}.
			\end{equation}
		\end{lemma}
		
		After the replacement, we recover the original physical state on $A_l$ by first applying the inverse channel $\mathcal{V}_{A_l}^{\dagger}$ to the system $A_lA_l'A_l''$, followed by the original BSM $\{M^{(l)}_{\vec{a}_l|\diamond}\}$. 
		To formalize $\mathcal{V}_{A_l}^{\dagger}$, note that the construction of $V_{A_l}$ in Fig.~\ref{fig:extended-swap-iso} allows us to express $\mathcal{V}_{A_l}$ as $\mathcal{V}_{A_l}(\rho) = \tilde{V} (\rho \otimes \ketbra{00}_{A_l'A_l''}) \tilde{V}^{\dagger}$ for a unitary $\tilde{V}$. 
		Accordingly, we define the inverse channel as  $\mathcal{V}_{A_l}^{\dagger} : \rho \mapsto \tr_{A_l'A_l''}(\tilde{V}^{\dagger}\rho \tilde{V})$, which guarantees that  $\mathcal{V}_{A_l}^{\dagger}\circ \mathcal{V}_{A_l}$ act as the identity channel.
		Crucially, the recovery unitary $U_{\vec{a}_l}$ now acts on $A\ab_l$ instead of $B_l'$. 
		Concretely, we introduce another channel $\Lambda_1$:
		\begin{equation}
			\Lambda_1 \coloneqq \bigotimes_{l=0}^{n-1}\Lambda_l^1, \quad \Lambda_l^1: \rho \rightarrow \sum_{\vec{a}_l\in  \{0,1\}^2} \mathcal{U}_{\vec{a}_l}^{A^{\ab}_l} \left( \tr_{A_l}[M_{\vec{a}_l|\diamond}^{(l)} \mathcal{V}_{A_l}^{\dagger} (\rho)]\right).
		\end{equation}
		We show that $\Lambda_1$ successfully extracts the wanted state from $\Lambda_0(\psi)$.
		
		\begin{lemma}[Teleportation on $AA^{\ab}$]\label{lem:teleportation_AAab}
			Suppose that the deviation from the maximal CHSH value (Eq.~\eqref{eq:CHSHestimator}) for each party $l\in [n]$ in basis $k\in \{0,1,2\}$ is bounded by $v_\mathrm{I}(l,k) \leq \epsilon_l$. Let $\Gamma$ be the quantum channel representing the teleportation followed by Pauli correction on $B'_l$ (Eq.~\eqref{eq:Gamma_channel}). Then, 
			\begin{equation}
				D\left(\Lambda_1 \circ \Lambda_0(\psi), \mathcal{CX}_{B''A^{\abb}}\left( \mathbb{I}_{B'\rightarrow A^{\ab}}[\Gamma(\psi)] \otimes \ketbra{0}_{A\abb}\right)\right) = \cO\left(\sum_{l=0}^{n-1}\sqrt{\epsilon_l}\right). 
			\end{equation}
			Here, $\mathbb{I}_{B'\rightarrow A\ab}$ replaces the system index $B'$ with $A\ab$, i.e.\ renames the system $B'$ to $A\ab$.
		\end{lemma}
		
		The quantum channel $\Lambda$ is then formally defined as 
		\begin{equation}\label{eq:explict_construction_Lambda}
			\Lambda\coloneqq \bigotimes_{l=0}^{n-1} \Lambda_{l}^1 \circ \tilde{\Lambda}_l^0, \quad  \tilde{\Lambda}_l^0: \rho \rightarrow \mathcal{CX}_{A''_lA\abb_l}(\tr_{A_l'}[\mathcal{V}_{A_l} (\rho)] \otimes \phi^+_{A'A\ab} \otimes \ketbra{0}_{A\abb_l})
		\end{equation}
		Comparing to $\Lambda_l^0$, $\tilde{\Lambda}_l^0$ replaces the partial trace $\tr_{A'_lB'_l}$ with $\tr_{A'_l}$, and removes $\mathcal{V}_{B_l}$.
		The channel $\Lambda_l^1 \circ \tilde{\Lambda}_l^0$ therefore operates entirely on the $A_l$, satisfying the condition that $\Lambda$ is a product of local channels on $A_l$. 
		We are now able to prove our main result, Theorem~\ref{thm:DI-estimation-L}.

		\begin{proof}[Proof of Theorem~\ref{thm:DI-estimation-L}]
			By the definition of $CX$ gate, for any class of observables $\{O_{\iota}\}_{\iota \in \{0,1\}^n}$ and state $\rho_{B'B''}$, we have
			\begin{equation}
				\tr[\left(\sum_{\iota} (O_{\iota})_{B'} \otimes \ketbra{\iota}_{B''} \right) \rho_{B'B''}] = \tr[\left(\sum_{\iota} (O_{\iota})_{B'} \otimes \ketbra{\iota}_{A\abb} \right) \mathcal{CX}_{B''A\abb}(\rho_{B'B''}\otimes \ketbra{0}_{A\abb})].
			\end{equation}
			Suppose $\opnorm{O_{\iota}} \le M$ for any $\iota$, we have 
			\begin{equation}\label{eq:trace_distance_Gamma_Lambda}
				\begin{split}
					\tr[ \left(\sum_{\iota} (O_{\iota})_{B'} \otimes \ketbra{\iota}_{B''} \right) \Gamma(\psi)] 
					&= \tr[ \left(\sum_{\iota} (O_{\iota})_{B'} \otimes \ketbra{\iota}_{A^{\abb}} \right)  \mathcal{CX}_{B''A^{\abb}}\left(\Gamma(\psi) \otimes \ketbra{0}_{A\abb}\right)] \\
					&= \tr[ \left(\sum_{\iota} (O_{\iota})_{A^{\ab}} \otimes \ketbra{\iota}_{A^{\abb}} \right)  \mathcal{CX}_{B''A^{\abb}}\left( \mathbb{I}_{B'\rightarrow A^{\ab}}[\Gamma(\psi)] \otimes \ketbra{0}_{A\abb}\right)] \\
					&\approx_{M \varepsilon_1} \tr[\left(\sum_{\iota} (O_{\iota})_{A^{\ab}} \otimes \ketbra{\iota}_{A^{\abb}} \right)  \Lambda_1 \circ \Lambda_0 (\psi)] \\
					&= \tr[\left(\sum_{\iota} (O_{\iota})_{A^{\ab}} \otimes \ketbra{\iota}_{A^{\abb}} \right)  \Lambda(\psi)].
				\end{split}
			\end{equation}
			Here, $\varepsilon_1 \coloneqq \sum_{l=0}^{n-1}\sqrt{\epsilon_{\mathrm{I},l}}$. The third line uses Lemma~\ref{lem:teleportation_AAab} and $\opnorm{\sum_{\iota} O_{\iota} \otimes \ketbra{\iota}} \le M$, and the fourth line uses the fact that $\mathcal{V}_B$ acting on $B$ system does not affect the expectation value of any observable on $A\ab A\abb$ system.
			
			Then, Eq.~\eqref{eq:thm-Pauli-prob} is proved by combining Eq.~\eqref{eq:trace_distance_Gamma_Lambda}, Eq.~\eqref{eq:thm-di-est-main-result-trace-lb} and $\opnorm{\ketbra{0}^{\otimes n} + \ketbra{1}^{\otimes n}} = 1$; Eq.~\eqref{eq:thm-Pauli-value} is proved by combining Eq.~\eqref{eq:trace_distance_Gamma_Lambda}, Eq.~\eqref{eq:thm-di-est-main-result-up-to-global-conj} and $\opnorm{\sum_{\iota} L^{T_\iota} \otimes \ketbra{\iota}} \le W$.
		\end{proof}
		
		\begin{proof}[Proof of Lemma \ref{lem:replace_system_B}]
			The proof proceeds by induction on the number of parties $j$. For compactness, we first establish the notation for the tensor product over the first $j$ parties.
			Let $\mathcal{V}_{A_{<j}}$ denote the tensor product of channels $\mathcal{V}_{A_l}$ for $l<j$, i.e., $\mathcal{V}_{A_{<j}} \coloneqq \bigotimes_{l=0}^{j-1} \mathcal{V}_{A_l}$. We define $\mathcal{V}_{B_{<j}}$, $\mathcal{CX}_{A''_{<j}A\abb_{<j}}$, $\ketbra{0}_{A\abb_{<j}}$, and $\phi^+_{A'_{<j}A\ab_{<j}}$ analogously. We further define
			\begin{equation}
				\begin{split}\label{eq:tau_big_state}
					\ket{\xi_{<j}}_{ABRA''_{<j}B''_{<j}}  &\coloneqq \bigotimes_{l=0}^{j-1} \left\{\frac{1}{2\sqrt{2}}\left[\ket{00}_{A''_lB''_l}(\mathbb{I}+i\hat{Y}_{A}^{(l)}\hat{X}_{A}^{(l)})+\ket{11}_{A''_lB''_l}(\mathbb{I}-i\hat{Y}_{A}^{(l)}\hat{X}_{A}^{(l)})\right](\mathbb{I}+\hat{Z}_{A}^{(l)}) \right\}\ket{\psi}_{ABR}, \\
					\tau_{j} &\coloneqq \mathcal{CX}_{A''_{<j}A\abb_{<j}}(\phi^+_{A'_{<j}A\ab_{<j}} \otimes {\xi_{<j, ABRA''_{<j}B''_{<j}}} \otimes \ketbra{0}_{A\abb_{<j}}).
				\end{split}
			\end{equation}
			We aim to show that for a certain constant $c_0 > 0$, 
			\begin{equation}\label{eq:induction}
				D\left(\bigotimes_{l=0}^{j-1}\Lambda_{l}^0(\psi),  \tau_j \right) \le c_0 \sum_{l=0}^{j-1} \sqrt{\epsilon_{l}}
			\end{equation}
			holds for any $0 < j \le n$. 
			
			For the base case $j=1$, we consider only the $l=0$ party. Because $D(\phi_1,\phi_2) \le \norm{\ket{\phi_1}-\ket{\phi_2}}$ for two pure states $\ket{\phi_1}, \ket{\phi_2}$, by Eq.~\eqref{eq:distance_Vl_xi},  
			\begin{equation}\label{eq:basic_case}
				D(\mathcal{V}_{A_0} (\mathcal{V}_{B_0} (\psi)), \phi^+_{A'_0B'_0} \otimes \xi_{0, ABRA''_0B''_0}) \le c_0\sqrt{\epsilon_0}, 
			\end{equation}
			for a certain constant $c_0$. Then, the base case $j=1$ is proved by combining Eq.~\eqref{eq:basic_case} and data processing.
			
			Assuming Eq.~\eqref{eq:induction} holds for $j=j_0$, we now prove it for $j_1 = j_0+1$:
			\begin{equation}
				\begin{split}
					D\left(\bigotimes_{l=0}^{j_1-1}\Lambda_{l}^0(\psi), \tau_{j_1}\right) 
					&\le D\left(\bigotimes_{l=0}^{j_1-1}\Lambda_{l}^0(\psi), \Lambda_{j_0}^0(\tau_{j_0})\right) +   D\left(\Lambda_{j_0}^0(\tau_{j_0}), \tau_{j_1}\right) \\
					&\le D\left(\bigotimes_{l=0}^{j_0-1}\Lambda_{l}^0(\psi), \tau_{j_0}\right) +   D\left(\mathcal{V}_{A_{j_0}} (\mathcal{V}_{B_{j_0}}(\psi)), \phi^+_{A'_{j_0}B'_{j_0}} \otimes \xi_{j_0, ABRA''_{j_0}B''_{j_0}}\right) \\
					&\le c_0 \sum_{l=0}^{j_1-1} \sqrt{\epsilon_l}.
				\end{split}
			\end{equation}
			where the first line uses triangular inequality, the second line uses data processing, the third line uses induction and Eq.~\eqref{eq:distance_Vl_xi}. 
			This establishes the claim for $j=j_1$. 
			By induction, Eq.~\eqref{eq:induction} holds for $j=n$.

			Finally, we can directly replace $\mathcal{CX}_{A''_{<j}A\abb_{<j}}$ in Eq.~\eqref{eq:tau_big_state} with $\mathcal{CX}_{B''_{<j}A\abb_{<j}}$ , as system $A''$ and $B''$ are symmetric in $\xi_{<j}$. Substituting $j=n$ into Eq.~\eqref{eq:induction} proves Eq.~\eqref{eq:replacing_B} .

		\end{proof}
		
		\begin{proof}[Proof of Lemma~\ref{lem:teleportation_AAab}]
			For compactness, we denote $\mathcal{U}_{\{\vec{a}_l\}}^{A^{\ab}} \coloneqq \bigotimes_{l=0}^{n-1} \mathcal{U}_{\vec{a}_l}^{A^{\ab}_l}$, and similarly for $M_{\{\vec{a}_l\}|\diamond}, \mathcal{V}_A^{\dagger}$. 
			By applying $\mathcal{V}_A^{\dagger}$, we have
			\begin{equation}
				\begin{split}
					\mathcal{V}_A^{\dagger}(\Lambda_0(\psi)) &\approx_{\varepsilon} \mathcal{V}_A^{\dagger} (\mathcal{CX}_{B''A\abb}(\phi^+_{A'A\ab} \otimes \xi_{ABRA''B''} \otimes \ketbra{0}_{A\abb}))\\
					&=  \mathcal{CX}_{B''A\abb}(\mathcal{V}_A^{\dagger}(\phi^+_{A'A\ab} \otimes \xi_{ABRA''B''}) \otimes \ketbra{0}_{A\abb}) \\
					&\approx_{\varepsilon} \mathcal{CX}_{B''A^{\abb}} \circ \mathbb{I}_{B'\rightarrow A^{\ab}} [\mathcal{V}_B(\psi)] \otimes \ketbra{0}_{A^{\abb}}.
				\end{split}
			\end{equation}
			Here, $\varepsilon \coloneqq \sum_{l=0}^{n-1} \sqrt{\epsilon_l}$,  and $\rho \approx_{\varepsilon} \sigma$ denotes $D(\rho,\sigma) = \cO(\varepsilon)$. The first line uses Lemma~\ref{lem:replace_system_B}, and the third line applies Eq.~\eqref{eq:distance_Vl_xi} sequentially for $l\in[n]$. 
			
			Then, by applying $M_{\{\vec{a}_l\}}$ and $\mathcal{U}_{\{\vec{a}_l\}}^{A^{\ab}}$, we have
			\begin{equation}
				\begin{split}
					\Lambda_1 \circ \Lambda_0(\psi) &\approx_{\varepsilon} \sum_{\{\vec{a}_l\}_{l=0}^{n-1}} \mathcal{U}_{\{\vec{a}_l\}}^{A\ab}\left(\tr_{A_l}[M_{\{\vec{a}_l|\diamond \}} \mathcal{CX}_{B''A^{\abb}} \circ \mathbb{I}_{B'\rightarrow A^{\ab}} [\mathcal{V}_B(\psi)] \otimes \ketbra{0}_{A^{\abb}}]\right) \\
					&= \mathcal{CX}_{B''A^{\abb}} \left( \mathbb{I}_{B'\rightarrow A^{\ab}} \left\{\sum_{\{\vec{a}_l\}_{l=0}^{n-1}} \mathcal{U}_{\{\vec{a}_l\}}^{B'}\left(\tr_{A_l}[M_{\{\vec{a}_l\}|\diamond} \mathcal{V}_B(\psi) ]\right)\right\} \otimes \ketbra{0}_{A\abb}\right) \\
					&= \mathcal{CX}_{B''A^{\abb}} \left( \mathbb{I}_{B'\rightarrow A^{\ab}}[\Gamma(\psi)] \otimes \ketbra{0}_{A\abb}\right),
				\end{split}
			\end{equation}
			where the third line uses the definition of $\Gamma$. This completes the proof.
		\end{proof}
		
		\section{Sample complexity analysis}\label{sec:sample_comp}
		Having established self-tested Pauli measurement schemes based on the accurate values of $\vI$, $\vII$, and $\vIII$ (Box~\ref{box:ProtocolEstimators} and Theorem~\ref{thm:DI-estimation-L}), we now analyze a crucial step for practical implementation and scalability: the sample complexity required for robustly estimating these estimators in a physical experiment using finite samples.
		
		A central challenge arises in estimating the multi-party estimator $\vIII$. The protocol requires the spatially separated main parties $A_l$ to collectively teleport the entire $n$-partite state to the ancillary qubits $B_l$. Subsequently, the ancillary parties $B_l$ must sample a product Pauli basis $\vec{P}$ from the probability distribution $\mathcal{D}$ that defines the observable $L$. This collaborative sampling becomes challenging when the target distribution $\mathcal{D}$ is correlated across the individual systems.
		
		Here, we address the sample complexity under two distinct scenarios, yielding two main sample complexity theorems:
		\begin{enumerate}
			\item No shared randomness (local sampling): The spatially separated parties use their local randomness to choose measurement settings independently from each other. While even correlated distributions $\mathcal{D}$ can be sampled in this manner through rejection sampling, we consider a target distribution $\mathcal{D}$ that can be decomposed into a product of local probability distributions (Theorem~\ref{thm:fin-sample-DI-estimation-L-local-randomness}).
			\item Shared randomness (collaborative sampling): The spatially separated parties have access to shared randomness, which could be distributed using DI quantum key distribution protocols or similar means. They utilize this resource to sample the measurement settings $\vec{P}$ collaboratively, enabling the efficient implementation of non-local probability distributions $\mathcal{D}$ (Theorem~\ref{thm:fin-sample-DI-estimation-L-shared-randomness}).
		\end{enumerate}
		
		\subsection{General case}
		We proceed by first considering the most general case, where an arbitrary joint probability distribution governs the choice of measurement settings across all $n$ parties. Subsequently, we will apply this general analysis to derive the specific sample complexity bounds for the two practical scenarios mentioned above.

		We will use standard concentration inequalities in the following formulation:
		\begin{fact}[Chernoff bound]\label{fact:chernoff}
			Let $T>0$ and $\{X_i\}_{i=1}^T$ be a set of $T$ i.i.d.\ binary random variables with $\bE[X_i]=p$. Then, with probability at least $1-\delta$, we have $\sum_{i=1}^T X_i \ge N$, provided that $T= \Theta\left(\frac{1}{p}(N+\log\frac{1}{\delta})\right)$.
		\end{fact}

		\begin{fact}[Hoeffding inequality]\label{fact:hoeffding}
			Let $T>0$ and $\{X_i\}_{i=1}^T$ be a set of $T$ i.i.d.\ random variables with $\bE[X_i] = \mu$ and $\abs{X_i}\le \sigma$. Then, with probability at least $1-\delta$, we have $\abs{\frac{1}{T}\sum_{i=1}^T X_i - \mu} \le \epsilon$, provided that $T= \Theta\left(\frac{\sigma^2}{\epsilon^2}\log\frac{1}{\delta}\right)$.
		\end{fact}
		
		We first analyze the sample complexity for estimating $\vI$ and $\vII$.  We use $\Pmsm[\vec{s};\vec{S}]$ to denote the probability that the measurement settings specified by the vector $\vec{s}$ are chosen on the systems specified by $\vec{S}$. 
		
		\begin{lemma}[Sample complexity for estimating $\vI$ and $\vII$]\label{lemma:sample-complexity-N_I-and-N_II}
			Let $N_{\mathrm{I}}$ be the sample complexity to jointly estimate $\vI(l,k)$ up to additive error $\epsilon_{\mathrm{I}}>0$ with joint failure probability $\delta>0$ for all $l\in [n], k\in [3]$. Then
			\begin{equation}
				N_{\mathrm{I}} = \cO\Bigg(\frac{1}{\min\limits_{l\in [n], k \in [3], i,j \in \{0,1\}}\left\{\Pmsm[x_l(i,k), y_l(j,k); A_l, B_l]\right\}}\frac{\operatorname{log}(\frac{n}{\delta})}{\epsilon_{\mathrm{I}}^2}\Bigg) \label{eq:N_I-asym}.
			\end{equation}
			Let $N_{\mathrm{II}}$ be the asymptotic sample complexity to jointly estimate $\vII(l,l+1)$ up to additive error $\epsilon_{\mathrm{II}}>0$ with joint failure probability $\delta>0$ for all $l\in [n-1]$. Then
			\begin{equation}
				N_{\mathrm{II}} = \cO\Bigg(\frac{1}{\min\limits_{y\in [3], l\in [n-1]}\left\{\Pmsm[\Dr, \Dl, y,y; A_l,A_{l+1}, B_l,B_{l+1}]\right\}}\frac{\log(\frac{n}{\delta})}{\epsilon_{\mathrm{II}}^2}\Bigg).\label{eq:N_II-asym}
			\end{equation}
		\end{lemma}
		
		\begin{proof}
			Let $N_{J;\vec{c}}(\epsilon', \delta')$ denote the number of samples to estimate $v_{J}(\vec{c})$ up to $\epsilon'$ with failure probability $\delta'$ for a given configuration $J$ and arguments $\vec{c}$ (Box~\ref{box:ProtocolEstimators}).
			Using Facts \ref{fact:chernoff} and \ref{fact:hoeffding} we get:
			\begin{equation}
				\begin{split}
					N_{\mathrm{I}; l,k}(\epsilon', \delta') &= \cO\left(\frac{1}{\min\limits_{i,j \in \{0,1\}}\left\{\Pmsm[x_l(i,k),y_l(j,k);A_l,B_l]\right\}}\frac{\operatorname{log}(\frac{1}{\delta'})}{\epsilon'^2}\right) \\
					N_{\mathrm{II}; l}(\epsilon', \delta') &= \cO\left(\frac{1}{\min\limits_{y\in [3]}\left\{\Pmsm[\Dr, \Dl, y,y; A_l,A_{l+1}, B_l,B_{l+1}]\right\}}\frac{\log(\frac{1}{\delta'})}{\epsilon'^2}\right).
				\end{split}
			\end{equation}
			Let $\delta' = \Theta\left(\tfrac{\delta}{n}\right)$. The joint failure probability can be obtained by applying a union bound over $l\in [n]$ and $k\in [3]$, which completes the proof.
		\end{proof}
		
		We denote $\Pmsm[\diamond^n; A\land \vec{y}\leftarrow \mathcal{D}]$ as the overall probability that a measurement round contributes to the estimation of $\vIII$. 
		Specifically, it is the probability that the measurement settings on the main parties $A_0A_1\cdots A_{n-1}$ are all set to the teleportation basis $\diamond$, and that the measurement settings $\vec{y}$ on the auxiliary system $B_0B_1\cdots B_{n-1}$ can be treated as being sampled from the distribution $\mathcal{D}$ (e.g., if using rejection sampling, the probability not to reject a sample). 
		Then, the sample complexity for estimating $\vIII$ follows similarly.
		
		\begin{lemma}[Sample complexity for estimating $\vIII$]\label{lem:sample-complexity-N_III}
			Let $L$ be the $n$-partite observable defined in Eq.~\eqref{eq:L-def}, with the weighting coefficients bounded by a constant $W>1$, i.e., $|\omega(\vec{b}|\vec{P})| \leq W$.
			Let $N_{\mathrm{III}}$ be the sample complexity to estimate $\vIII$ up to additive error $\epsilon_{\mathrm{III}}>0$ with failure probability $\delta>0$. Then
			\begin{equation}\label{eq:N_III-asym}
				N_{\mathrm{III}} = \cO\left(\frac{1}{\Pmsm[\diamond^n; A\land \vec{y}\leftarrow \mathcal{D}]}\frac{W^2}{\epsilon_{\mathrm{III}}^2}\log(\frac{1}{\delta})\right).
			\end{equation}
		\end{lemma}
		\begin{proof}
			
			Recall from Eq.\eqref{eq:v_III-protocol-specification} that $\vIII = \mathbb{E}_{\vec{P}\leftarrow \mathcal{D}}[\mathbb{E}[\omega(f(\vec{b}|\vec{P},\{\vec{a}_l\}_{l=0}^{n-1})|\vec{P})]]$.
			By Fact~\ref{fact:hoeffding}, the required number of measurement rounds that contribute to estimating $\vIII$ is
			$\cO\left(\frac{W^2}{\epsilon_{\mathrm{III}}^2}\log\left(\frac{1}{\delta}\right)\right)$. By the definition of $\Pmsm[\diamond^n; A\land \vec{y}\leftarrow \mathcal{D}]$ and Fact~\ref{fact:chernoff}, the required number of total measurement rounds is given by Eq.\eqref{eq:N_III-asym}.
		\end{proof}
		
		After establishing the general sample complexity analysis, we now proceed to analyze the detailed sample complexity in concrete scenarios. This is achieved by devising the probability distribution for the measurement settings and substituting it into the above analysis.
		
		\subsection{Using local randomness}\label{subsec:no-shared-randomness}
		We first consider the scenario in which each party uses local randomness to sample their measurement settings independently.
		
		In this scenario, we need to ensure that each main party chooses to teleport the input state to the auxiliary system with a noticeable probability. This is essential to prevent an unfavorable scaling in the estimation of $\vIII$ (Lemma~\ref{lem:sample-complexity-N_III}) due to $\Pmsm[\diamond^n; A\land \vec{y}\leftarrow \mathcal{D}]$ being too small.
		To achieve this, each main party $A_l$ selects the teleportation setting $\diamond$ with a high probability, for example:
		\begin{equation}
			\Pmsm[\diamond; A_l] = 1 - \frac{1}{n}.
		\end{equation}
		Then, the probability that all main parties simultaneously choose the teleportation setting $\diamond$ is larger than a constant:
		\begin{equation}\label{eq:alldiamond-prob-Omega1}
			\Pmsm[\diamond^n; A]  = \left(1 - \frac{1}{n}\right)^n = \Omega(1).
		\end{equation}
		Additionally, to ensure a noticeable probability for the 3-CHSH test and the coordinating partial transpose settings, the main parties sample the remaining $8$ measurement settings with equal probability:
		\begin{equation}
			\Pmsm[i; A_l] = \frac{1}{8n} \quad \text{for } i \in \{0, 1, 2, 3, 4, 5, \Dr, \Dl\}.
		\end{equation}
		
		Furthermore, we require that the sampled measurement settings $\vec{y}$ from the auxiliary parties be effectively sampled from the target distribution $\mathcal{D}$ of the desired Pauli measurements. However, a general distribution $\mathcal{D}$ might necessitate correlated sampling among the auxiliary parties, which could result in the probability $\Pmsm[\diamond^n; A \land \vec{y} \leftarrow \mathcal{D}]$ becoming exponentially small, thereby rendering the self-testing protocol inefficient (Lemma~\ref{lem:sample-complexity-N_III}).
		To avoid this, we focus on Pauli measurements where the desired distribution $\mathcal{D}$ allows for independent sampling across the different auxiliary systems $l$. That is, $\mathcal{D}$ admits a product distribution:
		\begin{equation}
			\mathcal{D} = \bigotimes_{l=0}^{n-1} \mathcal{D}^{(l)}, \quad \mathcal{D}^{(l)} = (p_X^{(l)},p_Y^{(l)},p_Z^{(l)}).
		\end{equation}
		In this case, each auxiliary party $B_l$ can simply sample their measurement setting according to the marginal distribution $\mathcal{D}_l$ independently. We additionally impose that $p_X^{(l)}, p_Y^{(l)}, p_Z^{(l)} = \Omega(1)$ for all $l$, ensuring that the probability of sampling the specific settings required for estimating $\vI$ and $\vII$ is sufficiently large.
		
		Crucially, because the marginal probabilities $p_X^{(l)}, p_Y^{(l)}, p_Z^{(l)}$ are generally not equal, the joint measurement settings $y_l = y_{l+1} \in \{X, Y, Z\}$ required for evaluating $\vII$ are sampled with varying probabilities. However, the original definition of $\vII$ relies on a uniform distribution over these three joint settings. To correct this discrepancy and construct an unbiased estimator, we apply a rejection sampling procedure. Specifically, to estimate $\vII$ for a given $l$, we accept a sample with an observed joint setting $y$ with probability
		\begin{equation}
			\frac{\min_{y'\in[3]}{\Pmsm[\Dr, \Dl, y',y'; A_l,A_{l+1}, B_l,B_{l+1}]}}{\Pmsm[\Dr, \Dl, y,y; A_l,A_{l+1}, B_l,B_{l+1}]}.
		\end{equation}
		This effectively equalizes the sampling rates, ensuring that each matching pair $y_l = y_{l+1}$ contributes to the final update with the same uniform probability $\min_{y'\in[3]}\{\Pmsm[\Dr, \Dl, y',y'; A_l,A_{l+1}, B_l,B_{l+1}]\}$.
		This guarantees an unbiased estimator for $\vII$, which is implemented as in \eqref{eq:rejection-sample-prot1}.
		
		The protocol using local randomness is summarized in Protocol~\ref{prot:di-witness-local-randomness}. The following theorem guarantees its performance:
		
		\begin{theorem}[Finite sample version of Theorem~\ref{thm:DI-estimation-L} using local randomness]\label{thm:fin-sample-DI-estimation-L-local-randomness}
			Let $L$ be the $n$-partite observable defined in Eq.~\eqref{eq:L-def}, with the weighting coefficients bounded by a constant $W>1$, i.e., $|\omega(\vec{b}|\vec{P})| \leq W$.
			Furthermore, suppose the target measurement distribution $\mathcal{D}$ for $L$ is a product distribution $\mathcal{D} = \bigotimes_{l=0}^{n-1}\mathcal{D}^{(l)}$, and for any party $l$, the marginal probabilities satisfy $p_X^{(l)}, p_Y^{(l)}, p_Z^{(l)}  =\Omega(1)$.
			
			For a desired precision $\varepsilon > 0$ and failure probability $0 < \delta < 1$, if the total number of samples $N$ is set to:
			\begin{equation}\label{eq:claim-sample-complexity-no-shared-randomness}
				N = \cO\left(W^4n^5\frac{\log(\frac{n}{\delta})}{\varepsilon^4}\right).
			\end{equation}
			Then, with probability at least $1-\delta$, if Protocol~\ref{prot:di-witness-local-randomness} outputs \texttt{CERTIFIED}, there exist $p \in [0,1]$, a state $\tau = p \tau_+ + (1-p)\tau_-$, and a product of local channels $\Gamma$ such that:
			\begin{equation}\label{eq:state-closeness}
				D\left(p \tau_+ \otimes \ketbra{0}^{\otimes n} + (1-p)\tau_-^* \otimes \ketbra{1}^{\otimes n}, \Gamma(\rho)\right) \le \cO\left(\frac{\varepsilon}{W}\right),
			\end{equation}
			and the returned estimator $\hvIII$ satisfies:
			\begin{equation}\label{eq:estimator_error_independent_measurement}
				| \hvIII - \tr(L\tau) | = \cO(\varepsilon).
			\end{equation}
		\end{theorem}
		
		\begin{protocol}{DI multipartite Pauli measurements using local randomness}{prot:di-witness-local-randomness}
			\textbf{Input: } 
			\begin{enumerate}
				\item Total number of measurement rounds $N$.
				\item Target observable $L$ defined in Eq.~\eqref{eq:L-def}, with weighting coefficients $\omega(\vec{b}|\vec{P})$ satisfying $|\omega(\vec{b}|\vec{P})| \leq W$. 
				\item Target measurement distribution $\mathcal{D} = \bigotimes_{l=0}^{n-1}\mathcal{D}^{(l)}$ of $L$, where $\mathcal{D}^{(l)} = (p_X^{(l)}, p_Y^{(l)}, p_Z^{(l)})$ are the marginal probabilities for party $B_l$. 
			\end{enumerate}
			
			\textbf{Initialization}
			
			Set all counters and estimators to zero:
			\begin{enumerate}
				\item $n_{\mathrm{III}} \leftarrow 0, e_{\mathrm{III}} \leftarrow 0$.
				\item For $l\in[n], k\in\{0,1,2\}, i,j\in\{0,1\}$:  $n_{\mathrm{I}}(l,k,i,j) \leftarrow 0, n_{\mathrm{II}}(l,l+1) \leftarrow 0$, $e_{\mathrm{I}}(l,k,i,j) \leftarrow 0, e_{\mathrm{II}}(l,l+1) \leftarrow 0$.
			\end{enumerate}
			
			\textbf{Measurement rounds}
			
			Execute the following steps for $N$ total rounds:
			\begin{enumerate}
				\item Setting selection: Each party independently chooses their measurement setting:
				\begin{enumerate}[label=(\arabic*)]
					\item 
					Party $A_l$ chooses $x_l \in \{0, \dots, 5, \diamond, \Dr, \Dl\}$ with probabilities:
					\begin{equation}
						\begin{split}
							\Pr[x_l = \diamond] &= 1 - \frac{1}{n} \\
							\Pr[x_l = i] &= \frac{1}{8n} \quad \text{for } i \in \{0, 1, 2, 3, 4, 5, \Dr, \Dl\}
						\end{split}
					\end{equation}
					\item Party $B_l$ chooses $y_l \in \{0, 1, 2\}$ with probabilities $\Pr[y_l=0]=p_X^{(l)}, \Pr[y_l=1]=p_Y^{(l)}, \Pr[y_l=2]=p_Z^{(l)}$.
				\end{enumerate}
				
				\item Measurement: Party $A_l$ performs $M^{(l)}_{\vec{a}_l|x_l}$ and obtains outcome $\vec{a}_l$; Party $B_l$ performs $N^{(l)}_{b_l|y_l}$ and obtains outcome $b_l$.
				
				\item Counter updates: A verifier receives all settings $\{x_l, y_l\}_l$ and outcomes $\{\vec{a}_l, b_l\}_l$ and updates counters based on the sampled settings.
				\begin{enumerate}[label=(\arabic*)]
					\item For each $l \in [n]$ and $k \in \{0, 1, 2\}$:
					Let $(x_l(0,k),x_l(1,k),y_l(0,k),y_l(1,k))$ be the one defined in Eq.~\eqref{eq:3-CSHS-inputs}.\\
					For $i, j \in \{0, 1\}:$
					If $(x_l, y_l) = (x_l(i,k), y_l(j,k)):$
					\begin{equation}
						n_{\mathrm{I}}(l,k,i,j) \leftarrow n_{\mathrm{I}}(l,k,i,j) + 1, \quad e_{\mathrm{I}}(l,k,i,j) \leftarrow e_{\mathrm{I}}(l,k,i,j) + (-1)^{a_l + b_l}.
					\end{equation}
					
					\item For each $l \in [n-1]:$
					If $x_l ={\Dr}$, $x_{l+1} ={\Dl}$ and $y_l = y_{l+1} \in \{0, 1, 2\}:$
					Let $P \in \{X,Y,Z\}$ be the Pauli corresponding to $y_l$.
					With probability
					\begin{equation}\label{eq:rejection-sample-prot1}
						\frac{\min_Q(p_Q^{(l)}p_Q^{(l+1)})}{p_P^{(l)}p_P^{(l+1)}}
					\end{equation}
					perform the following update:
					\begin{equation}
						n_{\mathrm{II}}(l,l+1) \leftarrow n_{\mathrm{II}}(l,l+1)+1, \quad
						e_{\mathrm{II}}(l,l+1) \leftarrow e_{\mathrm{II}}(l,l+1)-(-1)^{\mathbf{1}[P=Y]}(-1)^{f(b_l|P,\vec{a}_l)+f(b_{l+1}|P,\vec{a}_{l+1})}.
					\end{equation}
					\item If $x_l = \diamond$ for all $l \in [n]:$
					Let $\vec{P}$ be the vector of Paulis corresponding to $\vec{y} = (y_0, ..., y_{n-1})$.
					Let $\vec{a} = \{\vec{a}_l\}_{l=0}^{n-1}$ and $\vec{b} = \{b_l\}_{l=0}^{n-1}$.
					\begin{equation}
						n_{\mathrm{III}} \leftarrow n_{\mathrm{III}}+1, \quad e_{\mathrm{III}} \leftarrow e_{\mathrm{III}} + \omega(f(\vec{b}|\vec{P}, \vec{a})|\vec{P}).
					\end{equation}
				\end{enumerate}
			\end{enumerate}
			
			\textbf{Finalization and output}
			\begin{enumerate}
				\item Compute the final estimators: 
				\begin{equation}\label{eq:estimators-prot1}
					\begin{aligned}
						\hat{C}_{ij}(l,k) &= \frac{e_{\mathrm{I}}(l,k,i,j)}{n_{\mathrm{I}}(l,k,i,j)} \\
						\hat{S}(l,k) &= \sum_{i,j \in \{0,1\}} (-1)^{ij} \hat{C}_{ij}(l,k) \\
						\hat{v}_{\mathrm{I}} &= 2\sqrt{2} - \frac{1}{3n} \sum_{l=0}^{n-1} \sum_{k=0}^{2} \hat{S}(l,k) \\
						\hat{v}_{\mathrm{II}} &= 1+ \frac{1}{n-1}\sum_{l=0}^{n-2}\frac{e_{\mathrm{II}}(l,l+1)}{n_{\mathrm{II}}(l,l+1)} \\
						\hvIII &= \frac{e_{\mathrm{III}}}{n_{\mathrm{III}}}.
					\end{aligned}
				\end{equation}
				If the denominator of any ratio is zero, the corresponding ratio is set to zero.
				\item Certification check: If $\hat{v}_{\mathrm{I}} > \frac{\varepsilon^2}{n^2W^2}$ or $\hat{v}_{\mathrm{II}} > \frac{\varepsilon}{nW}$ return \texttt{FAILED}. Else, return \texttt{CERTIFIED} and $\hvIII$.
			\end{enumerate}
		\end{protocol}
		
		\begin{proof}[Proof of Theorem~\ref{thm:fin-sample-DI-estimation-L-local-randomness}]
			Define $\hvI(l,k) \coloneqq 2\sqrt{2} - \hat{S}(l,k)$, $\hvII(l,l+1) \coloneqq 1 + \frac{e_{\mathrm{II}}(l,l+1)}{n_{\mathrm{II}}(l,l+1)}$.
			First, we require the total sample complexity $N$ to be large enough such that, with a high probability of at least $1-\delta$, the following error bounds are simultaneously satisfied for all $l, k$:
			\begin{equation}\label{eq:requirement_error_bound}
				\begin{split}
					\abs{\hvI(l,k)  - \vI(l,k)} &\le \frac{\varepsilon^2}{n^2W^2}, \\
					\abs{\hvII(l,l+1) - \vII(l,l+1)} & \le \frac{\varepsilon}{nW}, \\
					\abs{\hvIII - \vIII} & \le \varepsilon.
				\end{split}
			\end{equation}
			We will analyze the required $N$ at the end of the proof and proceed assuming \eqref{eq:requirement_error_bound}.
			Given this, we analyze the consequences conditional on the protocol outputting \texttt{CERTIFIED}. 
			We can bound the sum of square roots of the true values $\vI$:
			\begin{equation}
				\begin{split}
					\sum_{l=0}^{n-1}\sum_{k=0}^2\sqrt{\vI(l,k)} &\le \sum_{l=0}^{n-1}\sum_{k=0}^2\sqrt{\hvI(l,k)+\frac{\varepsilon^2}{W^2n^2}}\\
					&= \sum_{l=0}^{n-1}\sum_{k=0}^2\frac{1}{3n}\sqrt{9n^2 \hvI(l,k)+\frac{9\varepsilon^2}{W^2}}\\
					&\le \sqrt{\frac{1}{3n}  \left[\sum_{l=0}^{n-1}\sum_{k=0}^2 \left(9n^2\hvI(l,k)  + \frac{9\varepsilon^2}{W^2}\right)\right]} \\
					&\le \sqrt{9n^2 \hat{v}_{\mathrm{I}} + \frac{9\varepsilon^2}{W^2}}\\
					&= \cO\left(\frac{\varepsilon}{W}\right).
				\end{split}
			\end{equation}
			Here, the third line uses Jensen's inequality, the fourth line uses $\hvI = \frac{1}{3n} \sum_l\sum_k \hvI(l,k) $, and the fifth line uses the condition for outputting \texttt{CERTIFIED}: $\hvI \le \frac{\varepsilon^2}{n^2W^2}$.  We can also bound the sum of the true values $\vII$:
			\begin{equation}
				\begin{split}
					\sum_{l=0}^{n-2} \vII(l,l+1) &\le \sum_{l=0}^{n-2} \hvII(l,l+1) + \frac{\varepsilon}{W} \\
					&\le \cO\left(\frac{\varepsilon}{W}\right),
				\end{split}
			\end{equation}
			where we use $\hvII = \frac{1}{n-1} \sum_{l=0}^{n-2} \hvII(l,l+1)$ and the condition for outputting \texttt{CERTIFIED}: $\hvII \le \frac{\varepsilon}{nW}$.

			Next, let $\sigma \coloneqq \Lambda(\rho)$, where $\Lambda$ is the product of local channels on the main systems $A$ as specified in Theorem~\ref{thm:DI-estimation-L}. 
			We decompose $\sigma$ with respect to the indicator systems. Define the unnormalized states $A_+ \coloneqq (\mathbb{I} \otimes \bra{0}^{\otimes n}) \sigma (\mathbb{I} \otimes \ket{0}^{\otimes n})$ and $A_- \coloneqq (\mathbb{I} \otimes \bra{1}^{\otimes n}) \sigma (\mathbb{I} \otimes \ket{1}^{\otimes n})$ with weights $p_\pm \coloneqq \tr(A_\pm)$.
			Let $\mathcal{D} = \bigotimes_{i=1}^n \mathcal{D}_i$ denote the product of single-qubit dephasing channels on the ancillary indicator systems. The action of the channel yields 
			\begin{equation}\label{eq:dephasing-decompose}
				\mathcal{D}(\sigma) = p_+ \tau_+ \otimes \ketbra{0}^{\otimes n} + p_- \tau_-^* \otimes \ketbra{1}^{\otimes n} + (1-p_+ - p_-)\tilde{\sigma}
			\end{equation}
			for some residual state $\tilde{\sigma}$.
			
			Defining the total map $\Gamma = \mathcal{D} \circ \Lambda$, we have:
			\begin{equation}
				\begin{split}
					p_+ + p_-  &= \tr[(\ketbra{0}^{\otimes n} +\ketbra{1}^{\otimes n})\Gamma(\rho)] \\
					&= \tr[(\ketbra{0}^{\otimes n} +\ketbra{1}^{\otimes n})\Lambda(\rho)] \\
					&= 1-\cO\left(\frac{\varepsilon}{W}\right),
				\end{split}
			\end{equation}
			where the second inequality comes from the property of dephasing, and the third inequality comes from Theorem~\ref{thm:DI-estimation-L}.
			Set $p = p_+ / (p_+ + p_-)$. By Eq.~\eqref{eq:dephasing-decompose}, the trace distance satisfies $D(p\tau_+ \otimes \ketbra{0} + (1-p)\tau_-^*\otimes \ketbra{1}^{\otimes n}, \Gamma(\rho)) = \cO(\tfrac{\varepsilon}{W})$, proving Eq.~\eqref{eq:state-closeness}. 
			
			Finally, the error in the estimator is bounded via Theorem~\ref{thm:DI-estimation-L} and the triangle inequality:
			\begin{equation}
				\begin{split}
					\abs{\vIII - \tr(L \tau)} &\le \abs{ \vIII - \tr\left[\left(L\otimes\ketbra{0}^{\otimes n}+L^*\otimes\ketbra{1}^{\otimes n} \right)\Gamma(\rho)\right] } \\ 
					&\quad + \abs{ \tr\left[\left(L\otimes\ketbra{0}^{\otimes n}+L^*\otimes\ketbra{1}^{\otimes n} \right)\Gamma(\rho)\right] - \tr(L\tau) } \\
					&\le \cO(\varepsilon) + W D(p\tau_+ \otimes \ketbra{0} + (1-p)\tau_-^*\otimes \ketbra{1}^{\otimes n}, \Gamma(\rho)) \\
					&= \cO(\varepsilon).
				\end{split}
			\end{equation}
			Combining this with the bound $\abs{\hvIII - \vIII} \le \varepsilon$ yields  Eq.~\eqref{eq:estimator_error_independent_measurement}.
			
			We now analyze the total sample complexity $N$ required to satisfy the error bounds in Eq.~\eqref{eq:requirement_error_bound}.
			By assumption, there exists a constant $p=\Omega(1)$ such that for all $l \in\{0, \dots, n-1\}$, the marginal Pauli measurement probabilities satisfy $p_X^{(l)}, p_Y^{(l)}, p_Z^{(l)} \ge p$. Given the independent sampling of measurement settings detailed in Protocol~\ref{prot:di-witness-local-randomness}, we have:
			\begin{equation}
				\begin{split}
					\min\limits_{y \in [3], x \in [6], l\in [n]}\left\{\Pmsm[x,y; A_l,B_l]\right\} &\geq \frac{p}{8n} = \Omega\left(\frac{1}{n}\right), \\
					\min\limits_{y\in [3], l\in [n-1]}\left\{\Pmsm[\Dr, \Dl, y,y; A_l,A_{l+1}, B_l,B_{l+1}]\right\} &\geq \frac{p^2}{64 n^2} = \Omega\left(\frac{1}{n^2}\right),\\
					\Pmsm[\diamond^n; A\land \vec{y}\leftarrow \mathcal{D}] &= \left(1-\frac{1}{n}\right)^n = \Omega(1).
				\end{split}
			\end{equation}
			
			By substituting these minimum probabilities and the desired error bounds from Eq.~\eqref{eq:requirement_error_bound} into the general sample complexity results, specifically Lemmas~\ref{lemma:sample-complexity-N_I-and-N_II} and \ref{lem:sample-complexity-N_III}, the overall sample complexity $N$ is determined by the maximum of the three required sample complexities:
			\begin{equation}
				\begin{split}
					N &= \cO\left(\max\left\{W^4n^5\frac{\log(\frac{n}{\delta})}{\varepsilon^4},W^2n^4\frac{\log(\frac{n}{\delta})}{\varepsilon^2},W^2\frac{\log(\frac{1}{\delta})}{\varepsilon^2}\right\}\right) \\
					&= \cO\left(W^4n^5\frac{\log(\frac{n}{\delta})}{\varepsilon^4}\right).
				\end{split}
			\end{equation} 
		\end{proof}

		As a side remark, if the target observable $L$ is invariant under complex conjugation (i.e., $L = L^*$), its expectation value is unaffected by the global complex conjugation of the state. This property simplifies our result and allows us to directly estimate the expectation value of $L$ on the extracted state $\Gamma(\rho)$ itself. We formalize this observation in the following corollary.
		
		\begin{corollary}[Theorem~\ref{thm:fin-sample-DI-estimation-L-local-randomness} for $L = L^*$]\label{cor:fin-sample-DI-estimation-L-local-randomness-conjugation-invariant}
			If $L = L^*$, then Equations~\eqref{eq:state-closeness} and~\eqref{eq:estimator_error_independent_measurement} in Theorem~\ref{thm:fin-sample-DI-estimation-L-local-randomness} can be replaced by $|\hvIII - \tr(L\Gamma(\rho))| = \cO(\varepsilon)$, where we extend $L$ to act as the identity on the indicator systems.
		\end{corollary}
		\begin{proof}
			If $L = L^*$ then $\tr(L\tau_-) = \tr(L\tau_-^*)$ and hence
			\begin{equation}
				\begin{split}
					\tr(L\tau) &= \tr(L(p\tau_+ + (1-p)\tau_-^*)) \\
					&= \tr(L\otimes \mathbb{I} \left(p\tau_+\otimes\ketbra{0}^{\otimes n}+ (1-p)\tau_-^*\otimes\ketbra{1}^{\otimes n}\right)).
				\end{split}
			\end{equation}
			The claim then follows from \eqref{eq:state-closeness}, the triangle inequality, and 
			\begin{equation}
				\begin{split}
					|\tr(L\left(p\tau_+\otimes\ketbra{0}^{\otimes n}+ (1-p)\tau_-^*\otimes\ketbra{1}^{\otimes n} - \Gamma(\rho)\right))|&\leq 2\|L\|_{\infty}D\left(p \tau_+ \otimes \ketbra{0}^{\otimes n} + (1-p)\tau_-^* \otimes \ketbra{1}^{\otimes n}, \Gamma(\rho)\right) \\
					&\leq W\cO\left(\frac{\varepsilon}{W}\right).
				\end{split}
			\end{equation}
		\end{proof}
		A sufficient condition for $L = L^*$ is that the distribution $\mathcal{D}$ only samples Pauli basis choices $\vec{P}$ with an even number of $Y$ measurements.

		\subsection{Using shared randomness}\label{subsec:shared-randomness}
		We now proceed to analyze the scenario where all parties share classical randomness. Shared randomness can be established either by generating it collaboratively in a single location beforehand or by employing DI methods such as DI quantum key distribution or conference key agreement.
		
		Conceptually, shared randomness allows the spatially separated parties to jointly roll a dice in each round to decide which set of estimators — $\vI$, $\vII$, or $\vIII$ — to test.
		More importantly, in the case of estimating $\vIII$, the parties can use the shared randomness to jointly generate a random sample $\vec{P}$ from a general distribution $\mathcal{D}$. This eliminates the requirement that $\mathcal{D}$ be in product form, as we used in the local randomness scenario.
		
		The protocol utilizing shared randomness is summarized in Protocol~\ref{prot:di-witness-shared-randomness}. The following theorem guarantees its performance:

		\begin{theorem}[Finite sample version of Theorem~\ref{thm:DI-estimation-L} using shared randomness]\label{thm:fin-sample-DI-estimation-L-shared-randomness}
			Let $L$ be the $n$-partite observable defined in Eq.~\eqref{eq:L-def}, with the weighting coefficients bounded by a constant $W>1$, i.e., $|\omega(\vec{b}|\vec{P})| \leq W$.
			
			For a desired precision $\varepsilon > 0$ and failure probability $0 < \delta < 1$, if the total number of samples $N$ is set to:
			\begin{equation}\label{eq:claim-sample-complexity-shared-randomness}
				N = \cO\left(W^4n^4\frac{\log(\frac{n}{\delta})}{\varepsilon^4}\right).
			\end{equation}
			Then, with probability at least $1-\delta$, if Protocol~\ref{prot:di-witness-shared-randomness} outputs \texttt{CERTIFIED}, then there exist $p \in [0,1]$, a state $\tau = p \tau_+ + (1-p)\tau_-$, and a product of local channels $\Gamma$ such that:
			\begin{equation}\label{eq:state-closeness-shared-randomness}
				D\left(p \tau_+ \otimes \ketbra{0}^{\otimes n} + (1-p)\tau_-^* \otimes \ketbra{1}^{\otimes n}, \Gamma(\rho)\right) \le \cO\left(\frac{\varepsilon}{W}\right),
			\end{equation}
			and the returned estimator $\hvIII$ satisfies:
			\begin{equation}\label{eq:estimator_error_shared_randomness}
				| \hvIII - \tr(L\tau) | = \cO(\varepsilon).
			\end{equation}
		\end{theorem}
		Compared to Theorem~\ref{thm:fin-sample-DI-estimation-L-local-randomness}, Theorem~\ref{thm:fin-sample-DI-estimation-L-shared-randomness} makes no assumption about the distribution $\mathcal{D}$ and its sample complexity is reduced by a factor of $n$.
		
		\begin{protocol}{DI multipartite Pauli measurements using shared randomness}{prot:di-witness-shared-randomness}
			\textbf{Input: } 
			\begin{enumerate}
				\item Total number of measurement rounds $N$.
				\item Target observable $L$ defined in Eq.~\eqref{eq:L-def}, with target measurement distribution $\mathcal{D}$ and  weighting coefficients $\omega(\vec{b}|\vec{P})$ satisfying $|\omega(\vec{b}|\vec{P})| \leq W$. 
			\end{enumerate}
			
			\textbf{Initialization}
			
			Set all counters and estimators to zero, same as in Protocol~\ref{prot:di-witness-local-randomness}.
			
			\textbf{Measurement rounds}
			
			Execute the following steps for $N$ total rounds:
			\begin{enumerate}
				\item Setting selection: The parties use their shared randomness to jointly select one of the three test types uniformly at random, with probability $\tfrac{1}{3}$ for each type.
				\begin{enumerate}[label=(\Roman*)]
					\item 
					For each $l\in [n]$: Choose $k\in\{0,1,2\}, i,j\in\{0,1\}$ uniformly at random. Let $(x_l(0,k),x_l(1,k),y_l(0,k),y_l(1,k))$ be the one defined in Eq.~\eqref{eq:3-CSHS-inputs}
					and set $x_l = x_l(i,k)$ and $y_l = y_l(j,k)$.
					
					\item 
					Jointly choose $s \in \{0, 1\}$ and $y \in \{0, 1, 2\}$ uniformly at random.
					For all $l \in \{0, \dots, n-2\}$ such that $l \equiv s \pmod{2}$: Set $x_l = \,\Dr$, $x_{l+1} = \, \Dl$, and $y_l = y_{l+1} = y$. 
					
					\item 
					Jointly choose the vector of Paulis $\vec{P}=(P_0, \dots, P_{n-1}) \in \{X, Y, Z\}^n$ according to the distribution $\mathcal{D}$.
					For all $l \in \{0, \dots, n-1\}$: Set $x_l = \diamond$, and set $y_l \in \{0, 1, 2\}$ to be the index corresponding to $P_l$.
				\end{enumerate}
				\item Measurement and counter updates: Parties perform measurements according to the selected settings and update the relevant counters and estimators exactly as detailed in Protocol~\ref{prot:di-witness-local-randomness} (skipping the rejection sampling in Eq.~\ref{eq:rejection-sample-prot1}).
			\end{enumerate}
			
			\textbf{Finalization and output}
			
			Compute the final estimators and perform the certification check, same as in Protocol~\ref{prot:di-witness-local-randomness}.
		\end{protocol}
		
		\begin{proof}[Proof of Theorem~\ref{thm:fin-sample-DI-estimation-L-shared-randomness}]
			The proof is analogous to the proof of Theorem~\ref{thm:fin-sample-DI-estimation-L-local-randomness}, with the difference that 
			\begin{equation}
				\begin{split}
					\min\limits_{l\in [n], k \in [3], i,j \in \{0,1\}}\left\{\Pmsm[x_l(i,k),y_l(j,k); A_l,B_l]\right\} & = \Omega(1), \\
					\min\limits_{y\in [3], l\in [n-1]}\left\{\Pmsm[\Dr, \Dl, y,y; A_l,A_{l+1}, B_l,B_{l+1}]\right\} & = \Omega(1),\\
					\Pmsm[\diamond^n; A\land \vec{y}\leftarrow \mathcal{D}] &= \Omega(1).
				\end{split}
			\end{equation}
			Substituting these probabilities and the desired error bounds from Eq.~\eqref{eq:requirement_error_bound} into Lemmas~\ref{lemma:sample-complexity-N_I-and-N_II} and \ref{lem:sample-complexity-N_III}, the overall sample complexity is given by 
			\begin{equation}
				\begin{split}
					N &= \cO\left(\max\left\{W^4n^4\frac{\log(\frac{n}{\delta})}{\varepsilon^4},W^2n^2\frac{\log(\frac{n}{\delta})}{\varepsilon^2},W^2\frac{\log(\frac{1}{\delta})}{\varepsilon^2}\right\}\right) \\
					&= \cO\left(W^4n^4\frac{\log(\frac{n}{\delta})}{\varepsilon^4}\right).
				\end{split}
			\end{equation} 
		\end{proof}
		
		Analogous to Corollary~\ref{cor:fin-sample-DI-estimation-L-local-randomness-conjugation-invariant}, the assumption that $L = L^*$ yields an identical simplification for the shared randomness scenario. The proof follows exactly as before.
		\begin{corollary}[Theorem~\ref{thm:fin-sample-DI-estimation-L-shared-randomness} for $L = L^*$]\label{cor:fin-sample-DI-estimation-L-shared-randomness-conjugation-invariant}
			If $L = L^*$, then Equations~\eqref{eq:state-closeness-shared-randomness} and~\eqref{eq:estimator_error_shared_randomness} in Theorem~\ref{thm:fin-sample-DI-estimation-L-shared-randomness}  can be replaced by $| \hvIII - \tr(L\Gamma(\rho)) | = \cO(\varepsilon)$,  where we extend $L$ to act as identity on the indicator systems.
		\end{corollary}
		
		\section{Completeness analysis}\label{sec:completeness}
		Our previous analysis has established the soundness of our self-testing protocols, guaranteeing a small estimation error when the protocols output \texttt{CERTIFIED}.
		However, a trivial protocol that always outputs \texttt{FAILED} would also satisfy this soundness requirement, yet it is useless for self-testing any quantum device. 
		Therefore, we now proceed to analyze the \emph{completeness} of our protocols. We show that even if the physical realization implements the reference experiment with noticeable but small imperfections, our self-testing protocols will still output \texttt{CERTIFIED} with high probability after $\poly(n,\varepsilon^{-1})$ rounds.
		This establishes the utility and practical applicability of our self-testing protocols.
		
		Here, we consider imperfections arising from both the underlying quantum state and the implemented measurements. We remark that the completeness analysis requires the existence of an embedding of the reference experiment into the physical experiment. That is, the support of the reference state, $\tau_T \otimes \Phi^+$, is a subspace of the support of the actual input state $\rho = \psi$ of the physical experiment, and the reference measurement operators can be extended trivially onto the domain of the measurement operators in the physical experiment.
		For simplicity in the subsequent derivations, we assume throughout this section that this embedding has already been performed and that the reference state and the actual input state reside in the same Hilbert space.
		
		Regarding imperfections in the state, we quantify the imperfection by the trace distance between the actual input state $\rho$ and the ideal reference state, which is a product of Bell states $\Phi^+$ tensored with a state $\tau_T$ on the system $T$. 
		Here, $\Phi^+$ is composed of Bell states shared between all required auxiliary register pairs (e.g., between $R^{\Dr}_l$ and $R^{\Dl}_{l+1}$, and between $S_l$ and $B_l$, for all $l$).
		The state error $\epsilon_{\mathrm{S}}$ is defined as the minimum trace distance over any states $\tau_T$ on the input system $\mathcal{H}_T$:
		\begin{equation}\label{eq:state_imperfection}
			\epsilon_{\mathrm{S}} \coloneqq \min_{\tau_T} D(\rho, \tau_T \otimes \Phi^+).
		\end{equation}

		For the measurement imperfections, we note that each measurement can be represented as a quantum channel. The measurement channel $\mathcal{M}_{x_l}^{(l)}$ maps the quantum state to a classical probability distribution over the outcomes $\vec{a}_l$:
		\begin{equation}
			\mathcal{M}_{x_l}^{(l)} (\rho) = \sum_{\vec{a}_l} \tr\left(\rho M_{\vec{a}_l | x_l}^{(l)}\right) \ketbra{\vec{a}_l},
		\end{equation}
		where $M_{\vec{a}_l | x_l}^{(l)}$ is the measurement operator and $\ketbra{\vec{a}_l}$ is the classical output state encoding the outcome. The channel $\mathcal{N}_{y_l}^{(l)}$ for party $B_l$ is defined similarly.
		Denote the ideal measurement channels in the reference experiment as $\widetilde{\mathcal{M}}_{x_l}^{(l)}$ and $\widetilde{\mathcal{N}}_{y_l}^{(l)}$. The measurement imperfection $\epsilon_{\mathrm{M}}$ is then quantified by the maximum diamond distance between the implemented channel and its ideal counterpart, maximized over all parties $l$ and all relevant measurement settings $x_l$ and $y_l$:
		\begin{equation}\label{eq:measurement_imperfection}
			\epsilon_{\mathrm{M}} = \max \left\{\max_{l, x_l} \left\{\norm{\mathcal{M}_{x_l}^{(l)} -\widetilde{\mathcal{M}}_{x_l}^{(l)}}_{\diamond} \right\}, \max_{l, y_l} \left\{\norm{\mathcal{N}_{y_l}^{(l)} -\widetilde{\mathcal{N}}_{y_l}^{(l)}}_{\diamond}\right\} \right\}.
		\end{equation}
		
		Then, we have the following completeness guarantee.
		
		\begin{theorem}[Completeness of DI multipartite Pauli measurement protocols]\label{thm:completeness}
			Let $L$ be the $n$-partite observable defined in Eq.~\eqref{eq:L-def}, with the weighting coefficients bounded by a constant $W>1$, i.e., $|\omega(\vec{b}|\vec{P})| \leq W$.
			Let $\epsilon_{\mathrm{S}}$ and $\epsilon_{\mathrm{M}}$ be the state and measurement device errors defined in Eqs.~\eqref{eq:state_imperfection} and \eqref{eq:measurement_imperfection}.
			For a desired precision $\varepsilon > 0$, failure probability $0 < \delta < 1$, suppose the device imperfections satisfy the bound:
			\begin{equation}
				\epsilon_{\mathrm{S}}, \epsilon_{\mathrm{M}} = \cO\left( \frac{\varepsilon^2}{n^2W^2}\right).
			\end{equation}
			Then, for a total number of samples $N$ chosen as specified below, with probability at least $1-\delta$:
			\begin{enumerate}
				\item\label{bullet:CompletenessProtocol1} Protocol~\ref{prot:di-witness-local-randomness} outputs \texttt{CERTIFIED} and $\hvIII$. This protocol requires the target measurement distribution $\mathcal{D}$ for $L$ to be a product distribution $\mathcal{D} = \bigotimes_{l=0}^{n-1}\mathcal{D}^{(l)}$, with marginal probabilities satisfying $p_X^{(l)}, p_Y^{(l)}, p_Z^{(l)} = \Omega(1)$ for all $l\in[n]$. The required sample complexity is set to $N = \cO\left(W^4n^5\frac{\log(\frac{n}{\delta})}{\varepsilon^4}\right)$.
				\item\label{bullet:CompletenessProtocol2} Protocol~\ref{prot:di-witness-shared-randomness} outputs \texttt{CERTIFIED} and $\hvIII$. The required sample complexity is set to $N = \cO\left(W^4n^4\frac{\log(\frac{n}{\delta})}{\varepsilon^4}\right)$.
			\end{enumerate}
			Let $\tilde{\tau}_T$ be any reference state in $T$. In both cases, the returned $\hvIII$ satisfies
			\begin{equation}\label{eq:claim-vIII-completeness}
				|\hvIII - \tr(L\tilde{\tau}_T)| = \cO(\varepsilon + W D(\rho, \tilde{\tau}_T \otimes \Phi^+)).
			\end{equation}
			
		\end{theorem}
		
		Equation~\eqref{eq:claim-vIII-completeness} tells us that if the state $\rho$ in the physical experiment approximately equals a product state $\tilde{\tau}_T$ on system $T$ and the reference bell states $\Phi^+$ on the other systems, then $\hvIII$ approximates $\tr(L\tilde{\tau}_T)$.
		
		\begin{proof}
			In the reference experiment, the deviation values $\vI(l,k)$ and $\vII(l,l+1)$ are exactly zero for all $l$, $k$ and test states $\tau_T$, as the ideal setup perfectly achieves the bounds $S(l,k)=2\sqrt{2}$ and the partial transpose identity (Eq.~\eqref{eq:v_II-protocol-specification}).
			
			We first analyze the values $\vI(l,k)$ and $\vII(l,l+1)$ in the presence of physical device imperfections $\epsilon_{\mathrm{S}}$ and $\epsilon_{\mathrm{M}}$.
			We consider $\vII(l,l+1)$ as an example. Its expected value is a summation (weighted by coefficients that are all $\cO(1)$) over measurement outcomes of the composed measurement channel $\mathcal{C}_{y_l} \coloneqq \mathcal{M}^{(l)}_{\Dr} \otimes \mathcal{M}^{(l+1)}_{\Dl} \otimes \mathcal{N}_{y_l}^{(l)} \otimes \mathcal{N}_{y_{l+1}}^{(l+1)}$ for $y_l=y_{l+1} \in \{0,1,2\}$ (see Eq.~\eqref{eq:v_II-protocol-specification}). We denote the corresponding ideal channel as $\widetilde{\mathcal{C}}_{y_l}$. By the triangle inequalities, we have:
			\begin{equation}\label{eq:proof-channel-diamond-norm-triangular}
				\norm{\mathcal{C}_{y_l} - \widetilde{\mathcal{C}}_{y_l}}_{\diamond} \le \norm{\mathcal{M}^{(l)}_{\Dr} - \widetilde{\mathcal{M}}^{(l)}_{\Dr}}_{\diamond} + \norm{\mathcal{M}^{(l+1)}_{\Dl} - \widetilde{\mathcal{M}}^{(l+1)}_{\Dl}}_{\diamond} + \norm{\mathcal{N}_{y_l}^{(l)} - \widetilde{\mathcal{N}}_{y_l}^{(l)}}_{\diamond} + \norm{\mathcal{N}_{y_{l+1}}^{(l+1)} - \widetilde{\mathcal{N}}_{y_{l+1}}^{(l+1)}}_{\diamond} \le 4\epsilon_{\mathrm{M}}.
			\end{equation}
			The value of $\vII$ can be written as $\sum_{y_l}\frac{1}{3}\tr[O_{\mathrm{II}, y_l} \mathcal{C}_{y_l}(\rho)]$ with $O_{\mathrm{II},y_l}$ a weighting of the different measurement outcomes (as specified in Box~\ref{box:ProtocolEstimators}).
			Using the general inequality $\tr[O \tau_1] - \tr[O \tau_2] \le \norm{O}_{\infty} D(\tau_1, \tau_2)$ we can hence bound, 
			\begin{equation}
				\begin{split}\label{eq:defivation-vII-ideal}
					\vII(l,l+1) - 0 &= \tr\left[\sum_{y_l}\frac{1}{3}O_{\mathrm{II}, y_l} (\mathcal{C}_{y_l}(\rho)-\widetilde{\mathcal{C}}_{y_l}(\tau \otimes \Phi^+))\right] \\
					&\le \cO\left(\sum_{y_l} D\left(\mathcal{C}_{y_l}(\rho), \widetilde{\mathcal{C}}_{y_l}(\tau \otimes \Phi^+)\right)\right) \\
					&\le \cO\left(\sum_{y_l} D\left(\mathcal{C}_{y_l}(\rho), \widetilde{\mathcal{C}}_{y_l}(\rho)\right)\right) + \cO\left(\sum_{y_l} D\left(\widetilde{\mathcal{C}}_{y_l}(\rho), \widetilde{\mathcal{C}}_{y_l}(\tau \otimes \Phi^+)\right)\right) \\
					&\le \cO\left( \sum_{y_l} \norm{\mathcal{C}_{y_l} - \widetilde{\mathcal{C}}_{y_l}}_{\diamond} \right) + \cO\left( D\left(\rho, \tau \otimes \Phi^+ \right)\right) \\
					&\le \cO(\epsilon_{\mathrm{M}} + \epsilon_{\mathrm{S}}).
				\end{split}
			\end{equation}
			Here, $\tau$ is the state that obtains the minimum in \eqref{eq:state_imperfection}. The third line uses the definition of diamond norm and the data processing inequality.
			Since we assume $\epsilon_{\mathrm{S}}, \epsilon_{\mathrm{M}} \le \cO\left( \frac{\varepsilon^2}{n^2W^2}\right)$, the expectation value is bounded by $\vII(l,l+1) = \cO\left( \frac{\varepsilon^2}{n^2W^2}\right) \leq\cO\left( \frac{\varepsilon}{nW}\right)$.
			A similar argument applies to $\vI(l,k)$, yielding $\vI(l,k) = \cO\left( \frac{\varepsilon^2}{n^2W^2}\right)$.
			
			By setting the sample complexity $N$ as specified in the theorem, the estimation errors for the finite sample estimates $\hat{v}$ are controlled by Eq.~\eqref{eq:requirement_error_bound} with probability at least $1-\delta$ (follow the proofs of Theorems~\ref{thm:fin-sample-DI-estimation-L-local-randomness} and \ref{thm:fin-sample-DI-estimation-L-shared-randomness}), which gives $\hvI(l,k) = \cO\left(\frac{\varepsilon^2}{n^2W^2}\right)$ and $ \hvII(l,l+1) = \cO\left(\frac{\varepsilon}{nW}\right)$ for any $l,k$. Therefore, 
			\begin{equation}
				\begin{split}
					\hvI &= \frac{1}{3n} \sum_l\sum_k \hvI(l,k) = \cO\left(\frac{\varepsilon^2}{n^2W^2}\right), \\
					\hvII &= \frac{1}{n-1} \sum_{l=0}^{n-2} \hvII(l,l+1) = \cO\left(\frac{\varepsilon}{nW}\right).
				\end{split}
			\end{equation}
			By increasing $N$ and tightening the errors $\epsilon_{\mathrm{S}}, \epsilon_{\mathrm{M}}$ by a constant scale (which does not change the order), we have $\hvI \le \frac{\varepsilon^2}{n^2W^2}, \hvII \le \frac{\varepsilon}{nW}$, which satisfies the criteria in the certification check steps (see Protocol~\ref{prot:di-witness-local-randomness}), therefore the protocols will output \texttt{CERTIFIED}.
			
			To conclude the proof, we analyze the value of $\hvIII$ produced by Protocol~\ref{prot:di-witness-local-randomness} and Protocol~\ref{prot:di-witness-shared-randomness}. In the ideal reference experiment performed on the state $\tilde{\tau}_T \otimes \Phi^+$, the value $\vIII$ equals $\tr(L\tilde{\tau}_T)$. We show that $\hvIII$ remains close to this ideal value, with a deviation bounded by the measurement and state imperfections.
			Concretely, we define the physical measurement channel for a given Pauli string $\vec{P}$ as $\mathcal{C_\vec{P}} = \bigotimes_{l=0}^{n-1} \mathcal{M}^{(l)}_{\diamond} \otimes \mathcal{N}_{(\vec{P})_l}^{(l)}$, and denote the corresponding ideal channel in the reference experiment as $\widetilde{\mathcal{C}}_\vec{P}$. The difference between the actual expectation value $\vIII$ and the ideal value $\tr(L\tilde{\tau}_T)$ can be bounded as follows:
			\begin{equation}
				\begin{split}
					|\vIII - \tr(L\widetilde{\tau}_T)| 
					&\leq \cO(W \mathbb{E}_{\vec{P}\leftarrow \mathcal{D}}[D\left(\mathcal{C}_{\vec{P}}(\rho), \widetilde{\mathcal{C}}_{\vec{P}}(\widetilde{\tau}_T \otimes \Phi^+)\right)]\\
					&\leq \cO(W \mathbb{E}_{\vec{P}\leftarrow \mathcal{D}}[D\left(\mathcal{C}_{\vec{P}}(\rho),\widetilde{\mathcal{C}}_{\vec{P}}(\rho)\right) + D\left(\widetilde{\mathcal{C}}_{\vec{P}}(\rho),\widetilde{\mathcal{C}}_{\vec{P}}(\widetilde{\tau}_T \otimes \Phi^+)\right)]\\
					& \leq \cO(W(n\epsilon_{\mathrm{M}} + D(\rho, \tilde{\tau}_T \otimes \Phi^+))) \\
					&\leq \cO(\varepsilon+W D(\rho, \tilde{\tau}_T \otimes \Phi^+)).
				\end{split}
			\end{equation}
			Here, we use similar arguments as in \eqref{eq:proof-channel-diamond-norm-triangular} and \eqref{eq:defivation-vII-ideal}, with the difference that here $\norm{O_{\mathrm{III}}}_{\infty} = \cO(W)$ (instead of $\cO(1)$) and $ \norm{\mathcal{C}_{\vec{P}} - \widetilde{\mathcal{C}}_{\vec{P}}}_{\diamond} = \cO(n\epsilon_{\mathrm{M}})$ (instead of $\cO(\epsilon_{\mathrm{M}})$). 
			
			Following the sample complexity analysis in the proofs of Theorems~\ref{thm:fin-sample-DI-estimation-L-local-randomness} and \ref{thm:fin-sample-DI-estimation-L-shared-randomness}, we ensure that for the specified $N$, the statistical error satisfies $\hvIII \approx_{\varepsilon} \vIII$ with high probability. Consequently,
			\begin{equation}
				\hvIII \approx_{\varepsilon + W D(\rho, \tilde{\tau}_T \otimes \Phi^+)} \tr(L\tilde{\tau}_T).
			\end{equation}  
		\end{proof}
		
		\section{Sample complexity for self-testing generic multipartite states}\label{app:self-testing-fin-samples}
		
		After establishing that our protocol robustly self-tests multipartite Pauli measurements, we now demonstrate that this framework leads to a sample-efficient self-testing protocol for a generic multipartite target state $\Psi$.
		
		\subsection{Definition for self-testing states}
		The goal of self-testing a state is to use statistics from local measurements on spatially separated systems to certify that the physical state is equivalent to a reference state $\Psi$, up to local quantum channels. A protocol that self-tests a state $\Psi$ is sound if it can only succeed for states that are close to $\Psi$. It is complete if it succeeds with high probability for a physical experiment that is close to the reference experiment.
		
		\begin{definition}[Robust self-testing of states]\label{def:approx-selftest-states}
			Consider a protocol which certifies whether a physical state $\rho$ is $\varepsilon$-close to an $n$-partite pure state $\Psi$ by outputting either \texttt{CERTIFIED} or \texttt{FAILED}.
			We characterize the protocol by the following two properties:
			\begin{enumerate}
				\item $(\varepsilon, \delta)-$\textbf{soundness}: Suppose that for any $\lambda \in [0,1]$ and any quantum channel $\Lambda = \bigotimes_{l=0}^{n-1} \Lambda_l$ consisting of a tensor product of local channels, the physical state $\rho$ satisfies:
				\begin{equation}
					D(\lambda \Psi \otimes \ketbra{0}^{\otimes n} + (1-\lambda) \Psi^* \otimes \ketbra{1}^{\otimes n}, \Lambda(\rho)) \ge \varepsilon
				\end{equation}
				where $\ketbra{0}^{\otimes n}$ and $\ketbra{1}^{\otimes n}$ are auxiliary indicator states. Then, the protocol outputs \texttt{FAILED} with probability at least $1-\delta$.
				
				\item $(\varepsilon_T, \varepsilon_{\mathrm{Net}}, \delta)-$\textbf{completeness}: If the input state $\rho$ satisfies $D(\rho , \Psi \otimes \Phi^+) < \varepsilon_T$, and the state and measurement imperfections satisfy $\epsilon_S,\epsilon_M < \varepsilon_{\mathrm{Net}}$ (see \eqref{eq:state_imperfection}, \eqref{eq:measurement_imperfection} for definitions of $\epsilon_S,\epsilon_M$), then the protocol outputs \texttt{CERTIFIED} with probability at least $1-\delta$.
			\end{enumerate}
		\end{definition}
		
		In the completeness definition, $\varepsilon_T$ limits the deviation of the physical state from the ideal reference state, while $\varepsilon_{\mathrm{Net}}$ limits the deviation of the auxiliary quantum network states as well as the measurements. 
		For example, if $\rho = \rho_T \otimes \rho_{SRB}$, then the definition requires $D(\rho_T, \Psi) < \varepsilon_T$ and $D(\rho_{SRB}, \Phi^+) <  \min(\varepsilon_T, \varepsilon_{\mathrm{Net}})$.
		
		The freedom up to a global complex conjugation in this and other self-testing statements is intrinsic to relying solely on measurement statistics. Indeed, for any quantum measurement $\{M_{a|x}\}$ performed on a state $\rho$, the resulting statistics are identical to those obtained by performing the complex-conjugated measurements $\{M_{a|x}^*\}$ on the complex-conjugated state $\rho^*$. 
		
		\subsection{Self-testing states using shared randomness}
		Having established a general framework for lifting device-dependent Pauli measurement protocols to their DI versions, we can now integrate our method with state certification protocols that utilize Pauli measurements to obtain a robust self-testing protocol of states.
		
		We first consider the shadow overlap protocol introduced in Ref.~\cite{Huang2024Certifying}, which is designed to certify whether an experimental state $\rho$ is close to a pure target state $\Psi$. For a given target state $\Psi$, the protocol can be viewed as the estimation of the following observable:
		\begin{equation}\label{eq:L_psi-def} 
			L_\Psi = \frac{1}{n} \sum_{k=0}^{n-1} \sum_{\vec{b}^{(k)} \in \{0,1\}^{n-1}} \ketbra{\vec{b}^{(k)}}_{[n]\setminus\{k\}} \otimes \ketbra{\Psi_{\vec{b}}^{(k)}}. 
		\end{equation}
		Here, the unnormalized conditional post-measurement state on the $k$-th qubit, conditioned on the outcomes $\vec{b}^{(k)} \coloneqq (b_0,\cdots,b_{k-1},b_{k+1},\cdots,b_{n-1})$ on the other qubits, is denoted by $\ket{\tilde{\Psi}_{\vec{b}}^{(k)}} \coloneqq (\mathbb{I}_k \otimes \bra{\vec{b}^{(k)}}) \ket{\Psi}$. 
		The normalized state is given by $\ket{\Psi_{\vec{b}}^{(k)}} \coloneqq \ket{\tilde{\Psi}_{\vec{b}}^{(k)}} / \sqrt{\braket{\tilde{\Psi}_{\vec{b}}^{(k)}}{\tilde{\Psi}_{\vec{b}}^{(k)}}}$ if $|\braket{\tilde{\Psi}_{\vec{b}}^{(k)}}{\tilde{\Psi}_{\vec{b}}^{(k)}}|>0$ and $0$ otherwise.

		The observable $L_\Psi$ measures the average overlap between the conditional post-measurement states of the experimental state $\rho$ and those of the target state $\Psi$. It is straightforward to verify that the largest eigenvalue of $L_{\Psi}$ is $1$, which is achieved by the ideal target state: $\tr(L_{\Psi} \Psi) = 1$. 
		Consequently, the effectiveness of $L_{\Psi}$ for certification is determined by its spectral gap $\Delta(L_{\Psi})$, defined as the difference between its largest and second-largest eigenvalues. A larger gap implies that any state significantly different from $\Psi$ will yield a lower expectation value, making errors easier to detect. Crucially, this spectral gap has been proved to be noticeable with respect to the system size for generic Haar-random states:
		\begin{fact}[{\cite[Theorem 16]{Huang2024Certifying}}]\label{fact:shadow_overlap}
			Let $\Psi$ be an $n$-qubit Haar-random state. Then, with probability at least $1-\exp(-cn)$ over the choice of $\Psi$, $\Delta(L_{\Psi}) \ge c'n^{-2}$ for some positive constants $c, c'$.
		\end{fact}
		
		Furthermore, an estimator of the expectation value $\tr(L_{\Psi} \rho)$ can be obtained by performing local Pauli measurements on $\rho$. This is achieved by randomly selecting the measurement settings and post-processing the outcomes as follows:
		\begin{itemize}
			\item Measurement Setting: Uniformly select an index $k \in \{1, \dots, n\}$ and a Pauli operator $P \in \{X, Y, Z\}$. Set the $k$-th party's measurement to $P^{(k)} = P$, and set $P^{(l)} = Z$ for all other parties $l \neq k$.
			\item Weighting Coefficient: Define $\omega^{(k)}(\vec{b}|\vec{P})$ as the coefficient derived from local classical shadow tomography targeting the conditional state $\ket{\Psi_{\vec{b}}^{(k)}}$:
			\begin{equation}\label{eq:omega-def-huang-protocol}
				\omega^{(k)}(\vec{b}|\vec{P}) = \bra{\Psi_{\vec{b}}^{(k)}} \left( 3 \frac{\mathbb{I} + (-1)^{b_k} \sigma_P}{2} - \mathbb{I} \right) \ket{\Psi_{\vec{b}}^{(k)}}.
			\end{equation}
		\end{itemize}
		As shown in Ref.~\cite{Huang2024Certifying}, local shadow tomography ensures that this procedure yields an unbiased estimator of $\tr(L_{\Psi} \rho)$. The noticeable spectral gap $\Delta(L_{\Psi})$ guaranteed by Fact~\ref{fact:shadow_overlap} ensures that the protocol is robust against infidelity for generic Haar-random states $\Psi$.
		
		Therefore, by integrating the shadow overlap certification method into our robust self-testing framework, we obtain a practical protocol for certifying generic Haar-random states. 
		It is important to note that the measurement settings in this construction require coordination across different parties, and the weighting coefficients depend not only on the chosen Pauli basis $\vec{P}$ but also on the specific party index $k$. 
		This coordination is accommodated by our shared randomness protocol, Protocol~\ref{prot:di-witness-shared-randomness}, as the parties can collaboratively sample both $k$ and $\vec{P}$ at the start of each round. 
		Furthermore, the weighting coefficients are strictly bounded by a constant, $|\omega^{(k)}(\vec{b}|\vec{P})| \le W$, where $W = \cO(1)$. These elements together give the following performance guarantee:
		\begin{theorem}[Self-testing states using shared randomness]\label{thm:self-testing-state-shared-randomness}
			Let $\varepsilon' > 0$ be the desired precision. For an $n$-partite pure target state $\ket{\Psi}$ with a spectral gap $\Delta \coloneqq \Delta(L_{\Psi}) > 0$, Protocol~\ref{prot:self-testing-state-shared-randomness} is an $(\cO(\varepsilon'), \delta)$-sound and $(\cO(\Delta \varepsilon'^2), \cO(\frac{\Delta^2\varepsilon'^4}{n^2}), \delta)$-complete self-testing protocol for $\Psi$. The total sample complexity is given by:
			\begin{equation}
				N = \cO\left( \frac{n^4}{\Delta^{4} \varepsilon'^8} \log\left(\frac{n}{\delta}\right) \right).
			\end{equation}
		\end{theorem}
		
		\begin{protocol}{Self-testing states using shared randomness}{prot:self-testing-state-shared-randomness}
			\begin{enumerate}
				\item For a pure state $\Psi$ with $\Delta(L_{\Psi}) > 0$, execute Protocol~\ref{prot:di-witness-shared-randomness} for the observable $L_{\Psi}$ using $N$ measurement rounds.
				
				\item  If the output of Protocol~\ref{prot:di-witness-shared-randomness} is \texttt{FAILED} or the output is an estimate $\hat{\omega}$ with $\hat{\omega} < 1 - \Delta\varepsilon'^2$, output \texttt{FAILED}.
				\item Else, output \texttt{CERTIFIED}.
			\end{enumerate}
		\end{protocol}
		
		\begin{proof}
			The sample complexity is obtained by substituting the precision $\varepsilon = \Delta\varepsilon'^2$ and the constant bound $W = \cO(1)$ into Theorem~\ref{thm:fin-sample-DI-estimation-L-shared-randomness}. We now verify that Protocol~\ref{prot:self-testing-state-shared-randomness} satisfies the requirements for robust self-testing in Definition~\ref{def:approx-selftest-states}.
			
			\emph{Soundness}---Theorem~\ref{thm:fin-sample-DI-estimation-L-shared-randomness} guarantees that with probability at least $1-\delta$, if the protocol outputs \texttt{CERTIFIED}, there exists $p\in [0,1]$ and states $\tau_+,\tau_-$, such that the input state $\rho$ and the local channels $\Gamma$ satisfy:
			\begin{equation}\label{eq:distance-state-extraction}
				D(p\tau_+ \otimes \ketbra{0}^{\otimes n} + (1-p)\tau_-^* \otimes \ketbra{1}^{\otimes n}, \Gamma(\rho)) \le \cO(\varepsilon), 
			\end{equation}
			and
			\begin{equation}
				\tr[L_{\Psi} (p\tau_+ + (1-p)\tau_-)] \ge 1- \cO(\varepsilon).
			\end{equation}
			As the maximal eigenvalue of $L_{\Psi}$ is $1$, this weighted sum implies that both terms must be close to unity:
			\begin{equation}
				\tr(L_{\Psi}\tau_+) \geq 1 - \cO\left(\frac{\varepsilon}{p}\right), \quad \tr(L_{\Psi}\tau_-) \geq 1 - \cO\left(\frac{\varepsilon}{1-p}\right).
			\end{equation}
			Since $\ket{\Psi}$ is the eigenvector for the eigenvalue $1$, which has spectral gap $\Delta$, we get:
			\begin{equation}
				\bra{\Psi}\tau_+\ket{\Psi} \geq 1 - \cO\left(\frac{\varepsilon}{\Delta p}\right), \quad \bra{\Psi}\tau_-\ket{\Psi} \geq 1 - \cO\left(\frac{\varepsilon}{\Delta (1-p)}\right).
			\end{equation}
			By the Fuchs–van de Graaf inequality, $\|\rho-\tau\|_1 \leq 2\sqrt{1-F(\rho,\tau)}$, we have:
			\begin{equation}\label{eq:fidelity-to-trace-norm}
				\|\ketbra{\Psi} - \tau_+\|_1 \leq \cO\left(\sqrt{\frac{\varepsilon}{\Delta p}}\right), \quad \|\ketbra{\Psi} - \tau_-\|_1 \leq \cO\left(\sqrt{\frac{\varepsilon}{\Delta (1- p)}}\right).
			\end{equation}
			By the concavity of the square root function,
			\begin{equation}
				p \|\ketbra{\Psi} - \tau_+\|_1 + (1-p)\|\ketbra{\Psi} - \tau_-\|_1 \leq  \cO(\sqrt{\Delta^{-1}\varepsilon}) = \cO(\varepsilon').
			\end{equation}
			Hence, 
			\begin{equation}\label{eq:distance-state-ideal}
				D(p \ketbra{\Psi} \otimes \ketbra{0}^{\otimes n} + (1-p) \ketbra{\Psi^*} \otimes \ketbra{1}^{\otimes n}, p\tau_+ \otimes \ketbra{0}^{\otimes n} + (1-p)\tau_-^* \otimes \ketbra{1}^{\otimes n}) = \cO(\varepsilon').
			\end{equation}
			Combining Eqs.~\eqref{eq:distance-state-extraction} and ~\eqref{eq:distance-state-ideal} via triangle inequality, we obtain that 
			\begin{equation}
				D(p \ketbra{\Psi} \otimes \ketbra{0}^{\otimes n} + (1-p) \ketbra{\Psi^*} \otimes \ketbra{1}^{\otimes n}, \Gamma(\rho)) = \cO(\varepsilon'), 
			\end{equation}
			which establishes $(\cO(\varepsilon'),\delta)$-soundness.
			
			\emph{Completeness}---Assume the physical implementation is sufficiently close to the reference experiment, such that $D(\rho, \Psi \otimes \Phi^+) < \cO(\varepsilon)$ and $\epsilon_M, \epsilon_S < \cO(\varepsilon^2/n^2)$. According to Theorem~\ref{thm:completeness}, with probability at least $1-\delta$, Protocol~\ref{prot:self-testing-state-shared-randomness} will output an estimator $\hat{\omega}$ such that:
			\begin{equation}
				|\hat{\omega} - \tr(L_{\Psi} \ketbra{\Psi})| \le \varepsilon = \Delta\varepsilon'^2.
			\end{equation}
			Since $\tr(L_{\Psi} \ketbra{\Psi}) = 1$, we have $\hat{\omega} \ge 1 - \Delta\varepsilon'^2$. Consequently, the protocol outputs \texttt{CERTIFIED}. This confirms the $(\cO(\Delta \varepsilon'^2), \cO(\frac{\Delta^2\varepsilon'^4}{n^2}),\delta)$-completeness of the protocol.
		\end{proof}
		Combining Theorem~\ref{thm:self-testing-state-shared-randomness} with Fact~\ref{fact:shadow_overlap} gives a self-testing protocol for generic Haar-random states. This protocol exhibits polynomial sample complexity and maintains noticeable robustness against both state and measurement imperfections.
		
		\begin{corollary}[Sample-efficient self-testing of generic states using shared randomness]
			Let $\varepsilon' > 0$. For an $n$-qubit Haar-random state $\ket{\Psi}$, with probability at least $1 - \exp(-\Omega(n))$ over the choice of $\Psi$, Protocol~\ref{prot:self-testing-state-shared-randomness} is a $(\cO(\varepsilon'), \delta)$-sound and $(\cO(n^{-2} \varepsilon'^2), \cO(n^{-6}\varepsilon'^4), \delta)$-complete self-testing protocol. 
			The required sample complexity is $N = \cO\left(n^{12}\varepsilon'^{-8}\log(\frac{n}{\delta})\right)$. 
		\end{corollary}
		
		While we have established a robust self-testing protocol for generic states, the requirement for shared randomness is demanding. Achieving this requirement may involve integrating our framework with DI quantum random number generators or distribution protocols, such as DI quantum key distribution or conference key agreement, which lie beyond the scope of the current work.
		Alternatively, it is possible to sample measurement settings i.i.d.\ and use rejection sampling to achieve a result similar to Theorem~\ref{thm:self-testing-state-shared-randomness}. However, this would come at the cost of additional polynomial overhead in the sample complexity.
		Therefore, to achieve robust self-testing without any additional costs, we now develop a modified method that uses solely local randomness.

		\subsection{Self-testing states using local randomness}
		The key strategy for removing shared randomness is to let each party randomly sample a local basis, and to show that the resulting observable still preserves a noticeable gap. 
		
		We introduce a rotated-basis version of the observable $L_{\Psi}$ for a fixed basis $\vec{P} = (P_0, \dots, P_{n-1}) \in \{X,Y,Z\}^n$:
		\begin{equation}
			L_{\Psi}^{\vec{P}} = \frac{1}{n}\sum_{k=0}^{n-1} \sum_{\vec{b}^{(k)} \in \{0,1\}^{n-1}}
			\ketbra{\vec{b}^{(k)}, \vec{P}^{(k)}}_{[n]\setminus\{k\}} \otimes \ketbra{\Psi_{\vec{b}, \vec{P}}^{(k)}},
		\end{equation}
		where $\vec{b}^{(k)} \coloneqq (b_0,\cdots,b_{k-1},b_{k+1},\cdots,b_{n-1}), \vec{P}^{(k)} = (P_0,\cdots,P_{k-1}, P_{k+1},\cdots,P_{n-1})$, $\ket{\vec{b}^{(k)}, \vec{P}^{(k)}}$ denotes the basis state corresponding to measurement outcomes $\vec{b}^{(k)}$ in basis $\vec{P}^{(k)}$, and $\ket{\Psi_{\vec{b}, \vec{P}}^{(k)}}$ denotes the normalized state on system $k$ after projecting $\ket{\Psi}$ onto $\ket{\vec{b}^{(k)}, \vec{P}^{(k)}}_{[n]\setminus\{k\}}$.
		Since the Haar measure is invariant under arbitrary local unitary rotations, the spectral properties of the shadow overlap observable are preserved across different bases. This allows us to extend Fact~\ref{fact:shadow_overlap} to any fixed basis:
		\begin{corollary}\label{col:spectral_gap_any_basis}
			For any basis $\vec{P} \in \{X,Y,Z\}^n$, let $\Psi$ be an $n$-qubit Haar-random state. Then, with probability at least $1-\exp(-cn)$ over the choice of $\Psi$, $\Delta(L_{\Psi}^{\vec{P}}) \ge c'n^{-2}$ for some positive constants $c, c'$ independent of the basis $\vec{P}$.
		\end{corollary}
		
		We next consider the random-basis-enhanced \cite{du2025certifyinglocalizablequantumproperties} observable $M_{\Psi}$, defined by averaging the rotated shadow overlap observables $L_{\Psi}^{\vec{P}}$ over all possible combinations of local Pauli bases:
		\begin{equation}\label{eq:modified-Huang-protocol-def}
			M_{\Psi} \coloneqq \frac{1}{3^n} \sum_{\vec{P} \in \{X,Y,Z\}^n} L_{\Psi}^{\vec{P}}.
		\end{equation}
		We show that this observable satisfies the following guarantee:
		\begin{lemma}\label{lem:random-basis-enh}
			Let $\Psi$ be an $n$-qubit Haar-random state. Then, with probability at least $1-\exp(-cn)$ over the choice of $\Psi$, $\Delta(M_{\Psi}) \ge c'n^{-2}$ for some positive constants $c, c'$.
		\end{lemma}

		\begin{proof}
			For a state $\Psi$ and basis $\vec{P}$, define
			\begin{equation}
				f(\Psi,\vec{P}) = 
				\begin{cases}
					1, & \Delta(L_{\Psi}^{\vec{P}})  \ge c'n^{-2}, \\
					0, & \text{otherwise}
				\end{cases}
			\end{equation}
			with $c'$ from Corollary~\ref{col:spectral_gap_any_basis}.
			By Corollary~\ref{col:spectral_gap_any_basis}, there exists a constant $c$ such that for any basis $\vec{P}$,
			\begin{equation}
				\bE_{\Psi}[f(\Psi,\vec{P})] \ge 1-\exp(-cn).
			\end{equation}
			Consequently,
			\begin{equation}
				\bE_{\Psi} \bE_{\vec{P}}[1-f(\Psi,\vec{P})] 
				= \bE_{\vec{P}} \bE_{\Psi}[1-f(\Psi,\vec{P})] 
				\le \exp(-cn).
			\end{equation}
			Setting $c_2 = c/2$ and applying Markov’s inequality yields
			\begin{equation}
				\Pr_{\Psi}\!\left[\bE_{\vec{P}}[1-f(\Psi,\vec{P})] \ge \exp(-c_2 n)\right] 
				\le \exp(-(c-c_2)n) = \exp(-c_2 n).
			\end{equation}
			Hence, with probability at least $1-\exp(-c_2 n)$,
			\begin{equation}
				\bE_{\vec{P}} f(\Psi,\vec{P}) \ge 1-\exp(-c_2 n).
			\end{equation}
			It follows that
			\begin{equation}
				\Delta(M_{\Psi}) 
				\ge \bE_{\vec{P}} \Delta(L_{\Psi}^{\vec{P}}) 
				\ge \bE_{\vec{P}} f(\Psi,\vec{P}) \cdot c'n^{-2} 
				= \bigl(1-\exp(-c_2 n)\bigr)\Omega(n^{-2}) 
				= \Omega(n^{-2}),
			\end{equation}
			where the first inequality follows from the concavity of the spectral gap for non-negative observables that share the same eigenvector corresponding to the largest eigenvalue $1$. 
		\end{proof}
		
		Importantly, the observable $M_{\Psi}$ admits a multipartite Pauli measurement scheme that can be implemented using solely local randomness:
		\begin{itemize}
			\item Measurement Setting: Each party $l$ independently and uniformly selects a basis $P_l \in \{X, Y, Z\}$.
			\item Weighting Coefficient: For a given outcome vector $\vec{b}$ and setting $\vec{P}$, the coefficient is computed as:
			\begin{equation}\label{eq:omega-average-over-k}
				\omega(\vec{b}|\vec{P}) = \frac{1}{n}\sum_{k=0}^{n-1} \bra{\Psi_{\vec{b}, \vec{P}}^{(k)}} \left(3\frac{\mathbb{I}+(-1)^{b_k} \sigma_{P_k}}{2}-\mathbb{I}\right)\ket{\Psi_{\vec{b}, \vec{P}}^{(k)}}.
			\end{equation}
		\end{itemize}
		By the unbiasedness of local shadow tomography~\cite{Huang2024Certifying}, this procedure yields an unbiased estimator of $\tr(M_{\Psi} \rho)$. Note that we incorporate the average over $k$ into the classical weighting coefficient rather than the sampling process. Moreover, we maintain $\abs{\omega(\vec{b}|\vec{P})} \le W=\cO(1)$. Integrating this into our self-testing framework for local randomness, we obtain the following performance guarantee:
		\begin{theorem}[Self-testing states using local randomness]\label{thm:self-testing-state-local-randomness}
			Let $\varepsilon' > 0$ be the desired precision. For an $n$-partite pure target state $\ket{\Psi}$ with a spectral gap $\Delta \coloneqq \Delta(M_{\Psi}) > 0$, Protocol~\ref{prot:self-testing-state-local-randomness} is an $(\cO(\varepsilon'), \delta)$-sound and $(\cO(\Delta \varepsilon'^2), \cO(\frac{\Delta^2\varepsilon'^4}{n^2}), \delta)$-complete self-testing protocol for $\Psi$. The total sample complexity is given by:
			\begin{equation}
				N = \cO\left( \frac{n^5}{\Delta^{4} \varepsilon'^8} \log\left(\frac{n}{\delta}\right) \right).
			\end{equation}
		\end{theorem}
		
		\begin{protocol}{Self-testing states using local randomness}{prot:self-testing-state-local-randomness}
			\begin{enumerate}
				\item For a pure state $\Psi$ with $\Delta(M_{\Psi}) > 0$, execute Protocol~\ref{prot:di-witness-local-randomness} for the observable $M_{\Psi}$ using $N$ measurement rounds.
				
				\item  If the output of Protocol~\ref{prot:di-witness-local-randomness} is \texttt{FAILED} or the output is an estimate $\hat{\omega}$ satisfying $\hat{\omega} < 1 - \Delta\varepsilon'^2$, output \texttt{FAILED}.
				\item Else, output \texttt{CERTIFIED}.
			\end{enumerate}
		\end{protocol}
		
		\begin{proof}
			The proof follows the same logic as that of Theorem~\ref{thm:self-testing-state-shared-randomness}. The specific sample complexity $N$ is obtained by substituting $\varepsilon = \Delta\varepsilon'^2$ and $W = \cO(1)$ into the performance guarantee in Theorem~\ref{thm:fin-sample-DI-estimation-L-local-randomness}.
		\end{proof}
		
		For Haar-random states, we obtain the following performance guarantee by combining Theorem~\ref{thm:self-testing-state-local-randomness} with Lemma~\ref{lem:random-basis-enh}:
		\begin{corollary}[Sample-efficient self-testing of generic states using local randomness]\label{col:self-testing-Haar-random-states}
			Let $\varepsilon' > 0$. For an $n$-qubit Haar-random state $\ket{\Psi}$, with probability at least $1 - \exp(-\Omega(n))$ over the choice of $\Psi$, Protocol~\ref{prot:self-testing-state-local-randomness} is a $(\cO(\varepsilon'), \delta)$-sound and $(\cO(n^{-2} \varepsilon'^2), \cO(n^{-6}\varepsilon'^4), \delta)$-complete self-testing protocol. 
			The required sample complexity is $N = \cO\left(n^{13}\varepsilon'^{-8}\log(\frac{n}{\delta})\right)$. 
		\end{corollary}
		
		Finally, we provide two remarks for our sample complexity bounds. First, our self-testing definitions are formulated with respect to the trace norm, which is the most stringent distance measure. If the certification were defined with respect to state fidelity $F(\rho, \sigma)$, the sample complexity scaling would improve from $\varepsilon'^{-8}$ to $\varepsilon'^{-4}$. This improvement stems from bypassing the quadratic loss in the Fuchs–van de Graaf inequality (Eq.~\eqref{eq:fidelity-to-trace-norm}) used in our soundness proof. 
		Second, while our current analytical bound for the spectral gap of $M_{\Psi}$ is $\Omega(n^{-2})$, recent numerical evidence and conjectures in Ref.~\cite{du2025certifyinglocalizablequantumproperties} suggest that for generic states, the random-basis-enhanced observable may actually possess a constant spectral gap, $\Delta(M_{\Psi}) = \Omega(1)$. If this conjecture holds, the factors of $n$ resulting from $\Delta$ could be removed, significantly enhancing both the robustness and the sample efficiency of self-testing for generic states.

		\section{Non-i.i.d.\ scenario}\label{sec:non-iid}
		
		In this section, we consider the non-i.i.d.\ scenario of self-testing. Previously, in the i.i.d.\ scenario, we assumed that the state and measurement operators are identical for every sample. In any practical experiment, this is not a realistic constraint. In this section, we allow measurements to depend on previous inputs and outputs, and to act on a quantum state that depends on previous measurements.
		In Lemma~\ref{lem:noniid-average}, we will show that the previously introduced protocols still work in the non-i.i.d.\ scenario with identical asymptotic sample complexity.
		
		In the non-i.i.d.\ scenario, there is no well-defined state $\rho$ of the physical experiment. Rather, this state may differ in each round of the protocol. Therefore, we first need to clarify how a statement such as the claim of Theorem~\ref{thm:DI-estimation-L} can be adopted to the non-i.i.d.\ case.
		One challenge, discussed in \cite{2022sampleefficientdiverification}, is that the shared state $\rho^{[N]}$ over $N$ rounds of measurement may be a classical mixture such as $\rho^{[N]} = \frac{1}{2}(\sigma^{\otimes N} + (\tau\otimes\Phi^+)^{\otimes N})$. Here, $\sigma$ can be any \enquote{garbage state} for which the protocol should fail while $\tau\otimes\Phi^+$ corresponds to the reference experiment. It appears obvious that this state should be rejected by the protocol as it is \enquote{garbage} one out of two times. Yet, any complete (see Sec.~\ref{sec:completeness}) protocol with independent measurements over the individual rounds will accept this state with around $50\%$ probability.
		
		We resolve this issue by using an idea from \cite{2022sampleefficientdiverification} to define extractability with respect to the \emph{accessed states}, that is, the states $\rho_{|m_{<i}}$ in a given round $i$ of the protocol, conditional on all previous inputs and outputs $m_{<i}$. 
		In \cite{2022sampleefficientdiverification}, the \emph{average extractability} is defined by the average achievable fidelity between a target state and the accessed states $\rho_{|m_{<i}},\ i=1,\cdots,N$ after applying a local extraction channel. 
		
		We consider an equivalent definition of average extractability: extractability with respect to an average state, which has a clear operational meaning. First, we need to introduce some additional notation:
		We write the \emph{accessed measurement} in round $i$ as $M_{|m_{<i}}$ (where the POVM elements $\left(\bigotimes_{l=0}^{n-1}M_{a_l|x_l}^{(l)}\otimes N_{b_l|y_l}^{(l)}\right)_{|m_{<i}}$ of the accessed measurement $M_{|m_{<i}}$ in round $i$ may also depend on outcomes in the previous rounds).
		Define the \emph{average state} $\bar\rho$ as the labeled, classical mixture of the accessed states in individual rounds $i$, i.e.\ $\bar\rho\coloneqq\frac1N\sum_{i=1}^N\ketbra{i}^{\otimes n}\otimes\rho_{|m_{<i}}$.
		We then call the following experiment an \emph{average experiment}: On the average state $\bar\rho$, each party first measures the round label $i$ and then proceeds to apply operations (which may depend on the round label $i$) to the remaining state $\rho_{|m_{<i}}$.
		
		Because the average state is the uniform mixture of $\ketbra{i}^{\otimes n}\otimes\rho_{|m_{<i}}$, an average experiment is operationally equivalent to randomly choosing a round label $i$, stopping before round $i$ and then performing an operation on the current state $\rho_{|m_{<i}}$. An example of such an operation would be measuring $M_{|m_{<i}}$ in the random round to obtain a result $(\vec{a},\vec{b})$. This average experiment corresponds to the (local) measurements 
		\begin{equation}\label{eq:measurement-avg-exp}
			\sum_{i=1}^N \ketbra{i}^{\otimes n}\otimes M_{|m_{<i}}
		\end{equation}
		on the average state $\bar{\rho}$.
		
		An extractor $\Lambda$ on the average state corresponds to a sequence of extractors $\Lambda^{(i)}=\bigotimes_{l=0}^{n-1}\Lambda_l^{(i)}$, $\Lambda_l^{(i)}:\rho\mapsto\Lambda_l(\ketbra{i}\otimes\rho)$, applied on each round of the physical experiment. Conversely, a sequence of extractors $\Lambda^{(i)}=\bigotimes_{l=0}^{n-1}\Lambda^{(i)}_l$ applied in the $N$ rounds of the physical experiment can be combined into one extractor
		\begin{equation}\label{eq:roundwise-extractors-to-average-extractor}
			\Lambda=\bigotimes_{l=0}^{n-1}\Lambda_l;\qquad\Lambda_l: \rho\mapsto\sum_{i=1}^N(\bra{i}\otimes\mathbb{I})(\mathcal{I}\otimes\Lambda_l^{(i)})(\rho)(\ket{i}\otimes\mathbb{I}).
		\end{equation}
		These two types of extractors are equivalent in the following sense:
		\begin{equation}
			\frac1N\sum_{i=1}^N\Lambda^{(i)}(\rho_{|m<i})=\Lambda\left(\sum_{i=1}^N\ketbra{i}^{\otimes n}\otimes\rho_{|m<i}\right)=\Lambda(\bar\rho).
		\end{equation}
		That is, the state extracted by randomly picking a round and then extracting with respect to that round can be implemented as extracting from the average state, and vice versa. This explains why the extractability notion with respect to the average state is equivalent to the average extractability in \cite{2022sampleefficientdiverification}: the average fidelity between any state $\ketbra{\psi}$ and the state $\Lambda_i(\rho_{|m<i})$ extracted from some round, is equal to the fidelity between $\ketbra{\psi}$ and the state $\Lambda(\bar\rho)$ extracted from the average state.
		
		We are now in a position to define average extractability: We define the \emph{average extractability} as the extractability with respect to the average state $\bar\rho$. Here, by extractability, we mean that there exists an extraction channel $\Lambda$ that is the product of local channels, such that $\Lambda(\bar\rho)$ extracts a target state. From our preceding discussion, it becomes clear that extractability with respect to the average state $\bar\rho$ is equivalent to average extractability in an average experiment and vice versa.

		We define the values $\vI^{[i]},\vII^{[i]},\vIII^{[i]}$ as the values of the estimators in Box~\ref{box:ProtocolEstimators} (we suppress the dependence on $l$ and $k$ for simplicity of notation) in round $i$ with respect to the state $\rho_{|m_{<i}}$ and with measurements $M_{|m_{<i}}$. Let $\vI, \vII, \vIII$ be the uniform averages of $\vI^{[i]},\vII^{[i]},\vIII^{[i]}$ over $i$. By the definition of the average experiment, $\vI, \vII, \vIII$ are exactly the values of the estimators in Box~\ref{box:ProtocolEstimators} applied in an average experiment with the states $\rho_{|m_{<i}}$ and measurements $M_{|m_{<i}}$ (see discussion above~\eqref{eq:measurement-avg-exp}).
		
		Our main result of this section is Theorem~\ref{thm:noniid}, which is a direct consequence of Lemma~\ref{lem:noniid-average}. Lemma~\ref{lem:noniid-average} proves that the experimental estimates $\hvI, \hvII, \hvIII$ are close to the expectation values $\vI, \vII, \vIII$ of the average experiment. Using Theorem~\ref{thm:noniid} together with the operator $M_{\Psi}$ from \eqref{eq:modified-Huang-protocol-def}, we verify that if our protocol outputs \texttt{CERTIFIED}, then we can (approximately) extract the target state from the average state.

		\begin{lemma}[Non-i.i.d.\ estimators]\label{lem:noniid-average}
			In the non-i.i.d.\ case, the estimators $\hat v \in \{\hvI, \hvII, \hvIII\}$ in Protocol~\ref{prot:di-witness-local-randomness} (Protocol~\ref{prot:di-witness-shared-randomness}) are close to the respective expectation values $v\in \{\vI,\vII,\vIII\}$ with respect to the average experiment. Specifically, they satisfy
			\begin{equation}\label{eq:claim-hatv-close-v-noniid}
				|\hat v - v| = \cO(\varepsilon)
			\end{equation}
			with failure probability at most $\delta$ and same sample complexity $N(W,\varepsilon,\delta)$ as in Theorem~\ref{thm:fin-sample-DI-estimation-L-local-randomness} (Theorem~\ref{thm:fin-sample-DI-estimation-L-shared-randomness}).
		\end{lemma}

		\begin{theorem}[Non-i.i.d.\ finite sample version of Theorem~\ref{thm:DI-estimation-L}]\label{thm:noniid}
			In the non-i.i.d.\ case, with probability at least $1-\delta$, Protocol~\ref{prot:di-witness-local-randomness} (Protocol~\ref{prot:di-witness-shared-randomness}) with same asymptotic sample complexity $N$ as Theorem~\ref{thm:fin-sample-DI-estimation-L-local-randomness} (Theorem~\ref{thm:fin-sample-DI-estimation-L-shared-randomness}) either returns \texttt{FAILED}, or for the average state $\bar\rho$, there exist $p \in [0,1]$, a state $\tau = p \tau_+ + (1-p)\tau_-$, and a product of local channels $\Gamma$ such that:
			\begin{equation}
				D\left(p \tau_+ \otimes \ketbra{0}^{\otimes n} + (1-p)\tau_-^* \otimes \ketbra{1}^{\otimes n}, \Gamma(\bar\rho)\right) \le \cO\left(\frac{\varepsilon}{W}\right),
			\end{equation}
			and the value $\hvIII$ returned by the protocol satisfies
			\begin{equation}\label{eq:claim-hvIII-vIII-non-iid}
				|\hvIII - \tr(L\tau)| = \cO(\varepsilon).
			\end{equation}
		\end{theorem}
		
		For the interpretation of our results in this section, it may be useful to note that these results can be extended beyond the operational meaning as an average statement to statements about a large fraction of rounds. Concretely, by utilizing the non-negativity of the estimators $\vI^{[i]},\vII^{[i]}$ and the Markov inequality, small values of $\vI, \vII$ (which follow from small observed estimates $\hvI, \hvII$ by Lemma~\ref{lem:noniid-average}) imply a small value for $\vI^{[i]},\vII^{[i]}$ in most rounds $i$. Theorem~\ref{thm:DI-estimation-L}, \eqref{eq:thm-Pauli-value} then implies that for most rounds, the single round expectation $\vIII^{[i]}$ provides a good estimate for the expectation value of $L$ on the state present in that round. We will not formalize this further here and now move on to the proofs of Lemma~\ref{lem:noniid-average} and Theorem~\ref{thm:noniid}.

		To prove Lemma~\ref{lem:noniid-average}, we need the following lemma, which is a mixed version of Chernoff bound and Azuma's inequality (martingale version of Hoeffding's inequality).
		\begin{lemma}[A mixed version of Chernoff bound and Azuma's inequality]\label{lem:modified-chernoff}
			Let $\epsilon>0,0<p\le1$. Let $m_1,\cdots,m_N$ be random variables. Let $\hat\delta_i,\hat v_i$ be functions of $m_1,\cdots,m_i$; $|\hat v_i|\le 1$ and $\hat\delta_i\in\{0,1\}$. Assume
			\begin{align}\label{eq:lem-chernoff-cond}
				\bE[\hat\delta_i\mid m_{<i}]=p;\quad \bE[\hat v_i\mid m_{<i},\hat\delta_i=1]=v_i,
			\end{align}
			where $v_i$ is a function of $m_{<i}$. Then,
			\begin{align}\label{eq:lem-chernoff-result}
				\Pr[\frac{\sum_{i=1}^N\hat\delta_i\hat v_i}{pN}-\frac{\sum_{i=1}^Nv_i}N\ge\epsilon]\le\exp(-\frac{\epsilon^2}{2+\epsilon}pN).
			\end{align}
		\end{lemma}
		\begin{proof}
			We slightly modify the proof of the Chernoff bound. Let $\lambda>0$ be a to-be-determined constant. By Markov's inequality and the law of iterated expectations,
			\begin{align}\label{eq:lem-chernoff-pf}\begin{split}
					\Pr[\frac{\sum_{i=1}^N\hat\delta_i\hat v_i}{pN}-\frac{\sum_{i=1}^Nv_i}N\ge\epsilon]&=\Pr[\exp(\lambda\sum_{i=1}^N(\hat\delta_i\hat v_i-pv_i-p\epsilon))\ge1] \\
					&\le\bE\left[\exp(\lambda\sum_{i=1}^N(\hat\delta_i\hat v_i-pv_i-p\epsilon))\right]=\bE\left[\prod_{i=1}^N\exp(\lambda(\hat\delta_i\hat v_i-pv_i-p\epsilon))\right]\\
					&=\bE\left[\exp(\lambda(\hat\delta_1\hat v_1-pv_1-p\epsilon))\right.
					\\&\quad\cdot\bE\left[\exp(\lambda\hat\delta_2\hat v_2-pv_2-p\epsilon)\right.
					\\&\quad\quad\left.\left.\cdot\bE\left[\exp(\lambda(\hat\delta_3\hat v_3-pv_3-p\epsilon))\cdot\cdots\mid m_1,m_2\right]\mid m_1\right]\right].
			\end{split}\end{align}
			
			Now we bound $\bE\left[\exp(\lambda(\hat\delta_i\hat v_i-pv_i-p\epsilon))\mid m_{<i}\right]$. By assumption, $v_i$ is determined given $m_{<i}$. By \eqref{eq:lem-chernoff-cond},
			\begin{align}\begin{split}
					&\bE\left[\exp(\lambda(\hat\delta_i\hat v_i-pv_i-p\epsilon))\mid m_{<i}\right]\\=&\exp(-\lambda p(v_i+\epsilon))\left(\Pr[\hat\delta_i=1\mid m_{<i}]\bE\left[\exp(\lambda\hat v_i)\mid m_{<i},\hat\delta_i=1\right]+\Pr[\hat\delta_i=0\mid m_{<i}]\right)
					\\=&\exp(-\lambda p(v_i+\epsilon))\left(p\bE\left[\exp(\lambda\hat v_i)\mid m_{<i},\hat\delta_i=1\right]+(1-p)\right)
					\\\le&\exp(-\lambda p(v_i+\epsilon))\exp(p\left(\bE\left[\exp(\lambda\hat v_i)\mid m_{<i},\hat\delta_i=1\right]-1\right))
					\\=&\exp(p\bE\left[\exp(\lambda\hat v_i)-1-\lambda(\hat v_i+\epsilon)\mid m_{<i},\hat\delta_i=1\right]).
			\end{split}\end{align}
			Here we used $1+x \leq \exp(x)$ for any real $x$.
			Since $\exp(\lambda\hat v_i)-1-\lambda(\hat v_i+\epsilon)$ is convex with respect to $\hat v_i$, and $v_i\in[-1,1]$,
			\begin{align}\begin{split}
					&\bE\left[\exp(\lambda(\hat\delta_i\hat v_i-pv_i-p\epsilon))\mid m_{<i}\right]
					\\\le&\exp(p\bE\left[\exp(\lambda\hat v_i)-1-\lambda(\hat v_i+\epsilon)\mid m_{<i},\hat\delta_i=1\right])
					\\\le&\exp(p\max(\exp(\lambda)-1-\lambda(1+\epsilon),\exp(-\lambda)-1-\lambda(-1+\epsilon))).
			\end{split}\end{align}
			Set $\lambda$ to be $\ln(1+\epsilon)$. By utilizing the fact that $\exp(x)-x>\exp(-x)+x$ for $x>0$, and $\ln(\frac{1+x}{1-x})>2x$ for $0<x<1$ (which implies $\log(x)>2\frac{x-1}{x+1}$ for $x>1$), we know
			\begin{align}
				\exp(-\lambda)-1-\lambda(-1+\epsilon)<\exp(\lambda)-1-\lambda(1+\epsilon)=\epsilon-(1+\epsilon)\ln(1+\epsilon)<\epsilon-2(1+\epsilon)\frac{\epsilon}{\epsilon+2}=-\frac{\epsilon^2}{2+\epsilon}.
			\end{align}
			Thus,
			\begin{align}
				\bE\left[\exp(\lambda(\hat\delta_i\hat v_i-pv_i-p\epsilon))\mid m_{<i}\right]
				\le\exp(p\max(\exp(\lambda)-1-\lambda(1+\epsilon),\exp(-\lambda)-1-\lambda(-1+\epsilon)))\le\exp(-\frac{p\epsilon^2}{2+\epsilon}).
			\end{align}
			Together with \eqref{eq:lem-chernoff-pf}, we get \eqref{eq:lem-chernoff-result}.
		\end{proof}

		\begin{proof}[Proof of Lemma~\ref{lem:noniid-average}]
			Consider any statistics of one round $i$. Let $p_{\mathrm{I}}, p_{\mathrm{II}}, p_{\mathrm{III}}$ denote the probability that these statistics contribute to estimator $\vI, \vII, \vIII$. The following analysis holds for any estimator, for simplicity we denote $p\in \{p_{\mathrm{I}}, p_{\mathrm{II}}, p_{\mathrm{III}}\}$ and $v_i \in \{\vI^{[i]}(l,k),\vII^{[i]}(l,l+1),\vIII^{[i]}\}$.
			Let $\hat\delta_i$ be a binary variable indicating whether $v_i$ is estimated in the $i$-th round of the actual protocol, and let $\hat v_i$ denote the (single sample) estimate for $v_i$ from that round. Assume that $v\in \{\vI,\vII,\vIII\}$ is bounded by $W$. The estimate of $v$ is
			\begin{align}
				\hat v=\frac{\sum_{i=1}^N\hat\delta_i\hat v_i}{\sum_{i=1}^N\hat\delta_i}.
			\end{align}
			On the other hand $v$ is equal to the average value of $v_i$, i.e.,
			\begin{align}
				v=\frac1N\sum_{i=1}^Nv_i.
			\end{align}
			We now show that the probability of ``$v$ and $\hat v$ are close'' has the same bound as in the i.i.d.\ case.
			
			We define another estimate
			\begin{align}
				\hat v'=\frac{\sum_{i=1}^N\hat\delta_i\hat v_i}{pN}.
			\end{align}
			We will prove below that $v$ and $\hat v'$ are close with high probability.
			Subsequently, we apply Lemma~\ref{lem:modified-chernoff} to $\hat v'$ in order to show that $\hat v'$ approximates $v$ with high probability: Setting $\hat v_i,v_i$ in Lemma~\ref{lem:modified-chernoff} to be $ \pm \hat v_i/W,\pm v_i/W$, we get
			\begin{align}
				\Pr[\hat v'-v>\epsilon/2]<\exp(-\frac{(\epsilon/2W)^2}{2+\epsilon/2W}pN);\qquad\Pr[\hat v'-v<-\epsilon/2]<\exp(-\frac{(\epsilon/2W)^2}{2+\epsilon/2W}pN).
			\end{align}
			Hence, with probability at least $1-2\exp(-pN(\epsilon/2W)^2/(2+\epsilon/2W))$, we have $\Pr[|\hat v'-v|\le\epsilon/2]$.
			
			Now, we prove that $\hat v$ and $\hat v'$ are close with high probability. Since $\{\hat\delta_i\}_i$ are independent binary random variables with expectation value $p$, we can apply a Chernoff bound for Bernoulli random variables to get
			\begin{align}
				\Pr\left[\left|\frac{\sum_{i=1}^N\hat\delta_i}{pN}-1\right|<\epsilon/2W\right]\ge1-2\exp\left(-\frac{pN(\epsilon/2W)^2}{2+\epsilon/2W}\right).
			\end{align}
			Hence, with probability at least $1-2\exp(-pN(\epsilon/2W)^2/(2+\epsilon/2W))$,
			\begin{align}
				|\hat v'-\hat v|=\left|\left(\frac{\sum_{i=1}^N\hat\delta_i}{pN}-1\right)\cdot\hat v\right|<\frac\varepsilon{2W}\cdot W=\epsilon/2.
			\end{align}
			Therefore, by union bound, with probability at least $1-4\exp(-pN(\epsilon/2W)^2/(2+\epsilon/2W))$, $|\hat v-v|\le|\hat v'-v|+|\hat v'-\hat v|<\epsilon/2+\epsilon/2=\epsilon$.
			In order to lower bound the probability to $1-\delta$, we get $N = \cO(\frac{1}{p}\frac{W^2}{\epsilon^2}\log(\frac{1}{\delta}))$, a bound similar to the i.i.d.\ scenario (see Appendix~\ref{sec:sample_comp}). Hence, the sample complexity analysis can proceed analogously to the i.i.d.\ case. This shows~\eqref{eq:claim-hatv-close-v-noniid}.
		\end{proof}

		\begin{proof}[Proof of Theorem~\ref{thm:noniid}]
			The result follows from~\eqref{eq:claim-hatv-close-v-noniid}, Theorem~\ref{thm:DI-estimation-L} applied to the average experiments that have expectation values $\vI,\vII$ and $\vIII$, and an identical derivation to Theorem~\ref{thm:fin-sample-DI-estimation-L-local-randomness} (Theorem~\ref{thm:fin-sample-DI-estimation-L-shared-randomness}).
			
		\end{proof}
	\end{document}